\documentclass[a4paper,USenglish,cleveref, autoref, thm-restate]{lipics-v2021}

\hideLIPIcs  

\title{Wheeler Bisimulations} 

\author{Nicola {Cotumaccio}}{University of Helsinki, Finland}{nicola.cotumaccio@helsinki.fi}{https://orcid.org/0000-0002-1402-5298}{}

\authorrunning{N. Cotumaccio} 

\Copyright{Jane Open Access and Joan R. Public} 



\category{} 

\relatedversion{} 

\funding{Funded by the Helsinki Institute for Information Technology (HIIT).}

\acknowledgements{I am grateful to Luca Aceto for his insightful comments and valuable feedback. I also thank the anonymous reviewers for carefully reviewing the paper, and in particular for suggesting a simpler proof of Lemma~\ref{lem:convexunion} and pointing out a mistake in an earlier proof of Theorem~\ref{theor:lowerboundquotienting}.}

\nolinenumbers 

\usepackage{graphicx}%
\usepackage{multirow}%
\usepackage{amsmath,amssymb,amsfonts}%
\usepackage{amsthm}%
\usepackage{mathrsfs}%
\usepackage[title]{appendix}%
\usepackage{xcolor}%
\usepackage{textcomp}%
\usepackage{manyfoot}%
\usepackage{booktabs}%
\usepackage{algorithm}%
\usepackage{algorithmicx}%
\usepackage{algpseudocode}%
\usepackage{listings}%

\usepackage{tikz}
\usetikzlibrary{arrows,automata,positioning}

\usepackage{amssymb}

\DeclareMathOperator{\Pref}{Pref}

\usepackage{etoolbox}

\usepackage{subcaption}

\EventEditors{Ana Sokolova and Patrick Totzke}
\EventNoEds{2}
\EventLongTitle{37th International Conference on Concurrency Theory (CONCUR 2026)}
\EventShortTitle{CONCUR 2026}
\EventAcronym{CONCUR}
\EventYear{2026}
\EventDate{September 1--4, 2026}
\EventLocation{Liverpool, UK}
\EventLogo{}
\SeriesVolume{391}
\ArticleNo{6}

\begin{document}

\maketitle


\begin{abstract}
Over the years, bisimulations have emerged as a pervasive paradigm, finding applications in numerous areas, including concurrency theory, model checking, automata theory, logic, programming languages and category theory. In this paper, we establish a connection between bisimulations and data compression. More precisely, we study the relationship between bisimulations and Wheeler automata (Alanko et al., SODA 2020), a class of automata that has received considerable attention in recent years. The standard notion of bisimulation is not appropriate, so we introduce \emph{Wheeler bisimulations}, that is, bisimulations that respect the convex structure of the considered Wheeler automata. We show that Wheeler bisimilarity induces a \emph{unique} minimal Wheeler NFA (analogously to standard bisimulations). In particular, in the deterministic case, we retrieve the minimal Wheeler deterministic automaton of a given language. We also show that the minimal Wheeler NFA induced by Wheeler bisimulations can be built in linear time. This is in contrast with standard bisimulations, for which the corresponding minimal NFA can be built in $ O(m \log n) $ time (where $ m $ is the number of edges and $ n $ is the number of states) by adapting Paige-Tarjan partition refinement algorithm. Compared to previous state-reduction techniques, our bisimulation-induced construction is the first for which (i)~we obtain a \emph{canonical} Wheeler NFA and (ii)~the resulting Wheeler NFA can be built in linear time.
\end{abstract}

\section{Introduction}\label{sec:introduction}

\emph{Bisimulations} were first introduced by Milner and Park in the early 1980s to model processes that cannot be distinguished by an external agent~\cite{milner1980calculus, park1981concurrency}. Over the years, bisimulations have emerged as a pervasive paradigm, finding applications in numerous areas, including concurrency theory, model checking, automata theory, logic, programming languages and category theory. We refer the reader to standard books on the topic~\cite{baier2008principles, sangiorgi2011introduction, kozen2007automata, aceto2007reactive}. Here we mainly focus on bisimulations between \emph{automata} (see in particular~\cite{kozen2007automata}), but a similar theory of bisimulations can be developed for other labeled systems~\cite{baier2008principles} (see also~\cite{van1987petri, best1991concurrent} for Petri nets and~\cite{degano1993universal} for Kripke structures).

Let $ \mathcal{A} = (Q, E, s, F) $ be an NFA, with $ |Q| = n $ and $ |E| = m $. Then, up to isomorphism, there exists a unique minimal NFA bisimilar to $ \mathcal{A} $~\cite{kozen2007automata}. The algorithmic problem of computing the minimal bisimilar NFA has a long history. In the deterministic setting, it corresponds to the problem of minimizing a DFA, which can be solved using, e.g., Moore algorithm~\cite{moore1956gedanken} or Hopcroft algorithm~\cite{hopcroft1971n}. In the general setting, the problem can be solved using, e.g., Kanellakis-Smolka algorithm~\cite{kanellakis1983ccs} or Paige-Tarjan algorithm~\cite{paige1987three}. The most efficient algorithms are Valmari and Lehtinen's variants of Hopcroft algorithm and Paige-Tarjan algorithm, both running in $ O(m \log n) $ time (see~\cite{valmari2008efficient, valmari2010simple}). The bound $ O(m \log n) $ is achieved under the standard assumption that the edge labels can be sorted in linear time~\cite[Section 4.3]{valmari2010simple}.

In this paper, we establish a connection between bisimulations and data compression. More precisely, we study the relationship between bisimulations and \emph{Wheeler automata}~\cite{alanko2020regular, gagie2017wheeler}, a class of automata that has received considerable attention in recent years (we refer the reader to~\cite{cotumaccio2025wheeler} for a survey). A Wheeler automaton is a (conventional) NFA $ \mathcal{A} $ endowed with a total order $ \le $ on the set of all states. Intuitively, if for some state $ u $ and $ v $ we have $ u < v $, then the strings that reach state $ u $ must be co-lexicographically smaller than the strings that reach state $ v $, up to intersections (see Section~\ref{sec:wheelerautomata} for details).

The idea of ordering the states of an automaton dates back at least to the seventies, when Shyr and Thierrin introduced \emph{ordered automata}~\cite{shyr1974ordered}. Wheeler automata are perhaps the most elegant and successful implementation of this paradigm:

\begin{itemize}
    \item A Wheeler automaton can be stored using only $ (e \log \sigma) (1 + o(1)) + O(e) $ bits, where $ e $ is the number of edges and $ \sigma $ is the size of the alphabet~\cite{gagie2017wheeler}. At the very least we need $ e \log \sigma $ bits to store the edge labels (we need $ \log \sigma $ bits per edge), so intuitively the topological complexity of a Wheeler automaton can be stored succinctly.
    \item We can decide whether a string of length $ m $ is accepted by a Wheeler automaton in time $ O(m \log \log \sigma) $~\cite{gagie2017wheeler}. For general automata, the time complexity is $ \Omega(m e^{1 - \epsilon}) $, where $ \epsilon > 0 $ is an arbitrarily small constant (under the Orthogonal Vector hypothesis)~\cite{equi2023complexity}. Remarkably, the algorithm running in $ O(m \log \log \sigma) $ time does not need to access the explicit representation of a Wheeler automaton, but only its compressed representation (consisting of $ (e \log \sigma) (1 + o(1)) + O(e) $ bits). 
\end{itemize}

In other words, Wheeler automata extend the most important data structures for solving pattern matching queries on compressed texts (based on the \emph{suffix array}~\cite{manber1993suffix}, the \emph{Burrows-Wheeler transform}~\cite{burrows1994} and the \emph{FM-index}~\cite{ferraginajacm2005}) from strings to automata. In particular, Wheeler automata are closely related to Bruijn graphs, which are used in bioinformatics to perform Eulerian sequence assembly (see~\cite{bowe2012succinct}).

The regular languages recognized by Wheeler automata are called \emph{Wheeler languages}. The class of Wheeler languages is one of the richest subclasses of regular languages:
\begin{itemize}
    \item If a regular language is recognized by some Wheeler non-deterministic automaton, it is also recognized by some Wheeler deterministic automaton~\cite{alanko2021wheeler}.
    \item Every Wheeler language admits a unique minimal Wheeler deterministic automaton (which, in general, is larger than the standard minimal DFA)~\cite{alanko2021wheeler}, and a Wheeler deterministic automaton can be minimized efficiently~\cite{alanko2022linear}.
    \item Every Wheeler language admits an algebraic characterization based on the \emph{Myhill-Nerode theorem}~\cite{alanko2020regular}.
    \item It is possible to efficiently detect if a regular language is Wheeler by looking at the cycle structure of its (standard) minimal automaton~\cite{becker2023optimal}.
\end{itemize}

Among other topics, previous work has investigated string-labeled Wheeler NFAs~\cite{cotumaccio2024myhill, cotumaccio2026generalized, cotumaccio2025fast}, the state complexity of Wheeler languages~\cite{d2023ordering}, their relationship with locally testable languages~\cite{becker2025universally} and circular strings~\cite{cotumaccio2025sorting, cotumaccio2025improved}, and more sophisticated applications in data compression~\cite{conte2023computing, cotumaccio2023space, alanko2024computing, cotumaccio2023prefix}.

\subsection{Our Contribution}

In the deterministic setting, both regular languages and Wheeler languages admit a \emph{unique} minimal automaton up to isomorphism. Note that, in general, the minimal Wheeler automaton can be larger than the standard minimal automaton~\cite{alanko2021wheeler} (we remark that different trade-offs between the size and the compressibility of a DFA recognizing a given language can be obtained by extending the ideas behind Wheeler automata to arbitrary automata~\cite{cotumaccio2023co, cotumaccio2021indexing}). In the non-deterministic setting, a regular language can admit multiple minimal automata, and the problem of building such a minimal automaton is notoriously PSPACE-complete~\cite{stockmeyer1973word}. By introducing additional constraints, bisimulations make it possible to recover the uniqueness of a minimal automaton in the non-deterministic setting: two bisimilar automata recognize the same language, but two automata that recognize the same language need not be bisimilar~\cite{kozen2007automata}.

In this paper, we study the minimality problem for Wheeler languages in the non-deterministic setting. We start our investigation by observing that, in general, Wheeler languages are not recognized by a unique minimal Wheeler NFA (Example~\ref{ex:nonisomorphicexample}). Then, we study the relationship between Wheeler languages and nondeterminism through the lens of bisimulations. The standard notion of bisimulation is not appropriate (see Example~\ref{ex:bisimulationisnotenough}), so we introduce \emph{Wheeler bisimulations}, that is, bisimulations that respect the convex structure of the considered Wheeler automata (Definition~\ref{def:wheelerbisimulation}). We show that Wheeler bisimilarity induces a \emph{unique} minimal Wheeler NFA (analogously to standard bisimulations), see Theorem~\ref{theorem:minimalbisimilar} and Corollary~\ref{cor:quotientisstateminimal}. In particular, in the deterministic case, we retrieve the minimal Wheeler deterministic automaton of a given language (Theorem~\ref{theor:determinismmaintheor} and Corollary~\ref{cor:determinismmaincor}).

We also study Wheeler bisimulations from an algorithmic perspective. We show that the minimal Wheeler NFA induced by Wheeler bisimulations can be computed in linear time (Theorem~\ref{theor:lineartimeminimization}). Consequently, in the Wheeler setting we can do better than in the conventional setting (where the best algorithm runs in $ O(m \log n) $ time). We remark that, in general, the minimal Wheeler NFA induced by Wheeler bisimulations can be distinct from the minimal NFA induced by conventional bisimulations. We achieve linear time under the usual assumption that the edge labels can be sorted in linear time, and we complete the picture with a lower bound showing that, in the comparison model, there exists no linear-time algorithm for the same problem (Theorem~\ref{theor:lowerboundquotienting}). In linear time we can also decide whether there exists a Wheeler bisimulation between two Wheeler NFAs (Corollary~\ref{cor:checkingifbisimlar}).

We can informally summarize our main results as follows.

\begin{theorem}
    Let $ (\mathcal{A}, \le) $ be a Wheeler NFA. Then, up to isomorphism, there exists a unique minimal Wheeler NFA $ (\mathcal{A}_*, \le_*) $ for which there exists a Wheeler bisimulation from $ (\mathcal{A}, \le) $ to $ (\mathcal{A}_*, \le_*) $. Moreover:
    \begin{itemize}
        \item $ (\mathcal{A}_*, \le_*) $ can be built from $ (\mathcal{A}, \le) $ in linear time.
        \item If $ (\mathcal{A}, \le) $ is a Wheeler DFA, then $ (\mathcal{A}_*, \le_*) $ is the minimal Wheeler DFA equivalent to $ (\mathcal{A}, \le) $.
    \end{itemize}
    In linear time we can also decide whether there exists a Wheeler bisimulation between two given Wheeler NFAs.
\end{theorem}

Consider a Wheeler language $ \mathcal{L} $. We know that every Wheeler NFA recognizing $ \mathcal{L} $ admits a compressed encoding. To obtain an even more compressed representation of $ \mathcal{L} $, we should select a Wheeler NFA with as few states as possible. Several papers have addressed this challenge, typically by devising procedures to collapse some states in a given Wheeler NFA recognizing $ \mathcal{L} $~\cite{becker2023sorting, cotumaccio2022graphs, becker2024indexing, becker2025encoding}. Remarkably, our bisimulation-induced construction is the first for which (i)~we obtain a \emph{canonical} Wheeler NFA and (ii)~the resulting Wheeler NFA can be built in linear time. From a broader perspective, bisimulations confirm that Wheeler languages form one of the richest and most robust subclasses of regular languages.

The paper is organized as follows. In Section~\ref{sec:preliminaries}, we discuss notation and preliminary results. In Section~\ref{sec:wheelerautomata} we introduce Wheeler automata, and in Section~\ref{sec:mainresults} we define Wheeler bisimulations. In Section~\ref{sec:minimality} we prove that Wheeler bisimilarity induces a unique minimal Wheeler NFA, and in Section~\ref{sec:algorithms} we study the algorithmic aspects of Wheeler bisimulations. In Section~\ref{sec:conclusions} we present our conclusions and we discuss future work. Due to space constraints, some proofs and some auxiliary results are deferred to the appendix.

\section{Preliminaries}\label{sec:preliminaries}

\subsection{Relations} Let $ Q $ and $ Q' $ be sets, and let $ R \subseteq Q \times Q' $ be a relation from $ Q $ to $ Q' $. For every $ u \in Q $ and for every $ u' \in Q' $, we write $ u \;R \; u' $ if and only if $ (u, u') \in R $. For every $ u \in Q $, let $ R(u) = \{u' \in Q \;|\; u \; R \; u' \} $, and for every $ U \subseteq Q $, let $ R(U) = \bigcup_{u \in U} R(u) $. If $ R $ is a function from $ Q $ to $ Q' $ and $ u \in Q $, we identify $ R(u) $ and its unique element.

If $ R \subseteq Q \times Q' $ is a relation, let $ R^{-1} \subseteq Q' \times Q $ be the inverse relation. If $ R \subseteq Q \times Q' $ and $ R' \subseteq Q' \times Q'' $ are relations, let $ R' \circ R \subseteq Q \times Q'' $ be the composition of the two relations. Let $ id_Q \subseteq Q \times Q $ be the identity function from $ Q $ to $ Q $.

An equivalence relation on $ Q $ is typically denoted by $ \equiv $ or $ \sim $. For every $ u \in Q $, let $ [u]_\equiv = \{v \in Q \;|\; u \equiv v \} $. Then, $ \{[u]_\equiv \;|\; u \in Q \} $ is a partition of $ Q $, denoted by $ \mathcal{P}_\equiv $. Every element of $ \mathcal{P}_\equiv $ is an $ \equiv $-equivalence class. A total order on $ Q $ is typically denoted by $ \le $. For every $ u, v \in Q $, we write $ u < v $ if $ (u \le v) \land (u \not = v) $.

\subsection{Convexity}

To introduce Wheeler bisimulations (Definition~\ref{def:wheelerbisimulation}), we will need the notion of \emph{convexity}. Let $ Q $ be a set, and let $ \le $ be a total order on $ Q $. We say that a subseteq $ C \subseteq Q $ is \emph{$ \le $-convex} if for every $ u, v, w \in Q $ such that $ u \le v \le w $ and $ u, w \in C $ we have $ v \in C $. In particular, $ \emptyset $ and $ Q $ are $ \le $-convex, and $ \{u \} $ is $ \le $-convex for every $ u \in Q $.

It is easy to check that, if $ C_1 $ and $ C_2 $ are $ \le $-convex, then $ C_1 \cap C_2 $ is $ \le $-convex. However, in general, $ C_1 \cup C_2 $ need not be $ \le $-convex. For example, if $ Q = \{1, 2, 3, 4, 5 \} $, $ \le $ is the total order on $ Q $ such that $ 1 < 2 < 3 < 4 < 5 $, $ C_1 = \{1, 2 \} $ and $ C_2 = \{4, 5 \} $, then $ C_1 $ and $ C_2 $ are $ \le $-convex, but $ C_1 \cup C_2 $ is not $ \le $-convex. Nonetheless, if $ C_1 \cap C_2 \not = \emptyset $, then $ C_1 \cup C_2 $ must be $ \le $-convex, as shown in the following lemma.

\begin{lemma}\label{lem:convexunion}
    Let $ Q $ be a set, and let $ \le $ be a total order on $ Q $. Let $ C_1 $ and $ C_2 $ be $ \le$-convex sets. If $ C_1 \cap C_2 \not = \emptyset $, then $ C_1 \cup C_2 $ is $ \le $-convex.
\end{lemma}

\begin{proof}
    Since $ C_1 \cap C_2 \not = \emptyset $, we can pick $ z \in C_1 \cap C_2 $. Let $ u, v, w \in Q $ such that $ u \le v \le w $ and $ u, w \in C_1 \cup C_2 $. We must prove that $ v \in C_1 \cup C_2 $. Assume that $ v \le z $ (the case $ z < v $ is analogous). Then, $ u \le v \le z $. Assume that $ u \in C_1 $ (the case $ u \in C_2 $ is analogous). From $ u \le v \le z $ and $ u, z \in C_1 $ we obtain $ v \in C_1 $ because $ C_1 $ is $ \le $-convex, so we conclude $ v \in C_1 \cup C_2 $.
\end{proof}

Lastly, we extend the notion of convexity to partitions and equivalence relations. Let $ Q $ be a set, and let $ \le $ be a total order on $ Q $. A partition $ \mathcal{P} $ of $ Q $ is $ \le $-convex if every $ U \in \mathcal{P} $ is $ \le $-convex. An equivalence relation $ \equiv $ on $ Q $ is $ \le $-convex if the partition $ \mathcal{P}_\equiv $ of $ Q $ is $ \le $-convex (equivalently, if $ [u]_\equiv $ is $ \le $-convex for every $ u \in Q $).

\subsection{Automata}

Let $ \Sigma $ be a fixed finite alphabet. We denote by $ \Sigma^* $ the set of all finite strings over $ \Sigma $, where $ \epsilon \in \Sigma^* $ is the empty string. Given a language $ \mathcal{L} \subseteq \Sigma^* $, we denote by $ \Pref(\mathcal{L}) $ the set of all prefixes of some element of $ \mathcal{L} $, that is, $ \Pref(\mathcal{L}) = \{\alpha \in \Sigma^* \;|\; (\exists \alpha' \in \Sigma^*)(\alpha \alpha' \in \mathcal{L}) \} $.

A non-deterministic finite automaton (NFA) is a 4-tuple $ \mathcal{A} = (Q, E, s, F) $, where $ Q $ is the finite set of states, $ E \subseteq Q \times Q \times \Sigma $ is the set of edges (also known as transitions), $ s \in Q $ is the (unique) initial state and $ F \subseteq Q $ is the set of final states. We denote by $ \mathcal{L}(\mathcal{A}) $ the language recognized by $ \mathcal{A} $. Throughout the paper, we assume that, for every $ u \in Q $, we have that $ u $ is reachable in $ \mathcal{A} $ (that is, there is a walk from the initial state $ s $ to $ u $) and $ u $ is co-reachable in $ \mathcal{A} $ (that is, there is a walk from $ u $ to a final state, and in particular $ u $ is co-reachable if $ u \in F $). These are standard assumptions in automata theory because, if a state $ u $ is not reachable or is not co-reachable, then it can be removed without changing $ \mathcal{L}(\mathcal{A}) $ (as long as the set of states does not become empty). Note that, under these assumptions, the empty language is not recognized by any NFA.

Let $ \mathcal{A} = (Q, E, s, F) $ be an NFA. For every $ u \in Q $, we denote by $ I^\mathcal{A}_u $ the set of all $ \alpha \in \Sigma^* $ such that $ \alpha $ can be read following some walk from $ s $ to $ u $ (in particular, $ \epsilon \in I^\mathcal{A}_s $). Note that $ \mathcal{L}(\mathcal{A}) = \bigcup_{u \in F} I_u^\mathcal{A} $. Moreover, we have $ I_u^\mathcal{A} \not = \emptyset $ for every $ u \in Q $ because every $ u \in Q $ is reachable in $ \mathcal{A} $, and $ \Pref (\mathcal{L}(\mathcal{A})) = \bigcup_{u \in Q} I_u^\mathcal{A} $ because every $ u \in Q $ is co-reachable in $ \mathcal{A} $.

An NFA $ \mathcal{A} = (Q, E, s, F) $ is a deterministic finite automaton (DFA) if for every $ u \in Q $ and for every $ a \in \Sigma $ there exists at most one $ v \in Q $ such that $ (u, v, a) \in E $ (so in this paper a DFA needs not be total). If $ \mathcal{A} = (Q, E, s, F) $ is a DFA, then for every $ u, v \in Q $ such that $ u \not = v $ we have $ I_u^\mathcal{A} \cap I_v^\mathcal{A} = \emptyset $.

Let $ \mathcal{A} = (Q, E, s, F) $ and $ \mathcal{A}' = (Q', E', s', F') $ be two NFAs. A bijective function $ \phi $ from $ Q $ to $ Q' $ is an isomorphism from $ \mathcal{A} $ to $ \mathcal{A}' $ if (i)~for every $ u, v \in Q $ and for every $ a \in \Sigma $ we have $ (u, v, a) \in E $ if and only if $ (\phi(u), \phi(v), a) \in E' $, (ii)~$ \phi(s) = s' $ and (iii)~for every $ u \in Q $, $ u \in F $ if and only if $ \phi(u) \in F' $.

When we consider a class $ C $ of NFAs, we say that an element $ \mathcal{A} $ of $ C $ is a minimal element of $ C $ if every element of $ C $ has at least as many states as $ \mathcal{A} $.

\section{Wheeler Automata}\label{sec:wheelerautomata}

As customary in the literature on Wheeler languages (e.g.,~\cite{alanko2020regular, cotumaccio2023co, conte2023computing}), we fix a total order $ \preceq $ on $ \Sigma $. In our examples, we always assume that $ \Sigma $ is the English alphabet and $ \preceq $ is the usual total order on $ \Sigma $ such that $ a \prec b $, $ b \prec c $, $ c \prec d $, and so on. Moreover, we extend $ \preceq $ to $ \Sigma^* $ \emph{co-lexicographically}, that is, for every $ \alpha \in \Sigma^* $ we have $ \alpha \prec \beta $ if and only if the reverse string $ \alpha ^R $ is smaller than the reverse string $ \beta^R $ in the standard dictionary order. Formally, let $ \gamma $ be the longest common suffix of $ \alpha $ and $ \beta $ (possibly, $ \gamma = \epsilon $), and let $ \alpha', \beta' \in \Sigma^* $ such that $ \alpha = \alpha' \gamma $ and $ \beta = \beta' \gamma $. Then, we have $ \alpha \prec \beta $ if and only $ \beta' \not = \epsilon $ and one of the following is true: (i)~$ \alpha' = \epsilon $ or (ii)~$ \alpha' \not = \epsilon $ and $ a \prec b $, where $ a $ is the last character of $ \alpha' $ and $ b $ is the last character of $ \beta' $. For example, we have $ a \prec aa $ and $ dba \prec ca $.

Let us recall the definition of Wheeler automata.

\begin{definition}
    Let $ \mathcal{A} = (Q, E, s, F) $ be an NFA, and let $ \le $ be a total order on $ Q $. We say that $ (\mathcal{A}, \le) $ is a \emph{Wheeler NFA} if:
    \begin{itemize}
        \item (Axiom~1) For every $ u \in Q $, $ s \le u $.
        \item (Axiom~2) For every $ u, v, u', v' \in Q $ and for every $ a, b \in \Sigma $, if $ (u, v, a), (u', v', b) \in E $ and $ v < v' $, then $ a \preceq b $.
        \item (Axiom~3) For every $ u, v, u', v' \in Q $ and for every $ a \in \Sigma $, if $ (u, v, a), (u', v', a) \in E $ and $ v < v' $, then $ u \le u' $.
    \end{itemize}
    We say that $ \le $ is a \emph{Wheeler order} on $ \mathcal{A} $. If $ \mathcal{A} $ is a DFA, we say that $ (\mathcal{A}, \le) $ is a \emph{Wheeler DFA}. Moreover, $ \mathcal{L} \subseteq \Sigma^* $ is a \emph{Wheeler language} if there exists a Wheeler NFA $ (\mathcal{A}, \le) $ such that $ \mathcal{L}(\mathcal{A}) = \mathcal{L} $.
\end{definition}

See Figure~\ref{fig:examplewheeler} for an example of a Wheeler NFA and for a graphical interpretation of Axiom~3 (if we list all states consistently with the Wheeler order, we duplicate the list, and we draw the edges of the Wheeler NFA starting from the bottom list and reaching the top list, then Axiom~3 ensures that equally labeled edges do not cross). This graphical interpretation is useful to follow some proofs (see Lemma~\ref{lem:ideabehindalgorithm} and see Lemma~\ref{lem:correctnessbisimulation} in the appendix). 

\begin{figure}[h!]
     \centering
     \begin{subfigure}[b]{0.49\textwidth}
        \centering
        \scalebox{.8}{
        \begin{tikzpicture}[->,>=stealth', semithick, auto, scale=1]
\node[state, initial] (1)    at (0,0)	{$ u_1 $};
\node[state] (2)    at (2, 1)	{$ u_2 $};
\node[state, accepting] (3)    at (4, 1)	{$ u_3 $};
\node[state, accepting] (4)    at (2, -1)	{$ u_4 $};

\draw (1) edge [] node [] {$ a $} (2);
\draw (2) edge [] node [] {$ a $} (3);
\draw (1) edge [bend left = 50] node [] {$ a $} (3);
\draw (3) edge [loop right] node [] {$ b $} (3);
\draw (1) edge [] node [] {$ c $} (4);
\draw (4) edge [] node [] {$ a $} (3);
\draw (2) edge [] node [] {$ c $} (4);
\draw (4) edge [loop right] node [] {$ c $} (4);
\end{tikzpicture}
}
\end{subfigure}
     \begin{subfigure}[b]{0.49\textwidth}
        \centering
        \scalebox{.8}{
        \begin{tikzpicture}[->,>=stealth', semithick, auto, scale=1]
\node[state] (1)    at (0,0)	{$ u_1 $};
\node[state] (2)    at (2, 0)	{$ u_2 $};
\node[state] (3)    at (4, 0)	{$ u_3 $};
\node[state] (4)    at (6, 0)	{$ u_4 $};
\node[state] (1')    at (0,-2)	{$ u_1 $};
\node[state] (2')    at (2, -2)	{$ u_2 $};
\node[state] (3')    at (4, -2)	{$ u_3 $};
\node[state] (4')    at (6, -2)	{$ u_4 $};

\draw (1') edge [] node [] {$ a $} (2);
\draw (2') edge [] node [] {$ a $} (3);
\draw (1') edge [] node [] {$ a $} (3);
\draw (4') edge [] node [] {$ a $} (3);
\end{tikzpicture}
}
\end{subfigure}
 	\caption{\emph{Left:} A Wheeler NFA. The states are numbered following the Wheeler order. \emph{Right:} Edges labeled with the same character (e.g., $ a $) do not cross.}
    \label{fig:examplewheeler}
\end{figure}

Let us recall some properties of Wheeler NFAs.

\begin{lemma}[{\cite[Lemma 3.7]{alanko2021wheeler}}]\label{lem:fromNFAtoDFA}
    Let $ \mathcal{L} \subseteq \Sigma^* $ be a Wheeler language. Then, there exists a Wheeler DFA $ (\mathcal{A}, \le) $ such that $ \mathcal{L}(\mathcal{A}) = \mathcal{L} $.
\end{lemma}

In other words, every language recognized by a Wheeler NFA is also recognized by a Wheeler DFA.

\begin{lemma}[{\cite[Lemma 3.2]{alanko2021wheeler}}]\label{lem:consistencystatesstrings}
    Let $ (\mathcal{A}, \le) $ be a Wheeler NFA, with $ \mathcal{A} = (Q, E, s, F) $. Let $ u, v \in Q $, and let $ \alpha \in I_u^\mathcal{A} $ and $ \beta \in I_v^\mathcal{A} $ such that $ \{\alpha, \beta \} \not \subseteq I_u^\mathcal{A} \cap I_v^\mathcal{A} $. Then, $ u < v $ if and only if $ \alpha \prec \beta $.
\end{lemma}

In other words, a Wheeler order must be consistent with the strings reaching each state, up to intersections. For example, in Figure~\ref{fig:examplewheeler} we have $ a \in I_{u_2}^\mathcal{A} $, $ aa \in I_{u_3}^\mathcal{A} $, $ \{a, aa \} \not \subseteq I_{u_2}^\mathcal{A} \cap I_{u_3}^\mathcal{A} $, $ 2 < 3 $ and, consistently, $ a \prec aa $. If we interpret the strings in $ I_u^\mathcal{A} $ as the set of all ``prefixes'' of $ u $, then Lemma~\ref{lem:consistencystatesstrings} intuitively states that a Wheeler order must (co-lexicographically) sort the prefixes of the states. This explains why Wheeler automata can be seen as a non-trivial extension of \emph{prefix arrays} and \emph{suffix arrays} from strings to automata, which was the original motivation for studying Wheeler automata~\cite{gagie2017wheeler}. The condition $ \{\alpha, \beta \} \not \subseteq I_u^\mathcal{A} \cap I_v^\mathcal{A} $ of Lemma~\ref{lem:consistencystatesstrings} implies that the order between two states is not disambiguated by the strings reaching both states. For example, consider the NFA $ \mathcal{A} = (Q, E, s, F) $, where $ Q = \{s, u_1, u_2, u_3 \} $ consists of four pairwise distinct states, $ E = \{(s, u_1, a), (s, u_1, b), (u_1, u_2, c), (u_1, u_3, c) \} $ and $ F = \{u_2, u_3 \} $. Let $ \le $ be the total order on $ Q $ such that $ s < u_1 < u_2 < u_3 $. Then, $ (\mathcal{A}, \le) $ is a Wheeler NFA. Notice that $ bc \in I_{u_2}^\mathcal{A} $, $ ac \in I_{u_3}^\mathcal{A} $, $ u_2 < u_3 $ and $ ac \prec bc $. This does not contradict Lemma~\ref{lem:consistencystatesstrings} because $ \{ac, bc \} \subseteq I_{u_2}^\mathcal{A} \cap I_{u_3}^\mathcal{A} $.

Recall that, if $ \mathcal{A} $ is a DFA, then $ u \not = v $ implies $ I_u^\mathcal{A} \cap I_v^\mathcal{A} = \emptyset $. This means that, if $ (\mathcal{A}, \le) $ is a Wheeler DFA, Lemma~\ref{lem:consistencystatesstrings} can be restated as follows: let $ u, v \in Q $ such that $ u \not = v $, and let $ \alpha \in I_u^\mathcal{A} $ and $ \beta \in I_v^\mathcal{A} $; then, $ u < v $ if and only if $ \alpha \prec \beta $. In particular, if $ \mathcal{A} $ is a DFA, then there exists \emph{at most} one Wheeler order $ \le $ on $ \mathcal{A} $, because the order between two states $ u, v \in Q $ such that $ u \not = v $ is uniquely determined by any $ \alpha \in I_u^\mathcal{A} $ and $ \beta \in I_v^\mathcal{A} $. In other words, if $ (\mathcal{A}, \le) $ is a Wheeler DFA, then $ \le $ is uniquely determined by $ \mathcal{A} $. In general, this is not true for Wheeler NFAs. For example, if we consider the NFA $ \mathcal{A} = (Q, E, s, F) $, where $ Q = \{s, u, v \} $ consists of three pairwise distinct states, $ E = \{(s, u, a), (s, v, a) \} $ and $ F = \{u, v \} $, then both total orders in which $ s $ comes first are Wheeler orders on $ \mathcal{A} $.

Every Wheeler language is star-free~\cite{alanko2021wheeler}, so $ (aa)^* $ is not Wheeler because $ (aa)^* $ is the classical example of a regular language that is not star-free~\cite{straubing2012finite}. A more interesting example is the regular language $ \mathcal{L} = ae^*b + ce^*d $ on the alphabet $ \Sigma = \{a, b, c, d, e \} $. Notice that $ \mathcal{L} $ is star-free because $ e^* $ is the complement of $ \Sigma^* (a + b + c + d) \Sigma^* $, where $ \Sigma^* $ is the complement of $ \emptyset $. However, $ \mathcal{L} $ is not Wheeler, as proved next (we adapt~\cite[Example 3]{alanko2021wheeler}).  If $ \mathcal{L} $ were Wheeler, then $ \mathcal{L} $ would be recognized by a Wheeler DFA $ (\mathcal{A}, \le) $ by Lemma~\ref{lem:fromNFAtoDFA}. Let $ \alpha_i = ae^i $ if $ i $ is odd and $ \alpha_i = ce^i $ if $ i $ is even. For every $ i $ we have $ \alpha_i \in \Pref (\mathcal{L}) $, so there exists a (unique) state $ u_i $ of $ \mathcal{A} $ such that $ \alpha_i \in I_{u_i} ^\mathcal{A} $. Notice that $ \alpha_1 \prec \alpha_2 \prec \alpha_3 \prec \dots $, so by Lemma~\ref{lem:consistencystatesstrings} we have $ u_1 \le u_2 \le u_3 \le \dots $, and the finiteness of the number of states implies that $ u_i = u_{i + 1} = u_{i + 2} = \dots $ for some $ i $. We can assume that $ i $ is odd. Then, $ \alpha_i, \alpha_{i + 1} \in I_{u_i}^\mathcal{A} $ (the two strings reach the same state), so $ \alpha_i b \in \mathcal{L} $ if and only if $ \alpha_{i + 1} b \in \mathcal{L} $. This is a contradiction because $ \alpha_i b = ae^i b \in \mathcal{L} $ and $ \alpha_{i + 1} b = ce^{i + 1} b \not \in \mathcal{L} $. 

\section{Wheeler Bisimulations}\label{sec:mainresults}

Let $ \mathcal{L} \subseteq \Sigma^* $ be a Wheeler language. As mentioned in the introduction, up to isomorphism there exists a unique (state)-minimal Wheeler DFA $ (\mathcal{A}, \le) $ such that $ \mathcal{L}(\mathcal{A}) = \mathcal{L} $~\cite{alanko2021wheeler}. The first result in the paper is a simple example showing that this is not true for Wheeler NFAs (analogously, in general there is not a unique minimal NFA recognizing a given regular language).

\begin{example}\label{ex:nonisomorphicexample}
    Consider the regular language $ \mathcal{L} = a a^* $. No (Wheeler) NFA with one state can recognize $ \mathcal{L} $, and Figure~\ref{fig:minimalNFAnotunique} shows two non-isomorphic Wheeler NFAs with two states recognizing $ \mathcal{L} $.

\begin{figure}[h!]
     \centering
     \begin{subfigure}[b]{0.49\textwidth}
        \centering
        \scalebox{.8}{
        \begin{tikzpicture}[->,>=stealth', semithick, auto, scale=1]
\node[state, initial] (1)    at (0,0)	{$ u_1 $};
\node[state, accepting] (2)    at (2, 0)	{$ u_2 $};

\draw (1) edge [loop above] node [] {$ a $} (1);
\draw (1) edge [] node [] {$ a $} (2);
\end{tikzpicture}
}
\end{subfigure}
     \begin{subfigure}[b]{0.49\textwidth}
        \centering
        \scalebox{.8}{
        \begin{tikzpicture}[->,>=stealth', semithick, auto, scale=1]
\node[state, initial] (1)    at (0,0)	{$ u_1 $};
\node[state, accepting] (2)    at (2, 0)	{$ u_2 $};

\draw (1) edge [] node [] {$ a $} (2);
\draw (2) edge [loop above] node [] {$ a $} (2);
\end{tikzpicture}
}
\end{subfigure}
 	\caption{Two non-isomorphic minimal Wheeler NFAs recogning $ \mathcal{L} = a a^* $ (see Example~\ref{ex:nonisomorphicexample}). The states of each NFA are numbered following the corresponding Wheeler order.}
    \label{fig:minimalNFAnotunique}
\end{figure}
\end{example}

Our goal is to introduce a meaningful notion of minimality for Wheeler NFAs that ensures uniqueness up to isomorphism. To this end, let us recall the standard definition of bisimulation (see, e,g.,~\cite{kozen2007automata, baier2008principles}).

\begin{definition}\label{def:bisimulation}
    Let $ \mathcal{A} $ and $ \mathcal{A}' $ be NFAs, with $ \mathcal{A} = (Q, E, s, F) $ and $ \mathcal{A}' = (Q', E', s', F') $. A relation $ R \subseteq Q \times Q' $ is a \emph{bisimulation} from $ \mathcal{A} $ to $ \mathcal{A}' $ if:
    \begin{itemize}
        \item For every $ u, v \in Q $, $ u' \in Q' $ and $ a \in \Sigma $, if $ u \; R \; u' $ and $ (u, v, a) \in E $, then there exists $ v' \in Q' $ such that $ v \; R \; v' $ and $ (u', v', a) \in E' $.
        \item For every $ u \in Q $, $ u', v' \in Q' $ and $ a \in \Sigma $, if $ u \; R \; u' $ and $ (u', v', a) \in E' $, then there exists $ v \in Q $ such that $ v \; R \; v' $ and $ (u, v, a) \in E $.
        \item We have $ s \; R \; s' $.
        \item For every $ u \in Q $ and $ u' \in Q $ such that $ u \; R \; u' $, we have $ u \in F $ if and only if $ u' \in F' $.
    \end{itemize}
\end{definition}

Let us recall why bisimulations lead to a notion of uniqueness for (standard) NFAs. Fix an NFA $ \mathcal{A} $, and consider the class of all NFAs $ \mathcal{A}' $ such that there exists a bisimulation from $ \mathcal{A} $ to $ \mathcal{A}' $. Then, all the \emph{minimal} NFAs in this class (that is, all the NFAs in the class having the minimum number of states) must be isomorphic~\cite{ kozen2007automata, baier2008principles}. In other words, up to isomorphism, there exists only one NFA $ \mathcal{A}' = (Q', E', s', F') $ for which (i)~there exists a bisimulation from $ \mathcal{A} $ to $ \mathcal{A}' $ and (ii)~for every NFA $ \mathcal{A}'' = (Q'', E'', s'', F'') $ such that $ |Q''| < |Q'| $, there exists no bisimulation from $ \mathcal{A} $ to $ \mathcal{A}'' $.

It is tempting to directly use the notion of bisimulation on Wheeler NFAs to identify a unique minimal Wheeler NFA. However, the next example shows that such an attempt fails.

\begin{example}\label{ex:bisimulationisnotenough}

\begin{figure}[h!]
     \centering
     \begin{subfigure}[b]{0.37\textwidth}
        \centering
        \scalebox{.7}{
        \begin{tikzpicture}[->,>=stealth', semithick, auto, scale=1]
\node[state, initial, accepting] (1)    at (0,0)	{$ u_1 $};
\node[state, accepting] (2)    at (-2, 2)	{$ u_2 $};
\node[state, accepting] (3)    at (0, 2)	{$ u_3 $};
\node[state, accepting] (4)    at (2, 2)	{$ u_4 $};

\draw (1) edge [bend left] node [] {$ a, b $} (2);
\draw (1) edge [] node [] {$ b $} (3);
\draw (1) edge [bend right] node [] {$ b,c $} (4);
\draw (2) edge [bend left] node [] {$ d $} (4);
\draw (4) edge [loop right] node [] {$ d $} (4);

\end{tikzpicture}
}
\end{subfigure}
     \begin{subfigure}[b]{0.37\textwidth}
        \centering
        \scalebox{.7}{
        \begin{tikzpicture}[->,>=stealth', semithick, auto, scale=1]
\node[state, initial, accepting] (1)    at (0,0)	{$ u'_1 $};
\node[state, accepting] (2)    at (-2, 2)	{$ u'_2 $};
\node[state, accepting] (3)    at (0, 2)	{$ u'_3 $};
\node[state, accepting] (4)    at (2, 2)	{$ u'_4 $};

\draw (1) edge [bend left] node [] {$ a, b $} (2);
\draw (1) edge [] node [] {$ b $} (3);
\draw (1) edge [bend right] node [] {$ c $} (4);
\draw (2) edge [bend left] node [] {$ d $} (4);
\draw (4) edge [loop right] node [] {$ d $} (4);

\end{tikzpicture}
}
\end{subfigure}
     \begin{subfigure}[b]{0.24\textwidth}
        \centering
        \scalebox{.7}{
        \begin{tikzpicture}[->,>=stealth', semithick, auto, scale=1]
\node[state, initial, accepting] (1)    at (0,0)	{$ u $};
\node[state, accepting] (2)    at (-1, 2)	{$ v $};
\node[state, accepting] (3)    at (1, 2)	{$ w $};

\draw (1) edge [bend left] node [] {$ a, b, c $} (2);
\draw (1) edge [bend right] node [] {$ b $} (3);
\draw (2) edge [loop above] node [] {$ d $} (2);

\end{tikzpicture}
}
\end{subfigure}
 	\caption{The three NFAs used in Example~\ref{ex:bisimulationisnotenough}. On the left, the Wheeler NFA $ (\mathcal{A}, \le) $. In the center, the Wheeler NFA $ (\mathcal{A}', \le') $. On the right, the NFA $ \mathcal{A}'' $. The states of $ (\mathcal{A}, \le) $ and $ (\mathcal{A}', \le') $ are numbered following the corresponding Wheeler orders.}
    \label{fig:standardbisimulationisbad}
\end{figure}

Consider the Wheeler NFAs $ (\mathcal{A}, \le) $ and $ (\mathcal{A}', \le') $ in Figure~\ref{fig:standardbisimulationisbad}, where $ \mathcal{A} = (Q, E, s, F) $ and $ \mathcal{A}' = (Q', E', s', F') $. Let $ R \subseteq Q \times Q' $ be the relation described by the following pairs: $ u_1 \; R \; u'_1 $, $ u_2 \; R \; u'_2 $, $ u_3 \; R \; u'_3 $, $ u_4 \; R \; u'_4 $, $ u_4 \; R \; u'_2 $. It is easy to check that $ R $ is a bisimulation from $ \mathcal{A} $ to $ \mathcal{A}' $. Moreover, $ id_Q $ is a bisimulation from $ \mathcal{A} $ to $ \mathcal{A} $.

We want to show that every Wheeler NFA $ (\mathcal{A}^\#, \le^\#) $ for which there exists a bisimulation from $ \mathcal{A} $ to $ \mathcal{A}^\# $ must have at least four states. Since $ \mathcal{A} $ and $ \mathcal{A}' $ are non-isomorphic and have four states, the notion of bisimulation does not lead to a \emph{unique} minimal Wheeler NFAs.

The NFA $ \mathcal{A}'' = (Q'', E'', s'', F'') $ in Figure~\ref{fig:standardbisimulationisbad} is a minimal NFA such that there exists a bisimulation from $ \mathcal{A} $ to $ \mathcal{A}'' $, and from the theory of bisimulation we know that $ \mathcal{A}'' $ is unique up to isomorphism~\cite{ kozen2007automata, baier2008principles}. Since $ \mathcal{A}'' $ has three states, we only need to prove that there exists no total order $ \le'' $ on $ Q'' $ such that $ (\mathcal{A}'', \le'') $ is a Wheeler NFA, because this implies that every Wheeler NFA $ (\mathcal{A}^\#, \le^\#) $ for which there exists a bisimulation from $ \mathcal{A} $ to $ \mathcal{A}^\# $ must have at least four states. Suppose for the sake of a contradiction that there exists a total order $ \le'' $ on $ Q $ such that $ (\mathcal{A}'', \le'') $ is a Wheeler NFA. Let us show that both $ v < w $ and $ w < v $ lead to a contradiction. We cannot have $ v < w $ because from $ (u, v, c), (u, w, b) \in E'' $ we would obtain $ c \preceq b $ by Axiom~2. We cannot have $ w < v $ because from $ (u, w, b), (u, v, a) \in E'' $ we would obtain $ b \preceq a $ again by Axiom~2.

We conclude that, in general, the standard notion of bisimulation does not induce a unique minimal Wheeler NFA.
\end{example}

In view of Example~\ref{ex:bisimulationisnotenough}, to retrieve a unique minimal Wheeler NFA, we now define \emph{Wheeler bisimulations}. The main idea is that we should only consider bisimulations that respect the convex structure of the considered Wheeler NFAs.

\begin{definition}\label{def:wheelerbisimulation}
Let $ (\mathcal{A}, \le) $ and $ (\mathcal{A}', \le') $ be Wheeler NFAs, with $ \mathcal{A} = (Q, E, s, F) $ and $ \mathcal{A}' = (Q', E', s', F') $. A relation $ R \subseteq Q \times Q' $ is a \emph{Wheeler bisimulation} from $ (\mathcal{A}, \le) $ to $ (\mathcal{A}', \le') $ if:
\begin{itemize}
    \item (Property~1) $ R $ is a bisimulation from $ \mathcal{A} $ to $ \mathcal{A} $'.
    \item (Property~2) For every $ C \subseteq Q $, if $ C $ is $ \le $-convex, then $ R(C) $ is $ \le' $-convex.
    \item (Property~3) For every $ C' \subseteq Q' $, if $ C' $ is $ \le' $-convex, then $ R^{-1}(C') $ is $ \le $-convex.
\end{itemize}
\end{definition}

Notice that the bisimulation $ R $ defined in Example~\ref{ex:bisimulationisnotenough} is not a Wheeler bisimulation from $ (\mathcal{A}, \le) $ to $ (\mathcal{A}', \le') $, because $ R(\{u_4 \}) = \{u'_2, u'_4 \} $, $ \{u_4 \} $ is $ \le $-convex, but $ \{u'_2, u'_4 \} $ is not $ \le' $-convex.

Consider the class $ \mathfrak{C} $ of all Wheeler NFAs, and for all Wheeler NFAs $ (\mathcal{A}, \le) $ and $ (\mathcal{A}', \le') $ let $ (\mathcal{A}, \le) \sim_{\mathfrak{C}} (\mathcal{A}', \le') $ if and only if there exists a Wheeler bisimulation from $ (\mathcal{A}, \le) $ to $ (\mathcal{A}', \le') $. It is not difficult to see that $ \sim_{\mathfrak{C}} $ is an equivalence relation on $ \mathfrak{C} $ (formally, it will follow from Lemma~\ref{lem:bisimequiv} in the appendix). Consequently, our goal is to prove that, in every $ \sim_{\mathfrak{C}} $-equivalence class, all Wheeler NFAs having the minimum number of states must be isomorphic.

Let $ \mathcal{L} \subseteq \Sigma^* $. Note that the class of all Wheeler NFAs $ (\mathcal{A}, \le) $ such that $ \mathcal{L}(\mathcal{A}) = \mathcal{L} $ may intersect infinitely many $ \sim_{\mathfrak{C}} $-equivalence classes, as shown in the next example.

\begin{example}\label{ex:canhaveinfiniteclasses}
    Consider the language $ \mathcal{L} = a^* $. Fix $ k \ge 3 $. Let $ \mathcal{A}^{(k)} = (Q^{(k)}, E^{(k)}, s^{(k)}, F^{(k)}) $ be the NFA defined as follows. Let $ Q^{(k)} = \{u_1^{(k)}, u_2^{(k)}, \dots, u_k^{(k)} \} $ be a set consisting of $ k $ pairwise distinct states, $ E^{(k)} = \{(u_i^{(k)}, u_{i + 1}^{(k)}, a) \;|\; 1 \le i \le k - 1 \} \cup \{(u_1^{(k)}, u_1^{(k)}, a) \} $, $ s^{(k)} = u_1^{(k)} $ and $ F^{(k)} = \{u_1^{(k)}, u_k^{(k)} \} $. Moreover, let $ \le^{(k)} $ be the total order on $ Q^{(k)} $ such that $ u_1^{(k)} < u_2^{(k)} < \dots < u_k^{(k)} $. It is immediate to check that $ (\mathcal{A}^{(k)}, \le^{(k)}) $ is a Wheeler NFA such that $ \mathcal{L}(\mathcal{A}^{(k)}) = \mathcal{L} $ (see Figure~\ref{fig:infiniteclasses}).

    Fix $ 3 \le h < k $. Let us prove that there exists no Wheeler bisimulation from $ (\mathcal{A}^{(h)}, \le^{(h)}) $ to $ (\mathcal{A}^{(k)}, \le^{(k)}) $. Suppose for the sake of a contradiction that there exists a Wheeler bisimulation $ R \subseteq Q^{(h)} \times Q^{(k)} $ from $ \mathcal{A}^{(h)} $ to $ \mathcal{A}^{(k)} $.
    
    Let us prove that $ u_i^{(h)} \; R \; u_i^{(k)} $ for every $ 1 \le i \le h $. We proceed by induction on $ i $. Since $ R $ is a bisimulation from $ \mathcal{A}^{(h)} $ to $ \mathcal{A}^{(k)} $ by Property~1, we have $ s^{(u)} \; R \; s^{(k)} $, or equivalently, $ u_1^{(h)} \; R \; u_1^{(k)} $. Now assume that $ 2 \le i \le h $. By the inductive hypothesis, we know that $ u_{i - 1}^{(h)} \; R \; u_{i - 1} ^{(k)} $. We know that $ R $ is a bisimulation from $ \mathcal{A}^{(h)} $ to $ \mathcal{A}^{(k)} $, so from $ (u_{i - 1}^{(h)}, u_i^{(h)}, a) \in E^{(h)} $ we obtain that there exists $ u' \in  Q^{(k)} $ such that $ u_i^{(h)} \; R \; u' $ and $ (u_{i - 1} ^{(k)}, u', a) \in E^{(k)} $. If $ i \ge 3 $, we must necessarily have $ u' = u_i^{(k)} $ by the definition of $ E^{(k)} $. If $ i = 2 $, by the definition of $ E^{(k)} $ we have $ (u' = u_2^{(k)}) \lor (u' = u_1^{(k)}) $. However, the case $ u' = u_1^{(k)} $ leads to a contradiction because $ R $ is a bisimulation from $ \mathcal{A}^{(h)} $ to $ \mathcal{A}^{(k)} $ and we would obtain $ u_2^{(h)} \; R \; u_1^{(k)} $, with $ u_2^{(h)} \not \in F^{(h)} $ (because $ h \ge 3 $) and $ u_1^{(k)} \in F^{(k)} $. We conclude that in both cases we must have $ u' = u_i^{(k)} $, so from $ u_i^{(h)} \; R \; u' $ we obtain $ u_i^{(h)} \; R \; u_i^{(k)} $.

    In particular, we have $ u_h^{(h)} \; R \; u_h^{(k)} $. This is a contradiction because $ R $ is a bisimulation from $ \mathcal{A}^{(h)} $ to $ \mathcal{A}^{(k)} $, $ u_h^{(h)} \not \in F^{(h)} $ and $ u_h^{(k)} \in F^{(k)} $ (because $ 3 \le h < k) $.

    \begin{figure}[h!]
     \centering
     \begin{subfigure}[b]{0.49\textwidth}
        \centering
        \scalebox{.7}{
        \begin{tikzpicture}[->,>=stealth', semithick, auto, scale=1]
\node[state, initial, accepting] (1)    at (0,0)	{$ u_1^{(3)} $};
\node[state] (2)    at (2, 0)	{$ u_2^{(3)} $};
\node[state, accepting] (3)    at (4, 0)	{$ u_3^{(3)} $};

\draw (1) edge [loop above] node [] {$ a $} (1);
\draw (1) edge [] node [] {$ a $} (2);
\draw (2) edge [] node [] {$ a $} (3);
\end{tikzpicture}
}
\end{subfigure}
     \begin{subfigure}[b]{0.49\textwidth}
        \centering
        \scalebox{.8}{
        \begin{tikzpicture}[->,>=stealth', semithick, auto, scale=1]
\node[state, initial, accepting] (1)    at (0,0)	{$ u_1^{(4)} $};
\node[state] (2)    at (2, 0)	{$ u_2^{(4)} $};
\node[state] (3)    at (4, 0)	{$ u_3^{(4)} $};
\node[state, accepting] (4)    at (6, 0)	{$ u_4^{(4)} $};

\draw (1) edge [loop above] node [] {$ a $} (1);
\draw (1) edge [] node [] {$ a $} (2);
\draw (2) edge [] node [] {$ a $} (3);
\draw (3) edge [] node [] {$ a $} (4);
\end{tikzpicture}
}
\end{subfigure}
 	\caption{The Wheeler NFAs $ (\mathcal{A}^{(3)}, \le^{(3)}) $ and $ (\mathcal{A}^{(4)}, \le^{(4)}) $ used in Example~\ref{ex:canhaveinfiniteclasses}.}
    \label{fig:infiniteclasses}
    \end{figure}
    
\end{example}

\section{Minimality}\label{sec:minimality}

Let $ \mathcal{A} $ be an NFA. An \emph{autobisimulation} on $ \mathcal{A} $ is a bisimulation from $ \mathcal{A} $ to $ \mathcal{A} $. It is well known that every NFA $ \mathcal{A} $ admits a \emph{maximal} autobisimulation which can be used to quotient $ \mathcal{A} $, and the quotient NFA is the \emph{minimal} NFA bisimilar to $ \mathcal{A} $~\cite{kozen2007automata, baier2008principles}.

In this section, our goal is to extend these results to Wheeler NFAs and Wheeler bisimulations. We define \emph{Wheeler autobisimulations} slightly differently from conventional bisimulations. The key idea is to add one more requirement, namely, \emph{reflexivity}. The intuition is that, in this way, we can use Lemma~\ref{lem:convexunion} to infer that the union of Wheeler autobisimulations \emph{must} be a Wheeler autobisimulation.

\begin{definition}
    Let $ (\mathcal{A}, \le) $ be a Wheeler NFA, with $ \mathcal{A} = (Q, E, s, F) $. A relation $ R \subseteq Q \times Q $ is a \emph{Wheeler autobisimulation} on $ (\mathcal{A}, \le) $ if $ R $ is a reflexive Wheeler bisimulation from $ (\mathcal{A}, \le) $ to $ (\mathcal{A}, \le) $.
\end{definition}

We can now show that every Wheeler NFA $ (\mathcal{A}, \le) $ admits a maximum Wheeler autobisimulation $ \equiv_{\mathcal{A}, \le} $. Crucially, $ \equiv_{\mathcal{A}, \le} $ is always $ \le $-convex.

\begin{theorem}\label{theor:maximumbisimulationproperties}
    Let $ (\mathcal{A}, \le) $ be a Wheeler NFA, with $ \mathcal{A} = (Q, E, s, F) $. Then, there exists a (unique) maximum Wheeler autobisimulation $ \equiv_{\mathcal{A}, \le} $ on $ (\mathcal{A}, \le) $, that is, a Wheeler autobisimulation $ \equiv_{\mathcal{A}, \le} $ on $ (\mathcal{A}, \le) $ that contains every Wheeler autobisimulation on $ (\mathcal{A}, \le) $. Moreover, $ \equiv_{\mathcal{A}, \le} $ is a $ \le $-convex equivalence relation on $ Q $.
\end{theorem}

Fix a Wheeler NFA $ (\mathcal{A}, \le) $. Consider the maximum Wheeler autobisimulation $ \equiv_{\mathcal{A}, \le} $ on $ (\mathcal{A}, \le) $, and consider the conventional maximum autobisimulation $ \equiv_\mathcal{A} $ on $ \mathcal{A} $. Both $ \equiv_{\mathcal{A}, \le} $ and $ \equiv_\mathcal{A} $ are always equivalence relations, and $ \equiv_{\mathcal{A}, \le} $ is always $ \le $-convex (see Theorem~\ref{theor:maximumbisimulationproperties} and~\cite{kozen2007automata}). In particular, $ \equiv_{\mathcal{A}, \le} $ is a conventional autobisimulation (by Property~1), so by the maximality of $ \equiv_\mathcal{A} $ we obtain that every $ \equiv_{\mathcal{A}, \le} $-equivalence class is contained in some $ \equiv_\mathcal{A} $-equivalence class. However, the next example shows that, in general, $ \equiv_{\mathcal{A}, \le} $ cannot be obtained from $ \equiv_\mathcal{A} $ by simply splitting each $ \equiv_\mathcal{A} $-equivalence class into its $ \le $-convex components.

\begin{example}\label{ex:splittingdoesntwork}
    Let $ (\mathcal{A}, \le) $ be the Wheeler NFA in Figure~\ref{fig:relationishipbetweenbisimulationandwheelerbisimulation}. The $ \equiv_\mathcal{A} $-equivalence classes are $ \{u_1 \} $, $ \{u_2, u_3 \} $, $ \{u_4, u_6 \} $, $ \{u_5, u_7 \} $, $ \{u_8 \} $ and the $ \equiv_{\mathcal{A}, \le} $-equivelence classes are $ \{u_1 \} $, $ \{u_2 \} $, $ \{u_3 \} $, $ \{u_4 \} $, $ \{u_5 \} $, $ \{u_6 \} $, $ \{u_7 \} $, $ \{u_8 \} $, so the $ \equiv_\mathcal{A} $-equivalence class $ \{u_2, u_3 \} $ must be split into $ \{u_2 \} $ and $ \{u_3 \} $ even if the set $ \{u_2, u_3 \} $ is $ \le $-convex.

    \begin{figure}
        \centering
        \scalebox{.8}{
        \begin{tikzpicture}[->,>=stealth', semithick, auto, scale=1]
\node[state] (1)    at (0,0)	{$ u_1 $};
\node[state] (2)    at (-2, 2)	{$ u_2 $};
\node[state] (3)    at (2, 2)	{$ u_3 $};
\node[state, accepting] (4)    at (-3, 4)	{$ u_4 $};
\node[state] (5)    at (-1, 4)	{$ u_5 $};
\node[state, accepting] (6)    at (1, 4)	{$ u_6 $};
\node[state] (7)    at (3, 4)	{$ u_7 $};
\node[state, accepting] (8)    at (0, 6)	{$ u_8 $};

\draw (1) edge [] node [] {$ a $} (2);
\draw (1) edge [] node [] {$ a $} (3);
\draw (2) edge [] node [] {$ b $} (4);
\draw (2) edge [] node [] {$ b $} (5);
\draw (3) edge [] node [] {$ b $} (6);
\draw (3) edge [] node [] {$ b $} (7);
\draw (4) edge [] node [] {$ c $} (8);
\draw (5) edge [] node [] {$ c $} (8);
\draw (6) edge [] node [] {$ c $} (8);
\draw (7) edge [] node [] {$ c $} (8);

\end{tikzpicture}
}
\caption{The Wheeler NFA $ (\mathcal{A}, \le) $ used in Example~\ref{ex:splittingdoesntwork}. The states are numbered following the Wheeler order $ \le $.}
\label{fig:relationishipbetweenbisimulationandwheelerbisimulation}
\end{figure} 
\end{example}

Let us show how to define the quotient Wheeler NFA $ (\mathcal{A}_*, \le_*) $ starting from the Wheeler NFA $ (\mathcal{A}, \le) $. The key point is that the quotient NFA $ \mathcal{A}_* $ can always be endowed with a Wheeler order $ \le_* $.

\begin{lemma}\label{lem:quotient}
    Let $ (\mathcal{A}, \le) $ be a Wheeler NFA, with $ \mathcal{A} = (Q, E, s, F) $. Let $ Q_* = \mathcal{P}_{\equiv_{\mathcal{A}, \le}} $, and let $ \le_* $ be the relation from $ Q_* $ to $ Q_* $ such that, for every $ U, V \in Q_* $, we have $ U \le_* V $ if and only if there exist $ u \in U $ and $ v \in V $ such that $ u \le v $.
    \begin{enumerate}
        \item Let $ U, V \in Q_* $ be such that $ U <_* V $. Then, for every $ u \in U $ and for every $ v \in V $ we have $ u < v $.
        \item $ \le_* $ is a total order on $ Q_* $.
        \item Let $ \mathcal{A}_* = (Q_*, E_*, s_*, F_*) $, where:
        \begin{itemize}
            \item $ E_* = \{(U, V, a) \in Q_* \times Q_* \times \Sigma \;|\; (\exists u \in U)(\exists v \in V)((u, v, a) \in E) \} $.
            \item $ s_* = [s]_{\equiv_{\mathcal{A}, \le}} $.
            \item $ F_* = \{U \in Q_* \;|\; (\exists u \in U)(u \in F) \} $.
        \end{itemize}
        Then, $ (\mathcal{A}_*, \le_*) $ is a Wheeler NFA.
        \item If $ \mathcal{A} $ is a DFA, then $ \mathcal{A}_* $ is a DFA.
    \end{enumerate}
\end{lemma}

The following lemma shows that, if there exists a Wheeler bisimulation from $ (\mathcal{A}, \le) $ to $ (\mathcal{A}', \le') $, then the corresponding quotients must be isomorphic, and the isomorphism is witnessed by a bisimulation $ R $ that respects the convex structure of each quotient.

\begin{lemma}\label{lem:quotientautomotaisomorphic}
    Consider a $ \sim_{\mathfrak{C}} $-equivalence class $ C $, let $ (\mathcal{A}, \le) $ and $ (\mathcal{A}', \le') $ be two elements of $ C $, and let $ (\mathcal{A}_*, \le_*) $ and $ (\mathcal{A}'_*, \le'_*) $ be the corresponding Wheeler NFAs built in Lemma~\ref{lem:quotient}, with $ \mathcal{A}_* = (Q_*, E_*, s_*, F_*) $ and $ \mathcal{A}'_* = (Q'_*, E'_*, s'_*, F'_*) $. Let $ Q_* = \{U_1, U_2, \dots, U_n \} $ and $ Q'_* = \{U'_1, U'_2, \dots, U'_{n'} \} $, where $ U_1 <_* U_2 <_* \dots <_* U_n $ and $ U'_1 <'_* U'_2 <'_* \dots <'_* U'_{n'} $. Then, $ n = n' $. Moreover, let $ \phi $ be the function from $ Q_* $ to $ Q'_* $ such that $ \phi(U_i) = U'_i $ for every $ 1 \le i \le n $. Then, (i)~$ \phi $ is a Wheeler bisimulation from $ (\mathcal{A}_*, \le_*) $ to $ (\mathcal{A}'_*, \le'_*) $, (ii)~$ \phi $ is the \emph{unique} Wheeler bisimulation from $ (\mathcal{A}_*, \le_*) $ to $ (\mathcal{A}'_*, \le'_*) $, and (iii)~$ \phi $ is an isomorphism from $ \mathcal{A}_* $ to $ \mathcal{A}'_* $.
\end{lemma}

\begin{proof}
    \emph{(Sketch)} One can show that there must exist a Wheeler bisimulation $ R $ from $ (\mathcal{A}_*, \le_*) $ to $ (\mathcal{A}'_*, \le'_*) $, and $ R $ must be an isomorphism (which implies $ n = n' $). To conclude the proof, we only need to prove that we must have $ R(U_i) = U'_i $ for every $ 1 \le i \le n $. We proceed by induction on $ i $. For $ i = 1 $, we have $ U_1 = s_* $ and $ U'_1 = s'_* $ by Axiom~1, and we know that $ R(s_*) = s'_* $ because $ R $ is an isomorphism from $ \mathcal{A}_* $ to $ \mathcal{A}'_* $, so $ R(U_i) = R(s_*) = s'_* = U'_1 $. Now assume that $ 2 \le i \le n $. By the inductive hypothesis, we know that $ R(U_j) = U'_j $ for every $ 1 \le j \le i - 1 $. We must show that $ R(U_i) = U'_i $. Let $ 1 \le k \le n $ be such that $ R(U_i) = U'_k $. We must prove that $ k = i $. We only need to prove that we cannot have $ k < i $ and we cannot have $ i < k $. We cannot have $ k < i $ because by the inductive hypothesis we would have $ R(U_k) = U'_k $, so we would have $ \{U_k, U_i \} \subseteq R^{-1}(U'_k) $ and $ R $ would not be injective. Moreover, we cannot have $ i < k $ because by the inductive hypothesis we would have $ R(\{U_1, U_2, \dots, U_{i - 2}, U_{i - 1}, U_i \}) = \{U'_1, U'_2, \dots, U'_{i - 2}, U'_{i - 1}, U'_k \} $, which would be a contradiction because $ R $ is a Wheeler bisimulation from $ (\mathcal{A}_*, \le_*) $ to $ (\mathcal{A}'_*, \le'_*) $, $ \{U_1, U_2, \dots, U_{i - 2}, U_{i - 1}, U_i \} $ is $ \le_* $-convex, but $ \{U'_1, U'_2, \dots, U'_{i - 2}, U'_{i - 1}, U'_k \} $ would not be $ \le'_* $-convex. \qedhere
\end{proof}

We can now prove that, if we consider the $ \sim_{\mathfrak{C}} $-equivalence class $ C $ of a Wheeler NFA $ (\mathcal{A}, \le) $, then $ (\mathcal{A}_*, \le_*) $ is a minimal element of $ C $ and, up to isomorphism, is the \emph{unique} minimal element of $ C $ (Theorem~\ref{theorem:minimalbisimilar} and Corollary~\ref{cor:quotientisstateminimal}).

\begin{theorem}\label{theorem:minimalbisimilar}
    Consider a $ \sim_{\mathfrak{C}} $-equivalence class $ C $, and let $ (\mathcal{A}, \le) $ and $ (\mathcal{A}', \le') $ be two minimal elements of $ C $. Then, there exists an isomorphism from $ \mathcal{A} $ to $ \mathcal{A}' $.
\end{theorem}

\begin{corollary}\label{cor:quotientisstateminimal}
   Consider a $ \sim_{\mathfrak{C}} $-equivalence class $ C $, let $ (\mathcal{A}, \le) $ be an element of $ C $, and let $ (\mathcal{A}_*, \le_*) $ be the Wheeler NFA of Lemma~\ref{lem:quotient}. Then, $ (\mathcal{A}_*, \le_*) $ is a minimal element of $ C $. 
\end{corollary}

We now show that every Wheeler language admits a unique minimal Wheeler DFA (up to isomorphism). This result was already proved in~\cite{alanko2020regular} using a different method; here, we use Wheeler bisimulations. In addition, we will show that \emph{all} Wheeler DFAs $ (\mathcal{A}, \le) $ such that $ \mathcal{L}(\mathcal{A}) = \mathcal{L} $ must be in the same $ \sim_{\mathfrak{C}} $-equivalence class $ T_\mathcal{L} $, and the minimal element of $ T_\mathcal{L} $ (unique up to isomorphism by Theorem~\ref{theorem:minimalbisimilar}) is the minimal Wheeler DFA. Note that a similar result holds for standard bisimulations: all DFAs recognizing a given regular language must be in the same bisimulation class, and the miminal element of this bisimulation class is the minimal DFA of the language~\cite{engelfriet1985determinancy}.

To prove our results, we will use the following lemma. Intuitively, \emph{in the deterministic case}, equivalent Wheeler automata are always related through a Wheeler bisimulation.

\begin{lemma}\label{lem:allDFAsbismilar}
    Let $ (\mathcal{A}, \le) $ and $ (\mathcal{A}', \le') $ be Wheeler DFAs, with $ \mathcal{A} = (Q, E, s, F) $ and $ \mathcal{A}' = (Q', E', s', F') $. If $ \mathcal{L}(\mathcal{A}) = \mathcal{L}(\mathcal{A}') $, there exists a Wheeler bisimulation $ R \subseteq Q \times Q' $ from $ (\mathcal{A}, \le) $ to $ (\mathcal{A}', \le') $.
\end{lemma}

\begin{proof}
    Let $ R \subseteq Q \times Q' $ be the relation such that, for every $ u \in Q $ and for every $ u' \in Q' $, we have $ u \; R \; u' $ if and only if $ I_u^\mathcal{A} \cap I_{u'}^{\mathcal{A}'} \not = \emptyset $. Let us prove that $ R $ is a Wheeler bisimulation from $ (\mathcal{A}, \le) $ to $ (\mathcal{A}', \le') $.
    
    \begin{itemize}
        \item Property~1 is true because $ R $ is a bisimulation from $ \mathcal{A} $ to $ \mathcal{A}' $ (see~\cite{engelfriet1985determinancy}).
        \item Let us prove Property~2. Let $ C \subseteq Q $ be a $ \le $-convex set. We must prove that $ R(C) $ is $ \le' $-convex. Fix $ u', v', w' \in Q' $ such that $ u' \le' v' \le w' $ and $ u', w' \in R(C) $. We must prove that $ v' \in R(C) $. If $ (v' = u') \lor (v' = w') $, the conclusion is immediate, so we can assume $ u' <' v' <' w' $.
        
        Since $ u', w' \in R(C) $, there exist $ u, w \in C $ such that $ u \; R \; u' $ and $ w \; R \; w' $. This means that there exist $ \alpha \in I_u^\mathcal{A} \cap I_{u'}^{\mathcal{A}'} $ and $ \gamma \in I_w^\mathcal{A} \cap I_{w'}^{\mathcal{A}'} $. We know that $ I_{v'}^{\mathcal{A}'} \not = \emptyset $, so we can pick $ \beta \in I_{v'}^{\mathcal{A}'} $. In particular, $ \beta \in \Pref (\mathcal{L}(\mathcal{A}')) $, so $ \beta \in \Pref (\mathcal{L}(\mathcal{A})) $ (we have $ \Pref (\mathcal{L}(\mathcal{A})) = \Pref (\mathcal{L}(\mathcal{A}')) $ because $ \mathcal{L}(\mathcal{A}) = \mathcal{L}(\mathcal{A}') $), and we conclude that there exists $ v \in Q $ such that $ \beta \in I_v^\mathcal{A} $.

        Let us prove that $ \alpha \prec \beta \prec \gamma $. We only prove that $ \alpha \prec \beta $ because one can prove analogously that $ \beta \prec \gamma $. Since $ \mathcal{A}' $ is a DFA and $ u' \not = v' $, we have $ I_{u'}^{\mathcal{A}'} \cap I_{v'}^{\mathcal{A}'} = \emptyset $. We have $ u' <' v' $, $ \alpha \in I_{u'}^{\mathcal{A}'} $ and $ \beta \in I_{v'}^{\mathcal{A}'} $, so by Lemma~\ref{lem:consistencystatesstrings} we obtain $ \alpha \prec \beta $.

        Let us prove that $ u \le v \le w $. We only prove that $ u \le v $ because one can prove analogously that $ v \le w $. Assume that $ u \not = v $. We need to prove that $ u < v $. Since $ \mathcal{A} $ is a DFA and $ u \not = v $, we have $ I_u^\mathcal{A} \cap I_v^\mathcal{A} = \emptyset $. We have $ \alpha \prec \beta $, $ \alpha \in I_u^\mathcal{A} $ and $ \beta \in I_v^\mathcal{A} $, so by Lemma~\ref{lem:consistencystatesstrings} we obtain $ u < v $.

        We know that $ C $ is $ \le $-convex, so from $ u \le v \le w $ and $ u, w \in C $ we conclude $ v \in C $. We have $ v \; R \; v' $ because $ \beta \in I_v^\mathcal{A} \cap I_{v'}^{\mathcal{A}'} $, so $ v' \in R(C) $.

        \item Let us prove Property~3. Let $ C' \subseteq Q' $ be a $ \le' $-convex set. We must prove that $ R^{-1}(C') $ is $ \le $-convex. The proof is analogous to the previous point. \qedhere
    \end{itemize}
\end{proof}

Fix a Wheeler language $ \mathcal{L} \subseteq \Sigma^* $. We define $ D_\mathcal{L} $ and $ T_\mathcal{L} $ as follows.
\begin{itemize}
    \item Let $ D_\mathcal{L} $ be the class of all Wheeler DFAs $ (\mathcal{A}, \le) $ such that $ \mathcal{L}(\mathcal{A}) = \mathcal{L} $.
    \item By Lemma~\ref{lem:fromNFAtoDFA}, there exists some Wheeler DFA $ (\mathcal{A}, \le) $ such that $ \mathcal{L}(\mathcal{A}) = \mathcal{L} $. Let $ T_\mathcal{L} $ be the $ \sim_{\mathfrak{C}} $-equivalence class to which $ (\mathcal{A}, \le) $ belongs. By Lemma~\ref{lem:allDFAsbismilar}, the definition of $ T_\mathcal{L} $ does not depend on the choice of $ (\mathcal{A}, \le) $. This means that, for every Wheeler DFA $ (\mathcal{A}', \le') $, if $ \mathcal{L}(\mathcal{A}') = \mathcal{L} $, then $ (\mathcal{A}', \le') $ belongs to $ T_\mathcal{L} $. Notice that this is consistent with Example~\ref{ex:canhaveinfiniteclasses} because the Wheeler NFAs $ \mathcal{A}^{(k)} $'s used in Example~\ref{ex:canhaveinfiniteclasses} are not Wheeler DFAs.
\end{itemize}

In particular, we have $ \emptyset \subsetneqq D_\mathcal{L} \subseteq T_\mathcal{L} $.

The key result is that $ D_\mathcal{L} $ and $ T_\mathcal{L} $ have the same minimal elements, as shown in the next lemma.

\begin{lemma}\label{lem:relatingdeterminismandbisimulation}
    Let $ \mathcal{L} \subseteq \Sigma^* $ be a Wheeler language, and let $ (\mathcal{A}, \le) $ be a Wheeler NFA. Then, $ (\mathcal{A}, \le) $ is a minimal element of $ D_\mathcal{L} $ if and only if it is a minimal element of $ T_\mathcal{L} $. In particular, every minimal element of $ T_\mathcal{L} $ is a Wheeler DFA.
\end{lemma}

We can now prove that every Wheeler language admits a unique minimal Wheeler DFA (up to isomorphism).

\begin{theorem}\label{theor:determinismmaintheor}
    Let $ \mathcal{L} \subseteq \Sigma^* $ be a Wheeler language, and let $ (\mathcal{A}, \le) $ and $ (\mathcal{A}', \le') $ be two minimal elements of $ D_\mathcal{L} $. Then, there exists an isomorphism from $ \mathcal{A} $ to $ \mathcal{A}' $.
\end{theorem}

\begin{proof}
    By Lemma~\ref{lem:relatingdeterminismandbisimulation}, $ (\mathcal{A}, \le) $ and $ (\mathcal{A}', \le') $ are two minimal elements of $ T_\mathcal{L} $. By Theorem~\ref{theorem:minimalbisimilar}, there exists an isomorphism from $ \mathcal{A} $ to $ \mathcal{A}' $. \qedhere
\end{proof}

Lastly, the minimal Wheeler DFA of a Wheeler language can be obtained by quotienting any Wheeler DFA that recognizes the same language.

\begin{corollary}\label{cor:determinismmaincor}
    Let $ \mathcal{L} \subseteq \Sigma^* $ be a Wheeler language, let $ (\mathcal{A}, \le) $ be a Wheeler DFA such that $ \mathcal{L}(\mathcal{A}) = \mathcal{L} $, and let $ (\mathcal{A}_*, \le_*) $ be the Wheeler automaton built in Lemma~\ref{lem:quotient}. Then, $ (\mathcal{A}_*, \le_*) $ is a minimal element of $ D_\mathcal{L} $.
\end{corollary}

\begin{proof}
    By Corollary~\ref{cor:quotientisstateminimal}, $ (\mathcal{A}_*, \le_*) $ is a minimal element of $ T_\mathcal{L} $, so by Lemma~\ref{lem:relatingdeterminismandbisimulation} $ (\mathcal{A}_*, \le_*) $ is a minimal element of $ D_\mathcal{L} $. \qedhere
\end{proof}

\section{Algorithms}\label{sec:algorithms}

Consider a Wheeler NFA $ (\mathcal{A}, \le) $, with $ \mathcal{A} = (Q, E, s, F) $, and assume that the edge labels can be sorted in $ O(|E|) $ time. The main result of this section is that the quotient Wheeler NFA can be built in linear time.

\begin{theorem}\label{theor:lineartimeminimization}
    Let $ (\mathcal{A}, \le) $ be a Wheeler NFA, with $ \mathcal{A} = (Q, E, s, F) $. We can build the Wheeler NFA $ (\mathcal{A}_*, \le_*) $ of Lemma~\ref{lem:quotient} in $ O(|E|) $ time.
\end{theorem}

In the special case where $ \mathcal{A} $ is a DFA, Theorem~\ref{theor:lineartimeminimization} states that we can compute the minimal Wheeler DFA in linear time (see Section~\ref{sec:minimality}), as already proved in~\cite{alanko2022linear}. However, in the non-deterministic setting, the algorithm and the correctness proof become much more involved. At the same time, our algorithm only requires a queue to achieve linear time, while the algorithm in~\cite{alanko2022linear} relies on auxiliary data structures.

Let us discuss the main ideas required to prove Theorem~\ref{theor:lineartimeminimization}. Fix a Wheeler NFA $ (\mathcal{A}, \le) $, with $ \mathcal{A} = (Q, E, s, F) $, and write $ Q = \{u_1, u_2, \dots, u_{|Q|} \} $, where $ u_1 < u_2 < \dots < u_{|Q|} $ (in particular, $ s = u_1 $ by Axiom~1). Recall that the set of all states of the quotient Wheeler NFA is the set of all $ \equiv_{\mathcal{A}, \le} $-equivalence classes (see Lemma~\ref{lem:quotient}). To prove Theorem~\ref{theor:lineartimeminimization}, the most difficult part is to efficiently compute a suitable representation of $ \equiv_{\mathcal{A}, \le} $.

A \emph{bit array} is an array whose entries are equal to zero or one. Let $ B_{\mathcal{A}, \le} [2, |Q|] $ be the bit array such that for every $ 2 \le i \le |Q| $ we have $ B_{\mathcal{A}, \le}[i] = 1 $ if and only if $ u_{i - 1} \not \equiv_{\mathcal{A}, \le} u_i $. The idea is that $ B_{\mathcal{A}, \le}[2, |Q|] $ stores all the information required to compute $ \equiv_{\mathcal{A}, \le} $ because $ \equiv_{\mathcal{A}, \le} $ is $ \le $-convex by Theorem~\ref{theor:maximumbisimulationproperties}. More precisely, we have the following result.

\begin{lemma}\label{lem:B_Aisallweneed}
    Let $ (\mathcal{A}, \le) $ be a Wheeler NFA, with $ \mathcal{A} = (Q, E, s, F) $. Then, for every $ 1 \le i, i' \le |Q| $, we have $ u_i \equiv_{\mathcal{A}, \le} u_{i'} $ if and only if for every $ \min \{i, i' \} + 1 \le k \le \max \{i, i' \} $ we have $ B_{\mathcal{A}, \le}[k] = 0 $.
\end{lemma}

\begin{proof}
    Recall that $ \equiv_{\mathcal{A}, \le} $ is a $ \le $-convex equivalence relation on $ Q $ by Theorem~\ref{theor:maximumbisimulationproperties}. Fix $ 1 \le i, i' \le |Q| $. We must prove that $ u_i \equiv_{\mathcal{A}, \le} u_{i'} $ if and only if for every $ \min \{i, i' \} + 1 \le k \le \max \{i, i' \} $ we have $ B_{\mathcal{A}, \le}[k] = 0 $.
    
    $ (\Leftarrow) $ For every $ \min \{i, i' \} + 1 \le k \le \max \{i, i' \} $, we have $ B_{\mathcal{A}, \le}[k] = 0 $, so $ u_{k - 1} \equiv_{\mathcal{A}, \le} u_k $. Since $ \equiv_{\mathcal{A}, \le} $ is transitive, we have $ u_{\min \{i, i' \}} \equiv_{\mathcal{A}, \le} u_{\max \{i, i' \}} $. We know that $ \equiv_{\mathcal{A}, \le} $ is symmetric, so $ u_i \equiv_{\mathcal{A}, \le} u_{i'} $.

    $ (\Rightarrow) $ Since $ u_i \equiv_{\mathcal{A}, \le} u_{i'} $ and $ \equiv_{\mathcal{A}, \le} $ is symmetric, we obtain $ u_{\min \{i, i' \}} \equiv_{\mathcal{A}, \le} u_{\max \{i, i' \}} $. For every $ \min \{i, i' \} + 1 \le k \le \max \{i, i' \} $, we have $ u_{\min \{i, i' \}} \equiv_{\mathcal{A}, \le} u_k $ because $ \min \{i, i' \} + 1 \le k \le \max \{i, i' \} $ and $ \equiv_{\mathcal{A}, \le} $ is $ \le $-convex by Theorem~\ref{theor:maximumbisimulationproperties}. We know that $ \equiv_{\mathcal{A}, \le} $ is an equivalence relation on $ Q $, so $ u_{k - 1} \equiv_{\mathcal{A}, \le} u_k $ for every $ \min \{i, i' \} + 1 \le k \le \max \{i, i' \} $, and we conclude $ B_{\mathcal{A}, \le}[k] = 0 $ for every $ \min \{i, i' \} + 1 \le k \le \max \{i, i' \} $. \qedhere
\end{proof}

Lemma~\ref{lem:B_Aisallweneed} suggests that computing $ B_{\mathcal{A}, \le} [2, |Q|] $ is the key step required to prove Theorem~\ref{theor:lineartimeminimization}. To this end, let us define $ a_i^{\min} $, $ \ell_i^{\min} $, $ a_i^{\max} $, $ \ell_i^{\max} $ and $ Out(i) $ for every $ 1 \le i \le |Q| $. 
\begin{itemize}
    \item For $ 2 \le i \le |Q| $, let $ a_i^{\min} $ be the smallest character labeling some edge reaching $ u_i $. In other words, let $ a_i^{\min} $ be the smallest $ c \in \Sigma $ (with respect to the total order $ \preceq $) such that $ (u_j, u_i, c) \in E $ for some $ 1 \le j \le |Q| $. Note that $ a_i^{\min} $ is well defined because we know that $ u_i $ must have at least one incoming edge (because $ u_i $ is reachable from the initial state $ s $). If $ u_1 $ has at least one incoming edge, define $ a_1^{\min} $ in the same way; otherwise, let $ a_1^{\min} = \bot $.
    \item For $ 2 \le i \le |Q| $, let $ \ell_i^{\min} $ be the smallest $ 1 \le j \le |Q| $ such that $ (u_j, u_i, a_i^{\min}) \in E $. If $ a_1^{\min} \not = \bot $, define $ \ell_1^{\min} $ in the same way; if $ a_1^{\min} = \bot $, define $ \ell_1^{\min} = \bot $.
    \item For $ 1 \le i \le |Q| $, define $ a_i^{\max} $ and $ \ell_i^{\max} $ in the same way as $ a_i^{\min} $and $ \ell_i^{\min} $, but replacing ``smallest'' with ``largest''.
    \item For $ 1 \le i \le |Q| $, let $ Out(i) $ be the set of all characters labeling some edge leaving $ u_i $. In other words, let $ Out(i) $ be the set of all $ a \in \Sigma $ for which there exists $ 1 \le j \le |Q| $ such that $ (u_i, u_j, a) \in E $.
\end{itemize}

Let us discuss some situations in which we have $ B_{\mathcal{A}, \le}[i] = 1 $.

\begin{figure}
     \centering
           \begin{subfigure}[b]{0.3\textwidth}
        \centering
        \scalebox{.7}{
        \begin{tikzpicture}[->,>=stealth', semithick, auto, scale=1]
\node[state, draw=none] (2)    at (1,0)	{$ u_{i} $};
\node[state, draw=none] (4)    at (4, 0)	{$ u_{i'} $};
\node[state, draw=none] (1')    at (-2, -3)	{$ u_{\ell_i - 1}^{\min} $};
\node[state, draw=none] (2')    at (0,-3)	{$ u_{\ell_i}^{\min} $};

\draw (2') edge [bend left] node [] {$ a_i^{\min} $} (2);
\draw (1') edge [] node [] {$ a_i^{\min} $} (4);

\end{tikzpicture}
}
\end{subfigure}
     \begin{subfigure}[b]{0.3\textwidth}
        \centering
        \scalebox{.7}{
        \begin{tikzpicture}[->,>=stealth', semithick, auto, scale=1]
\node[state, draw=none] (2)    at (1,0)	{$ u_i $};
\node[state, draw=none] (1')    at (-2, -3)	{$ u_{\ell_i - 1}^{\min} $};
\node[state, draw=none] (2')    at (0,-3)	{$ u_{\ell_i}^{\min} $};

\draw (2') edge [bend right] node [] {$ a_i^{\min} $} (2);
\draw (1') edge [] node [] {$ a_i^{\min} $} (2);

\end{tikzpicture}
}
\end{subfigure}
     \begin{subfigure}[b]{0.3\textwidth}
        \centering
        \scalebox{.7}{
        \begin{tikzpicture}[->,>=stealth', semithick, auto, scale=1]
\node[state, draw=none] (2)    at (1,0)	{$ u_{i'} $};
\node[state, draw=none] (4)    at (4, 0)	{$ u_i $};
\node[state, draw=none] (1')    at (-2, -3)	{$ u_{\ell_i - 1}^{\min} $};
\node[state, draw=none] (2')    at (0,-3)	{$ u_{\ell_i}^{\min} $};

\draw (2') edge [] node [] {$ a_i^{\min} $} (4);
\draw (1') edge [] node [] {$ a_i^{\min} $} (2);

\end{tikzpicture}
}
\end{subfigure}  
 	\caption{The three cases considered in the proof of Lemma~\ref{lem:ideabehindalgorithm}. \emph{Left:} the case $ i < i' $ would imply that equally labeled edges cross (see Figure~\ref{fig:examplewheeler}). \emph{Center:} the case $ i' = i $ would contradict the minimality of $ \ell_i^{\min} $. \emph{Right:} the case $ i' < i $ would imply $ B_{\mathcal{A}, \le}[i] \not = 1 $ because $ \equiv_{\mathcal{A}, \le} $ is $ \le $-convex.}\label{fig:notcrossingalgorithms}
\end{figure}

\begin{lemma}\label{lem:ideabehindalgorithm}
    Let $ (\mathcal{A}, \le) $ be a Wheeler NFA, with $ \mathcal{A} = (Q, E, s, F) $. Fix $ 2 \le i \le |Q| $ (in particular, $ \ell_i^{\min} \not = \bot $).
    \begin{enumerate}
        \item If $ Out(i - 1)  \not = Out(i) $, then $ B_{\mathcal{A}, \le}[i] = 1 $.
        \item If $ (u_{i - 1} \in F \land u_i \not \in F) \lor (u_{i - 1} \not \in F \land u_i  \in F) $, then $ B_{\mathcal{A}, \le}[i] = 1 $.
        \item Assume that $ B_{\mathcal{A}, \le}[i] = 1 $. If $ \ell_i^{\min} \ge 2 $, then $ B_{\mathcal{A}, \le}[\ell_i^{\min}] = 1 $, and if $ (\ell_{i - 1}^{\max} \not = \bot) \land (\ell_{i - 1}^{\max} \le |Q| - 1) $, then $ B_{\mathcal{A}, \le}[\ell_{i - 1}^{\max} + 1] = 1 $.
    \end{enumerate}
\end{lemma}

\begin{proof}
    \emph{(Sketch)} The first two claims follow immediately because $ \equiv_{\mathcal{A}, \le} $ is a bisimulation from $ \mathcal{A} $ to $ \mathcal{A} $. We only focus on the first part of the third claim.

    Assume that $ B_{\mathcal{A}, \le}[i] = 1 $ and $ \ell_i^{\min} \ge 2 $. We know that $ \ell_i^{\min} \not = \bot $ (which implies $ a_i^{\min} \not = \bot $), and $ B_{\mathcal{A}, \le}[\ell_i^{\min}]  $ is well defined because $ \ell_i^{\min} \ge 2 $. Suppose for the sake of a contradiction that $ B_{\mathcal{A}, \le}[\ell_i^{\min}] \not = 1 $. This means that $ u_{\ell_i^{\min} - 1} \equiv_{\mathcal{A}, \le} u_{\ell_i^{\min}} $. Since $ \equiv_{\mathcal{A}, \le} $ is a bisimulation from $ \mathcal{A} $ to $ \mathcal{A} $ and $ (u_{\ell_i^{\min}}, u_i, a_i^{\min}) \in E $, there exists $ 1 \le i' \le |Q| $ such that $ u_{i'} \equiv_{\mathcal{A}, \le} u_i $ and $ (u_{\ell_i^{\min} - 1}, u_{i'}, a_i^{\min}) \in E $. We distinguish three cases, and we show that of all them lead to a contradiction (see Figure~\ref{fig:notcrossingalgorithms}).
    \begin{itemize}
        \item Case 1: $ i < i' $. From $ (u_{\ell_i^{\min}}, u_i, a_i^{\min}), (u_{\ell_i^{\min} - 1}, u_{i'}, a_i^{\min}) \in E $ and $ i < i' $ we conclude $ \ell_i^{\min} \le \ell_i^{\min} - 1 $ by Axiom~3, which is a contradiction.
        \item Case 2: $ i' = i $. We have $ (u_{\ell_i^{\min} - 1}, u_i, a_i^{\min}) \in E $, so by the minimality of $ \ell_i^{\min} $ we conclude $ \ell_i^{\min} \le \ell_i^{\min} - 1 $, which is a contradiction.
        \item Case 3: $ i' < i $. We know that $ \equiv_{\mathcal{A}, \le} $ is $ \le $-convex (by Theorem~\ref{theor:maximumbisimulationproperties}), $ i' \le i - 1 < i $ and $ u_{i'} \equiv_{\mathcal{A}, \le} u_i $, so we obtain that $ u_{i - 1} \equiv_{\mathcal{A}, \le} u_i $ and we conclude that $ B_{\mathcal{A}, \le}[i] \not = 1 $, which is a contradiction. \qedhere
    \end{itemize} 
\end{proof}

We are now ready to compute $ B_{\mathcal{A}, \le} [2, |Q|] $. Informally, we use Algorithm~\ref{alg:computingmaximum} to implement the properties stated in Lemma~\ref{lem:ideabehindalgorithm}. At the beginning of the algorithm, we initialize $ B[2, |Q|] $ to zero, and at the end of the algorithm, we will have $ B[2, |Q|] = B_{\mathcal{A}, \le}[2, |Q|] $. We compute the bit array $ Z[2, |Q|] $ such that, for every $ 2 \le i \le |Q| $, we have $ Z[i] = 1 $ if and only if $ Out(i - 1)  \not = Out(i) $. Then, Lines~\ref{line:mainstart}-\ref{line:mainintermidate1} implement the first and the second properties of Lemma~\ref{lem:ideabehindalgorithm}, and Lines~\ref{line:mainintermediate2}-\ref{line:mainendbefore} propagate the third property of Lemma~\ref{lem:ideabehindalgorithm}. 

\begin{algorithm}
\caption{The input is a Wheeler NFA $ (\mathcal{A}, \le) $, with $ \mathcal{A} = (Q, E, s, F) $. The output is $ B_{\mathcal{A}, \le}[2, |Q|] $. The algorithm uses a queue $ K $ and two arrays $ Z[2, |Q|] $ and $ B[2, |Q|] $.}\label{alg:computingmaximum}
{
\begin{algorithmic}[1]
    \small
    \For{$ i \gets 1, 2, \dots, |Q| $}\Comment{Initialization}\label{line:initializationstart}
        \State Compute $ \ell_i^{\min} $ and $ \ell_i^{\max} $.
    \EndFor
    \State compute $ Z[2, |Q|] $ \Comment{$ Z[i] = 1 $ if $ Out(i - 1)  \not = Out(i) $ and $ Z[i] = 0 $ if $ Out(i - 1)  = Out(i) $}
\State initialize $ B[2, |Q|] $ to zero
\State initialize an empty queue $ K $ \label{line:initializationend}
    \\

    \For{$ i \gets 2, 3, \dots, |Q| $} \Comment{Main part}\label{line:mainstart}
        \If{$ (u_{i - 1} \in F \land u_i \not \in F) \lor (u_{i - 1} \not \in F \land u_i  \in F)  \lor Z[i] = 1 $}
            \State $ Enqueue(K, i) $ \label{line:next1}
            \State $ B[i] \gets 1 $ \label{line:firstB1}
        \EndIf
    \EndFor\label{line:mainintermidate1}

    \While{$ K $ is nonempty}\label{line:mainintermediate2}
        \State $ i \gets Dequeue (K) $ \label{line:dequeue}
        \If{$ (\ell_i^{\min} \ge 2) \land (B[\ell_i^{\min}] = 0) $}\label{line:recursivecheck}
            \State $ Enqueue(K, \ell_i^{\min}) $ \label{line:next2}
            \State $ B[\ell_i^{\min}] \gets 1 $ \label{line:firstB2}
        \EndIf
        \If{$ (\ell^{\max}_{i - 1} \not = \bot) \land (\ell^{\max}_{i - 1} \le |Q| - 1) \land (B[\ell^{\max}_{i - 1} + 1] = 0) $}\label{line:recursivecheckbis}
            \State $ Enqueue(K, \ell^{\max}_{i - 1} + 1) $ \label{line:next2bis}
            \State $ B[\ell^{\max}_{i - 1} + 1] \gets 1 $ \label{line:firstB2bis}
        \EndIf
    \EndWhile\label{line:mainendbefore}

    \State \Return $ B[2, |Q|] $\label{line:mainend}
    
    \end{algorithmic}
}
\end{algorithm}

Let us show that Algorithm~\ref{alg:computingmaximum} correctly computes $ B_{\mathcal{A}, \le}[2, |Q|] $.

\begin{lemma}\label{lem:algcorrectness}
    Let $ (\mathcal{A}, \le) $ be a Wheeler NFA, with $ \mathcal{A} = (Q, E, s, F) $. On input $ (\mathcal{A}, \le) $, Algorithm~\ref{alg:computingmaximum} returns the array $ B_{\mathcal{A}, \le}[2, |Q|] $. 
\end{lemma}

\begin{proof}
    \emph{(Sketch)} Fix $ 2 \le i \le |Q| $. We must prove that at the end of the algorithm we have $ B[i] = B_{\mathcal{A}, \le}[i] $. If $ B[i] = 1 $, then we intuitively obtain $ B_{\mathcal{A}, \le}[i] = 1 $ by Lemma~\ref{lem:ideabehindalgorithm}. The difficult part is to show that, if $ B[i] = 0 $, then $ B_{\mathcal{A}, \le}[i] = 0 $. In other words, we need to show that the properties stated in Lemma~\ref{lem:ideabehindalgorithm} are sufficient to detect whether $ B[i] = 1 $. Intuitively, we proceed as follows. Let $ \sim $ be the relation from $ Q $ to $ Q $ such that for every $ 1 \le i, i' \le |Q| $ we have $ u_i \sim u_{i'} $ if and only if for every $ \min \{i, i' \} + 1 \le k \le \max \{i, i' \} $ we have $ B[k] = 0 $ at the end of the algorithm. In particular, by the definition of $ \sim $, we have $ u_{i - 1} \sim u_i $ for every $ 2 \le i \le |Q| $ such that $ B[i] = 0 $ at the end of the algorithm. We only need to prove $ \sim $ is a Wheeler autobisimulation on $ (\mathcal{A}, \le) $, because then by the maximality of $ \equiv_{\mathcal{A}, \le} $ (see Theorem~\ref{theor:maximumbisimulationproperties}) we have $ u_{i - 1} \equiv_{\mathcal{A}, \le} u_i $ for every $ 2 \le i \le |Q| $ such that $ B[i] = 0 $ at the end of the algorithm, and so $ B_{\mathcal{A}, \le}[i] = 0 $ for every $ 2 \le i \le |Q| $ such that $ B[i] = 0 $ at the end of the algorithm. We prove that $ \sim $ is a Wheeler autobisimulation on $ (\mathcal{A}, \le) $ in the appendix. \qedhere
\end{proof}

Let discuss the running time of Algorithm~\ref{alg:computingmaximum} (Lemma~\ref{lem:algrunning}). We start with a remark.

\begin{remark}\label{rem:complexitystates}
Let $ (\mathcal{A}, \le) $ be a Wheeler NFA, with $ \mathcal{A} = (Q, E, s, F) $. Recall that every $ u \in Q $ is reachable in $ \mathcal{A} $, so every $ u \in Q \setminus \{s \} $ must have at least one incoming edge, which implies $ |E| \ge |Q| - 1 $. In particular, if an algorithm has time complexity $ O(|Q| + |E|) $, we can simply write that it has time complexity $ O(|E|) $.
\end{remark}

\begin{lemma}\label{lem:algrunning}
    Let $ (\mathcal{A}, \le) $ be a Wheeler NFA, with $ \mathcal{A} = (Q, E, s, F) $. Algorithm~\ref{alg:computingmaximum} can be implemented to run in $ O(|E|) $ time on input $ (\mathcal{A}, \le) $.
\end{lemma}

\begin{proof}
\emph{(Sketch)} The initialization (Lines~\ref{line:initializationstart}-\ref{line:initializationend}) can be implemented in linear time because the edge labels can be sorted in linear time. The main part of the algorithm (Lines~\ref{line:mainstart}-\ref{line:mainend}) is executed in linear time because every $ 2 \le i \le |Q| $ is added to the queue $ K $ at most once, so we only need $ O(1) $ time per state. \qedhere
\end{proof}

Theorem~\ref {theor:lineartimeminimization} implies that in linear time we can check whether two Wheeler NFAs are in the same $ \sim_{\mathfrak{C}} $-equivalence class.

\begin{corollary}\label{cor:checkingifbisimlar}
     Let $ (\mathcal{A}, \le) $ and $ (\mathcal{A}', \le') $ be  Wheeler NFAs, with $ \mathcal{A} = (Q, E, s, F) $ and $ \mathcal{A}' = (Q', E', s', F') $. We can determine whether there exists a Wheeler bisimulation from $ (\mathcal{A}, \le) $ to $ (\mathcal{A}', \le') $ in $ O(|E| + |E|') $ time.
\end{corollary}

\begin{proof}
   \emph{(Sketch)} By Theorem~\ref{theor:lineartimeminimization}, in linear time we can build the quotient Wheeler NFAs of $ (A, \le) $ and $ (\mathcal{A}', \le') $. Then, Lemma~\ref{lem:quotientautomotaisomorphic} suggests that there exists a Wheeler bisimulation from $ (A, \le) $ to $ (\mathcal{A}', \le') $ if and only if the two quotient Wheeler NFAs are isomorphic. We can check whether the two quotient Wheeler NFAs are isomorphic in linear time because (i)~if the two quotient Wheeler NFAs do not have the same number of states, they cannot be isomorphic, and (ii)~if the two quotient Wheeler NFAs have the same number of states, then by Lemma~\ref{lem:quotientautomotaisomorphic} they are isomorphic if and only if the (unique) bijection that respects the Wheeler orders of the quotient Wheeler NFAs is an isomorphism. \qedhere
\end{proof}

Let $ (\mathcal{A}, \le) $ be a Wheeler NFA, with $ \mathcal{A} = (Q, E, s, F) $. Theorem~\ref{theor:lineartimeminimization} states that we can build the quotient Wheeler NFA in linear time, under the assumption that the edge labels can be sorted in linear time. It is natural to wonder if it is possible to achieve linear time even when the edge labels cannot be sorted in linear time. In principle, this does not seem unreasonable: we know the Wheeler order $ \le $, and for every $ u, v \in Q $, if $ u < v $, then by Axiom~2 all characters labeling some edge reaching $ u $ must be smaller than all characters labeling some edge reaching $ v $. In other words, the Wheeler order $ \le $ \emph{inherently} contains information on the mutual order of the edge labels. However, from the Wheeler order we cannot retrieve any information on the mutual order of the characters labeling some edge reaching the \emph{same} state, and we will see that this prevents us from achieving linear time.

More concretely, we consider the \emph{comparison model}, in which the only query allowed to compare the elements of the alphabet $ \Sigma $ is the following: given $ a_1, a_2 \in \Sigma $, decide whether $ a_1 \preceq a_2 $. The model assumes that each such query can be solved in $ O(1) $ time. We can show that, in the comparison model, building the quotient NFA requires $ \Omega(|E| \log |E|) $ time.

\begin{theorem}\label{theor:lowerboundquotienting}
    Let $ (\mathcal{A}, \le) $ be a Wheeler NFA, with $ \mathcal{A} = (Q, E, s, F) $. In the comparison model, the time complexity of building the Wheeler NFA $ (\mathcal{A}_*, \le_*) $ of Lemma~\ref{lem:quotient} is $ \Theta(|E| \log |E|) $. This is true even if we know that $ \mathcal{A} $ is a DFA.
\end{theorem}

Theorem~\ref{theor:lowerboundquotienting} is consistent with Theorem~\ref{theor:lineartimeminimization} because it is well known that the complexity of sorting $ n $ elements in the comparison model is not linear but $ \Omega (n \log n) $ (see, e.g.,~\cite{knuth1998art, cormen2022introduction, cotumaccio2025inverting}).

\section{Conclusions and Future Work}\label{sec:conclusions}

In this paper, we have introduced Wheeler bisimulations. We have shown that they induce a unique minimal Wheeler NFA for every Wheeler language, which can be built in linear time.

Definition~\ref{def:bisimulation} refers to the most popular notion of bisimulation, often called \emph{strong} bisimulation. However, many variants appear in the literature. The properties listed in Definition~\ref{def:bisimulation} ensure that a bisimulation is propagated when following edges in a \emph{forward} fashion, but it is possible to require the same compatibility to hold when edges are followed in a \emph{backward} fashion (or both in a forward fashion and a backward fashion). In addition, several notions of \emph{weak} bisimulation have been proposed (such as stutter bisimulation and normed bisimulation), which typically still induce a unique minimal NFA. Usually, every (strong) bisimulation is a weak bisimulation, so weak bisimulations yield smaller minimal NFAs. For details, we refer the reader to~\cite{ciric2014nondeterministic, baier2008principles}. Moreover, as we have seen in the introduction, bisimulations appear in different domains, so there has been much effort to provide a unified view of bisimulation~\cite{roggenbach2000towards}, also via algebraic and co-algebraic approaches~\cite{bezhanishvili2023minimisation, bonchi2014algebra, jacobs2015trace, rot2016coalgebraic, malacaria1995studying, abramsky1991domain}. Consequently, our theory of Wheeler bisimulations can be extended in two directions: (i)~by modifying the definition of Wheeler bisimulation, and (ii)~by switching from Wheeler automata to other formalisms (which should probably capture similar ideas of convexity, see~\cite{cotumaccio2024convex} for Petri nets).

The definitions of bisimulation and Wheeler bisimulation are symmetrical: each property must hold even when the roles of the two automata are exchanged. This is not a mandatory requirement: when this symmetry is broken, one works with \emph{simulations} (and not bisimulations), see~\cite{baier2008principles}. It seems reasonable to define \emph{Wheeler simulations} by removing Property~3 in Definition~\ref{def:wheelerbisimulation} (because Property~2 and Property~3 are symmetrical).

One important point that needs further clarification is the relationship between Wheeler bisimulations and logic. It is well known that the largest bisimulation on an NFA can be characterized through several logics, notably CTL and CTL* (\cite{clarke1981design}, see also~\cite{baier2008principles}). In general, the relationship between Wheeler languages and logic is not well understood. We have seen in the introduction that Wheeler languages enjoy the typical properties of the class of regular languages: determinism and non-determinism have the same expressive power, there exists a unique minimal Wheeler DFA, and Wheeler languages admit a Myhill-Nerode theorem. Nonetheless, no logical characterization of Wheeler languages is known, while regular languages can be characterized via the monadic second-order logic~\cite{straubing2012finite}. A characterization of Wheeler languages in terms of syntactic monoids is also missing.

Consider the class of all regular languages. Bisimulations lead to the most interesting notion of uniqueness (up to isomorphism) for NFAs, but other formalisms achieve the same goal, such as residual automata~\cite{denis2002residual}, the universal automaton~\cite{lombardy2008universal}, and the átomaton~\cite{brzozowski2014} (see~\cite{maarand2024yet} for more information and more options). The natural question is whether we can also extend these results to Wheeler automata. We conjecture that this is possible at least for residual automata. Indeed, residual automata are closely related to the Myhill-Nerode theorem, and we know that there exists a Myhill-Nerode theorem for Wheeler automata~\cite{alanko2021wheeler}.

\newpage

\bibliography{lipics-v2021-sample-article}

\appendix

\section{Omitted Proofs and Results from Section~\ref{sec:minimality}}

In this section, we discuss the relationship between Wheeler automata and minimality.

\subsection{Properties of Wheeler Bisimulations}

Let $ (\mathcal{A}, \le) $ and $ (\mathcal{A}', \le') $ be Wheeler NFAs. Every Wheeler bisimulation from $ (\mathcal{A}, \le) $ to $ (\mathcal{A}', \le') $ is a bisimulation from $ \mathcal{A} $ to $ \mathcal{A}' $ by Property~1, so every Wheeler bisimulation must satisfy the properties of a standard bisimulation. Keeping this in mind, we can easily infer the following three lemmas.

The first lemma shows that Wheeler bisimulations can only relate Wheeler NFAs that recognize the same language.

\begin{lemma}[see {\cite[Theorem B.4]{kozen2007automata}}]\label{lem:bisimulationandequivalence}
    Let $ (\mathcal{A}, \le) $ and $ (\mathcal{A}', \le') $ be Wheeler NFAs, with $ \mathcal{A} = (Q, E, s, F) $ and $ \mathcal{A}' = (Q', E', s', F') $. If there exists a Wheeler bisimulation $ R \subseteq Q \times Q' $ from $ (\mathcal{A}, \le) $ to $ (\mathcal{A}', \le') $, then $ \mathcal{L}(\mathcal{A}) = \mathcal{L}(\mathcal{A}') $.
\end{lemma}

The second lemma shows that, intuitively, all states must play a role in every Wheeler bisimulation. This follows from the fact that every state is reachable from the initial state of the corresponding automaton.

\begin{lemma}[see {\cite[Lemma B.5]{kozen2007automata}}]\label{lem:R(u)nonempty}
    Let $ (\mathcal{A}, \le) $ and $ (\mathcal{A}', \le') $ be Wheeler NFAs, and let $ R \subseteq Q \times Q' $ be a Wheeler bisimulation from $ (\mathcal{A}, \le) $ to $ (\mathcal{A}', \le') $, with $ \mathcal{A} = (Q, E, s, F) $ and $ \mathcal{A}' = (Q', E', s', F') $. Then:
    \begin{enumerate}
        \item For every $ u \in Q $, $ R(u) \not = \emptyset $.
        \item For every $ u' \in Q' $, $ R^{-1}(u') \not = \emptyset $.
    \end{enumerate}
\end{lemma}

The third lemma shows that every bijective Wheeler bisimulation must be an isomorphism.

\begin{lemma}[see {\cite[Exercise 63]{kozen2007automata}}]\label{lem:bisimulationisomorphism}
    Let $ (\mathcal{A}, \le) $ and $ (\mathcal{A}', \le') $ be Wheeler NFAs, and let $ R \subseteq Q \times Q' $ be a Wheeler bisimulation from $ (\mathcal{A}, \le) $ to $ (\mathcal{A}', \le') $, with $ \mathcal{A} = (Q, E, s, F) $ and $ \mathcal{A}' = (Q', E', s', F') $. If $ R $ is a bijective function from $ Q $ to $ Q' $, then $ R $ is an isomorphism from $ \mathcal{A} $ to $ \mathcal{A}' $.
\end{lemma}

The next lemma shows that the relation $ \sim_{\mathfrak{C}} $ defined in Section~\ref{sec:mainresults} is an equivalence relation.

\begin{lemma}\label{lem:bisimequiv}
    \begin{enumerate}
        \item Let $ (\mathcal{A}, \le) $ be a Wheeler NFA, with $ \mathcal{A} = (Q, E, s, F) $. Then, $ id_Q $ is a Wheeler bisimulation from $ (\mathcal{A}, \le) $ to $ (\mathcal{A}, \le) $.
        \item Let $ (\mathcal{A}, \le) $ and $ (\mathcal{A}', \le') $ be Wheeler NFAs, with $ \mathcal{A} = (Q, E, s, F) $ and $ \mathcal{A}' = (Q', E', s', F') $. If $ R \subseteq Q \times Q' $ is a Wheeler bisimulation from $ (\mathcal{A}, \le) $ to $ (\mathcal{A}', \le') $, then $ R^{-1} \subseteq Q' \times Q $ is a Wheeler bisimulation from $ (\mathcal{A}', \le') $ to $ (\mathcal{A}, \le) $.
        \item Let $ (\mathcal{A}, \le) $, $ (\mathcal{A}', \le') $ and $ (\mathcal{A}'', \le'') $ be Wheeler NFAs, with $ \mathcal{A} = (Q, E, s, F) $, $ \mathcal{A}' = (Q', E', s', F') $ and $ \mathcal{A}'' = (Q'', E'', s'', F'') $. If $ R \subseteq Q \times Q' $ is a Wheeler bisimulation from $ (\mathcal{A}, \le) $ to $ (\mathcal{A}', \le') $ and $ R' \subseteq Q' \times Q'' $ is a Wheeler bisimulation from $ (\mathcal{A}', \le') $ to $ (\mathcal{A}'', \le'') $, then $ R' \circ R \subseteq Q \times Q'' $ is a Wheeler bisimulation from $ (\mathcal{A}, \le) $ to $ (\mathcal{A}'', \le'') $.
    \end{enumerate}
\end{lemma}

\begin{proof}
    The first statement is immediate, and the second statement is also immediate because $ (R^{-1})^{-1} = R $ (which ensures that $ R^{-1} $ satisfies Property~3 of Definition~\ref{def:wheelerbisimulation}). Hence, we focus on the third statement. Let us prove that $ R' \circ R $ is a Wheeler bisimulation from $ (\mathcal{A}, \le) $ to $ (\mathcal{A}'', \le'') $.

    Let us prove Property~1. We must prove that $ R' \circ R $ is a bisimulation from $ \mathcal{A} $ to $ \mathcal{A}'' $. In particular, $ R $ is a bisimulation from $ \mathcal{A} $ to $ \mathcal{A}' $ and $ R' $ is a bisimulation from $ \mathcal{A}' $ to $ \mathcal{A}'' $, so the conclusion follows from~\cite[Lemma B.2]{kozen2007automata}
    
    Let us prove Property~2. Fix $ C \subseteq Q $ such that $ C $ is $ \le $-convex. We must prove that $ (R' \circ R) (C) $ is $ \le'' $-convex. Since $ (R' \circ R) (C) = R'(R(C)) $, we need to prove that $ R'(R(C)) $ is $ \le'' $-convex. We know that $ R $ is a Wheeler bisimulation from $ (\mathcal{A}, \le) $ to $ (\mathcal{A}', \le') $ and $ C $ is $ \le $-convex, so we obtain that $ R(C) $ is $ \le' $-convex. We also know that $ R' $ is a Wheeler bisimulation from $ (\mathcal{A}', \le') $ to $ (\mathcal{A}'', \le'') $ and $ R(C) $ is $ \le' $-convex, so we conclude that $ R'(R(C)) $ is $ \le'' $-convex. 
    
    Let us prove Property~3. Fix $ C'' \subseteq Q'' $ such that $ C'' $ is $ \le'' $-convex. We must prove that $ ((R' \circ R)^{-1}) (C'') $ is $ \le $-convex. Since $ ((R' \circ R)^{-1}) (C'') = (R^{-1} \circ (R')^{-1})(C'') = R^{-1} ((R')^{-1}(C'')) $, we need to prove that $ R^{-1} ((R')^{-1}(C'')) $ is $ \le $-convex. The conclusion follows by arguing as we did for Property~2. \qedhere
\end{proof}

\subsection{Wheeler Autobisimulations}

The following lemma shows that we can canonically build Wheeler \emph{autobisimulations} starting from a Wheeler \emph{bisimulation}.

\begin{lemma}\label{lem:compositionisautobisimulation}
    Let $ (\mathcal{A}, \le) $ and $ (\mathcal{A}', \le') $ be Wheeler NFAs, with $ \mathcal{A} = (Q, E, s, F) $ and $ \mathcal{A}' = (Q', E', s', F') $, and let $ R \subseteq Q \times Q' $ a Wheeler bisimulation from $ (\mathcal{A}, \le) $ to $ (\mathcal{A}', \le') $. Then:
    \begin{enumerate}
        \item $ R^{-1} \circ R $ is a Wheeler autobisimulation on $ (\mathcal{A}, \le) $.
        \item $ R \circ R^{-1} $ is a Wheeler autobisimulation on $ (\mathcal{A}', \le') $.
    \end{enumerate}
\end{lemma}

\begin{proof}
    We only prove the first statement because the proof of the second statement is analogous. By Lemma~\ref{lem:bisimequiv}, we know that $ R^{-1} $ is a Wheeler bisimulation from $ (\mathcal{A}', \le') $ to $ (\mathcal{A}, \le) $ and $ R^{-1} \circ R $ is a Wheeler bisimulation from $ (\mathcal{A}, \le) $ to $ (\mathcal{A}, \le) $. To conclude the proof, we only need to show that $ R^{-1} \circ R $ is reflexive. To this end, we need to prove that $ u \; (R^{-1} \circ R) \; u $ for every $ u \in Q $. Fix $ u \in Q $. By Lemma~\ref{lem:R(u)nonempty} we have $ R(u) \not = \emptyset $, so there exists $ u' \in Q $ such that $ u \; R \; u' $. Hence, we have $ u' \; R^{-1} \; u $ and we conclude $ u \; (R^{-1} \circ R) \; u $. \qedhere
\end{proof}

Since autobisimulations are reflexive by definition, we can use Lemma~\ref{lem:convexunion} to show that the union of autobisimulations must be an autobisimulation.

\begin{lemma}\label{lem:unionautobisimulation}
     Let $ (\mathcal{A}, \le) $ be a Wheeler NFA, with $ \mathcal{A} = (Q, E, s, F) $. Let $ R_1, R_2 \subseteq Q \times Q $ be Wheeler autobisimulations on $ (\mathcal{A}, \le) $. Then, $ R_1 \cup R_2 $ is a Wheeler autobisimulation on $ (\mathcal{A}, \le) $.
\end{lemma}

\begin{proof}
    We know that $ R_1 $ and $ R_2 $ are Wheeler autobisimulations on $ \mathcal{A} $, so $ R_1 $ and $ R_2 $ are reflexive Wheeler bisimulations from $ (\mathcal{A}, \le) $ to $ (\mathcal{A}, \le) $. Since $ R_1 $ is reflexive, we have $ id_Q \subseteq R_1 $, so $ id_Q \subseteq R_1 \cup R_2 $ and we conclude that $ R_1 \cup R_2 $ is reflexive. Consequently, we only need to prove that $ R_1 \cup R_2 $ is a Wheeler bisimulation from $ (\mathcal{A}, \le) $ to $ (\mathcal{A}, \le) $.

    Let us prove Property~1. We must prove that $ R_1 \cup R_2 $ is a bisimulation from $ \mathcal{A} $ to $ \mathcal{A}' $. In particular, $ R_1 $ and $ R_2 $ are bisimulations from $ \mathcal{A} $ to $ \mathcal{A}' $, so the conclusion follows from~\cite[Lemma B.2]{kozen2007automata}

    Let us prove Property~2. Let $ C \subseteq Q $ such that $ C $ is $ \le $-convex. We must prove that $ (R_1 \cup R_2)(C) $ is $ \le $-convex. Since $ (R_1 \cup R_2)(C) = R_1(C) \cup R_2 (C) $, we need to prove that $ R_1(C) \cup R_2 (C) $ is $ \le $-convex. If $ C = \emptyset $, then $ R_1(C) \cup R_2 (C) = \emptyset $ and we are done, so in the rest of the proof we can assume $ C \not = \emptyset $. Since $ R_1 $ and $ R_2 $ are reflexive, we have $ \emptyset \subsetneqq C \subseteq R_1(C) \cap R_2 (C) $, so by Lemma~\ref{lem:convexunion} we only need to prove that $ R_1(C) $ and $ R_2 (C) $ are $ \le $-convex. This follows from the fact that $ R_1 $ and $ R_2 $ are Wheeler bisimulations from $ \mathcal{A} $ to $ \mathcal{A} $.

     Let us prove Property~3. Let $ C' \subseteq Q $ such that $ C' $ is $ \le $-convex. We must prove that $ (R_1 \cup R_2)^{-1}(C') $ is $ \le $-convex. Since $ (R_1 \cup R_2)^{-1}(C') = R_1^{-1}(C') \cup R_2^{-1} (C') $, we need to prove that $ R_1^{-1}(C') \cup R_2^{-1} (C') $ is $ \le $-convex. The conclusion follows by arguing as we did for Property~2. \qedhere
\end{proof}

In the proof of Theorem~\ref{theor:maximumbisimulationproperties} we will use the following remark.

\begin{remark}\label{rem:exautobisimulation}
    Let $ (\mathcal{A}, \le) $ be a Wheeler NFA, with $ \mathcal{A} = (Q, E, s, F) $. From Lemma~\ref{lem:bisimequiv} we obtain the following: (i)~$ id_Q $ is a Wheeler autobisimulation on $ (\mathcal{A}, \le) $; (ii)~if $ R \subseteq Q \times Q $ is a Wheeler autobisimulation on $ (\mathcal{A}, \le) $, then $ R^{-1} $ is a Wheeler autobisimulation on $ (\mathcal{A}, \le) $; (iii)~if $ R, R' \subseteq Q \times Q $ are Wheeler autobisimulations on $ (\mathcal{A}, \le) $, then $ R' \circ R $ is a Wheeler autobisimulation on $ (\mathcal{A}, \le) $.
\end{remark}

\noindent{\textbf{Statement of Theorem~\ref{theor:maximumbisimulationproperties}.}
    Let $ (\mathcal{A}, \le) $ be a Wheeler NFA, with $ \mathcal{A} = (Q, E, s, F) $. Then, there exists a (unique) maximum Wheeler autobisimulation $ \equiv_{\mathcal{A}, \le} $ on $ (\mathcal{A}, \le) $, that is, a Wheeler autobisimulation $ \equiv_{\mathcal{A}, \le} $ on $ (\mathcal{A}, \le) $ that contains every Wheeler autobisimulation on $ (\mathcal{A}, \le) $. Moreover, $ \equiv_{\mathcal{A}, \le} $ is a $ \le $-convex equivalence relation on $ Q $.

\begin{proof}
    To prove the existence of a (unique) maximum Wheeler autobisimulation $ \equiv_{\mathcal{A}, \le} $ on $ (\mathcal{A}, \le) $, we will proceed analogously to the corresponding result for standard bisimulations~\cite{kozen2007automata}. Let $ \mathcal{S} $ be the set of all Wheeler autobisimulations on $ (\mathcal{A}, \le) $. Notice that $ \mathcal{S} $ is (i)~nonempty because $ id_Q \in \mathcal{S} $ by Remark~\ref{rem:exautobisimulation} and (ii)~finite because $ Q $ is finite and every element of $ \mathcal{S} $ is a subset of $ Q \times Q $. Define $ \equiv_{\mathcal{A}, \le} $ as $\bigcup_{R \in \mathcal{S}} R $. Since $ \mathcal{S} $ is finite and nonempty, by Lemma~\ref{lem:unionautobisimulation} we obtain that $ \equiv_{\mathcal{A}, \le} $ is a Wheeler autobisimulation on $ (\mathcal{A}, \le) $. Moreover, for every Wheeler autobisulation $ R $ on $ \mathcal{A} $ we have $ R \in \mathcal{S} $, so $ R $ is a subset of $ \equiv_{\mathcal{A}, \le} $ and we conclude that $ \equiv_{\mathcal{A}, \le} $ is the (unique) maximum Wheeler bisimulation on $ (\mathcal{A}, \le) $.

    Let us prove that $ \equiv_{\mathcal{A}, \le} $ is an equivalence relation on $ Q $. First, we know that $ \equiv_{\mathcal{A}, \le} $ is reflexive because it is an autobisimulation. Second, let us prove that $ \equiv_{\mathcal{A}, \le} $ is symmetrical. Assume that $ u \equiv_{\mathcal{A}, \le} v $ for some $ u, v \in Q $. We must prove that $ v \equiv_{\mathcal{A}, \le} u $. By the definition of $ \equiv_{\mathcal{A}, \le} $, there exists $ R \in \mathcal{S} $ such that $ u \; R \; v $, and so $ v \; R^{-1} \; u $. Since $ R \in \mathcal{S} $, by Remark~\ref{rem:exautobisimulation} we have $ R^{-1} \in \mathcal{S} $, so by the definition of $ \equiv_{\mathcal{A}, \le} $ we conclude $ v \equiv_{\mathcal{A}, \le} u $. Third, let us prove that $ \equiv_{\mathcal{A}, \le} $ is transitive. Assume that $ u \equiv_{\mathcal{A}, \le} v $ and $ v \equiv_{\mathcal{A}, \le} w $ for some $ u, v, w \in Q $. We must prove that $ u \equiv_{\mathcal{A}, \le} w $. By the definition of $ \equiv_{\mathcal{A}, \le} $, there exists $ R, R' \in \mathcal{S} $ such that $ u \; R \; v $ and $ v \; R' \; w $, and so $ u \; (R' \circ R) \; z $. Since $ R, R' \in \mathcal{S} $, by Remark~\ref{rem:exautobisimulation} we have $ R' \circ R \in \mathcal{S} $, so by the definition of $ \equiv_{\mathcal{A}, \le} $ we conclude $ u \equiv_{\mathcal{A}, \le} w $.

    Let us prove that the equivalence relation $ \equiv_{\mathcal{A}, \le} $ is $ \le $-convex. Fix $ u \in Q $. We must prove that $ [u]_{ \equiv_{\mathcal{A}, \le}} $ is $ \le $-convex. Let $ v, w, x \in Q $ such that $ v \le w \le x $ and $ v, x \in [u]_{\equiv_{\mathcal{A}, \le}} $. We must prove that $ w \in [u]_{ \equiv_{\mathcal{A}, \le}} $. Since $ v, x \in [u]_{\equiv_{\mathcal{A}, \le}} $, we have $ v, x \in \{u' \in Q \;|\; u \equiv_{\mathcal{A}, \le} u' \} $. We know that $ \equiv_{\mathcal{A}, \le} $ is a Wheeler autobisimulation on $ (\mathcal{A}, \le) $, so the set $ \{u' \in Q \;|\; u \equiv_{\mathcal{A}, \le} u' \} $ is $ \le $-convex by Property~2. Since $ v \le w \le x $ and $ v, x \in \{u' \in Q \;|\; u \equiv_{\mathcal{A}, \le} u' \} $, we obtain that $ w \in \{u' \in Q \;|\; u \equiv_{\mathcal{A}, \le} u' \} $, so $ u \equiv_{\mathcal{A}, \le} w $ and we conclude that $ w \in [u]_{\equiv_{\mathcal{A}, \le}} $. \qedhere    
\end{proof}

\subsection{The Quotient Wheeler NFA}

Let $ (\mathcal{A}, \le) $ be a Wheeler NFA. We now show how to build a quotient Wheeler NFA $ (\mathcal{A}_*, \le_*) $.

\medskip

\noindent{\textbf{Statement of Lemma~\ref{lem:quotient}.}
    Let $ (\mathcal{A}, \le) $ be a Wheeler NFA, with $ \mathcal{A} = (Q, E, s, F) $. Let $ Q_* = \mathcal{P}_{\equiv_{\mathcal{A}, \le}} $, and let $ \le_* $ be the relation from $ Q_* $ to $ Q_* $ such that, for every $ U, V \in Q_* $, we have $ U \le_* V $ if and only if there exist $ u \in U $ and $ v \in V $ such that $ u \le v $.
    \begin{enumerate}
        \item Let $ U, V \in Q_* $ be such that $ U <_* V $. Then, for every $ u \in U $ and for every $ v \in V $ we have $ u < v $.
        \item $ \le_* $ is a total order on $ Q_* $.
        \item Let $ \mathcal{A}_* = (Q_*, E_*, s_*, F_*) $, where:
        \begin{itemize}
            \item $ E_* = \{(U, V, a) \in Q_* \times Q_* \times \Sigma \;|\; (\exists u \in U)(\exists v \in V)((u, v, a) \in E) \} $.
            \item $ s_* = [s]_{\equiv_{\mathcal{A}, \le}} $.
            \item $ F_* = \{U \in Q_* \;|\; (\exists u \in U)(u \in F) \} $.
        \end{itemize}
        Then, $ (\mathcal{A}_*, \le_*) $ is a Wheeler NFA.
        \item If $ \mathcal{A} $ is a DFA, then $ \mathcal{A}_* $ is a DFA.
    \end{enumerate}

\begin{proof}
    \begin{enumerate}
    \item Fix $ u \in U $ and $ v \in V $. We must prove that $ u < v $. We know that $ U <_* V $ or, equivalently, $ (U \le_* V) \land (U \not = V) $. From $ U \le_* V $ we obtain that there exist $ u' \in U $ and $ v' \in V $ such that $ u' \le v' $. By Theorem~\ref{theor:maximumbisimulationproperties} $ \equiv_{\mathcal{A}, \le} $ is $ \le $-convex, so $ U $ and $ V $ are $ \le $-convex. 

    Let us prove that $ u < v' $. Suppose for the sake of contradiction that $ v' \le u $. From  $ u' \le v' \le u $ and $ u, u' \in U $ we obtain $ v' \in U $ because $ U $ is $ \le $-convex, so $ U = [v']_{\equiv_{\mathcal{A}, \le}} = V $, which is a contradiction.

    Let us prove that $ u < v $. Suppose for the sake of contradiction that $ v \le u $. From $ v \le u < v' $ and $ v, v' \in V $ we obtain $ u \in V $ because $ V $ is $ \le $-convex, so $ U = [u]_{\equiv_{\mathcal{A}, \le}} = V $, which is a contradiction.
    
    \item Let us prove that $ \le_* $ is reflexive. For every $ U \in Q_* $ we have $ U \le_* U $ because if we pick any $ u \in U $ we have $ u \le u $.
    
    Let us prove that $ \le_* $ is antisymmetric. Assume that for some $ U, V \in Q_* $ we have $ U <_* V $ (or equivalently, $ (U \le_* V) \land (U \not = V) $). We need to prove that we have $ V \not \le_* U $. To this end, we only need to prove that for every $ u \in U $ and for every $ v \in V $ we have $ v \not \le u $. Fix $ u \in U $ and $ v \in V $. From $ U \le_* V $ and point 1 we obtain $ u \le v $, and from $ U \not = V $ we obtain $ u \not = v $. By the antisymmetry of $ \le $ we conclude $ v \not \le u $.
    
    Let us prove that $ \le_* $ is transitive. Assume that for some $ U, V, W \in Q_* $ we have $ U \le_* V $ and $ V \le_* W $. We must prove that $ U \le_* W $. If $ U = V $ the conclusion follows from $ V \le_* W $, and if $ V = W $ the conclusion follows from $ U \le_* V $, so we can assume $ U \not = V $ and $ V \not = W $. Pick $ u \in U $, $ v \in V $ and $ w \in W $. By point 1, we have $ u \le v $ and $ v \le w $, so we obtain $ u \le w $ by the transitivity of $ \le $ and we conclude $ U \le_* W $ by the definition of $ \le_* $.
    
    Let us prove that $ \le_* $ is total. Pick $ U, V \in Q_* $. We must prove that $ (U \le_* V) \lor (V \le_* U) $. Pick $ u \in U $ and $ v \in V $. We have $ (u \le v) \lor (v \le u) $ because $ \le $ is total, so we conclude $ (U \le_* V) \lor (V \le_* U) $.
    
    \item First, let us prove that $ \mathcal{A}_* $ is a well-defined NFA. Note that $ Q_* $ is finite (because $ Q $ is finite) and $ [s]_{\equiv_{\mathcal{A}, \le}} \in \mathcal{P}_{\equiv_{\mathcal{A}, \le}} $ (because $ s \in Q $). We only need to prove that every $ U \in Q^* $ is reachable in $ \mathcal{A}_* $ and co-reachable in $ \mathcal{A}_* $. Fix $ U \in Q $.
    
    \begin{itemize}
        \item Let us prove that $ U $ is reachable in $ \mathcal{A}_* $. Pick $ u \in U $. Then, $ u $ is reachable in $ \mathcal{A} $, so there exist $ t \ge 1 $, $ u_1, u_2, \dots, u_t \in Q $ and $ a_1, a_2, \dots, a_{t - 1} \in \Sigma $ such that (i)~$ u_1 = s $, (ii)~$ u_t = u $ and (iii)~$ (u_i, u_{i + 1}, a_i) \in E $ for $ 1 \le i \le t - 1 $. Define $ U_i = [u_i]_{\equiv_{\mathcal{A}, \le}} $ for $ 1 \le i \le t $. Then, we have (i)~$ U_1 = s^* $, (ii)~$ U_t = U $ and (iii)~$ (U_i, U_{i + 1}, a_i) \in E_* $ for $ 1 \le i \le t - 1 $, so $ U $ is reachable in $ \mathcal{A}_* $.
        \item Let us prove that $ U $ is co-reachable in $ \mathcal{A}_* $. Pick $ u \in U $. Then, $ u $ is co-reachable in $ \mathcal{A} $, so there exist $ t \ge 1 $, $ u_1, u_2, \dots, u_t \in Q $ and $ a_1, a_2, \dots, a_{t - 1} \in \Sigma $ such that (i)~$ u_1 = u $, (ii)~$ u_t \in F $ and (iii)~$ (u_i, u_{i + 1}, a_i) \in E $ for $ 1 \le i \le t - 1 $. Define $ U_i = [u_i]_{\equiv_{\mathcal{A}, \le}} $ for $ 1 \le i \le t $. Then, we have (i)~$ U_1 = U $, (ii)~$ U_t \in F_* $ and (iii)~$ (U_i, U_{i + 1}, a_i) \in E_* $ for $ 1 \le i \le t - 1 $, so $ U $ is co-reachable in $ \mathcal{A}_* $.
    \end{itemize}

    Now, let us prove that $ (\mathcal{A}_*, \le_*) $ is a Wheeler NFA. By point 2, $ \le_* $ is a total order on $ Q_* $, so we only need to prove Axiom~1, Axiom~2 and Axiom~3.

    \begin{itemize}
        \item Let us prove Axiom~1. Let $ U \in Q_* $. We must prove that $ S \le_* U $. Pick any $ u \in U $. We have $ s \in S $, and we have $ s \le u $ because $ \le $ is a Wheeler order on $ \mathcal{A} $, so we conclude $ S \le_* U $.
        \item Let us prove Axiom~2. Let $ U, V, U', V' \in Q_* $ and $ a, b \in \Sigma $ be such that we have $ (U, V, a), (U', V', b) \in E_* $ and $ V <_* V' $. We must prove that $ a \preceq b $. By the definition of $ E_* $, there exist $ u \in U $, $ v \in V $, $ u' \in U' $ and $ v' \in V' $ such that $ (u, v, a), (u', v', b) \in E $. We cannot have $ v' \le v $ because we would conclude $ V' \le_* V $, which is a contradiction because $ V <_* V' $ and $ \le_* $ is a total order by point 2. Since $ \le $ is also a total order, we obtain $ v < v' $. We know that $ \le $ is a Wheeler order on $ \mathcal{A} $, so from $ (u, v, a), (u', v', b) \in E $ and $ v < v' $ we conclude $ a \preceq b $.
        \item Let us prove Axiom~3. Let $ U, V, U', V' \in Q_* $ and $ a \in \Sigma $ be such that we have $ (U, V, a), (U', V', a) \in E_* $ and $ V <_* V' $. We must prove that $ U \le_* U' $. By the definition of $ E_* $, there exist $ u \in U $, $ v \in V $, $ u' \in U' $ and $ v' \in V' $ such that $ (u, v, a), (u', v', a) \in E $. As in the previous point, we obtain $ v < v' $. We know that $ \le $ is a Wheeler order on $ \mathcal{A} $, so from $ (u, v, a), (u', v', a) \in E $ and $ v < v' $ we conclude $ u \le u' $, which implies $ U \le_* U' $.
    \end{itemize}

    \item Let $ U, V_1, V_2 \in Q_* $ and $ a \in \Sigma $ such that $ (U, V_1, a), (U, V_2, a) \in E_* $. We must prove that $ V_1 = V_2 $. By the definition of $ E_* $, there exist $ u_1, u_2 \in U $, $ v_1 \in V_1 $ and $ v_2 \in V_2 $ such that $ (u_1, v_1, a), (u_2, v_2, a) \in E $. From $ u_1, u_2 \in U $ we obtain $ u_1 \equiv_{\mathcal{A}, \le} u_2 $. We know that $ \equiv_{\mathcal{A}, \le} $ is a bisimulation from $ \mathcal{A} $ to $ \mathcal{A} $, so from $ u_1 \equiv_{\mathcal{A}, \le} u_2 $ and $ (u_1, v_1, a) \in E $ we obtain that there exists $ \overline{v}_2 \in Q $ such that $ v_1 \equiv_{\mathcal{A}, \le} \overline{v}_2 $ and $ (u_2, \overline{v}_2, a) \in E $. Since $ \mathcal{A} $ is a DFA, from $ (u_2, v_2, a) \in E $ and $ (u_2, \overline{v}_2, a) \in E $ we obtain $ v_2 = \overline{v}_2 $, so $ v_1 \equiv_{\mathcal{A}, \le} v_2 $. From $ v_1 \in V_1 $, $ v_2 \in V_2 $ and $ v_1 \equiv_{\mathcal{A}, \le} v_2 $ we conclude $ V_1 = V_2 $. \qedhere
\end{enumerate}    
\end{proof}

The following lemma shows that, if we quotient a Wheeler NFA, the resulting Wheeler NFA is in the same $ \sim_{\mathfrak{C}} $-equivalence class.

\begin{lemma}\label{lem:bisimulationtoquotient}
    Let $ (\mathcal{A}, \le) $ be a Wheeler NFA, with $ \mathcal{A} = (Q, E, s, F) $. Consider the Wheeler NFA $ (\mathcal{A}_*, \le_*) $ of Lemma~\ref{lem:quotient}, with $ \mathcal{A}_* = (Q_*, E_*, s_*, F_*) $. Let $ R_{\mathcal{A}, \le} \subseteq Q \times Q_* $ be the relation defined by the pairs $ u \; R_{\mathcal{A}, \le} \; [u]_{\equiv_{\mathcal{A}, \le}} $, for every $ u \in Q $. Then, $ R_{\mathcal{A}, \le} $ is a Wheeler bisimulation from $ (\mathcal{A}, \le) $ to $ (\mathcal{A}_*, \le_*) $.
\end{lemma}

\begin{proof}
    Let us prove that $ R_{\mathcal{A}, \le} $ is a Wheeler bisimulation from $ (\mathcal{A}, \le) $ to $ (\mathcal{A}_*, \le_*) $.

    Let us prove Property~1. We must prove that $ R_{\mathcal{A}, \le} $ is a bisimulation from $ \mathcal{A} $ to $ \mathcal{A}_* $. We will argue similarly to~\cite[Lemma B.9]{kozen2007automata}.
    
    \begin{itemize}
        \item Let $ u, v \in Q $, $ U \in Q_* $ and $ a \in \Sigma $ such that $ u \; R_{\mathcal{A}, \le} \; U $ and $ (u, v, a) \in E $. We must prove that there exists $ V \in Q_* $ such that $ v \; R_{\mathcal{A}, \le} \; V $ and $ (U, V, a) \in E_* $. Since $ u \; R_{\mathcal{A}, \le} \; U $, by the definition of $ R_{\mathcal{A}, \le} $ we have $ u \in U $. Let $ V = [v]_{\equiv_{\mathcal{A}, \le}} $. By the definition of $ R_{\mathcal{A}, \le} $ we have $ v \; R_{\mathcal{A}, \le} \; V $. Moreover, from $ (u, v, a) \in E $, $ u \in U $ and $ v \in V $ we conclude that $ (U, V, a) \in E_* $ by the definition of $ E_* $.
        \item Let $ u \in Q $, $ U, V \in Q_* $ and $ a \in \Sigma $ such that $ u \; R_{\mathcal{A}, \le} \; U $ and $ (U, V, a) \in E_* $. We must prove that there exists $ v \in Q $ such that $ v \; R_{\mathcal{A}, \le} \; V $ and $ (u, v, a) \in E $. Since $ u \; R_{\mathcal{A}, \le} \; U $, by the definition of $ R_{\mathcal{A}, \le} $ we have $ u \in U $. Since $ (U, V, a) \in E_* $, by the definition of $ E_* $ there exist $ u_1 \in U $ and $ v_1 \in V $ such that $ (u_1, v_1, a) \in E $. From $ u, u_1 \in U $ we obtain $ u \equiv_{\mathcal{A}, \le} u_1 $. We know that $ \equiv_{\mathcal{A}, \le} $ is a bisimulation from $ \mathcal{A} $ to $ \mathcal{A} $, so from $ u \equiv_{\mathcal{A}, \le} u_1 $ and $ (u_1, v_1, a) \in E $ we obtain that there exists $ v \in Q $ such that $ v \equiv_{\mathcal{A}, \le} v_1 $ and $ (u, v, a) \in E $. From $ v \equiv_{\mathcal{A}, \le} v_1 $ and $ v_1 \in V $ we obtain $ v \in V $. By the definition of $ R_{\mathcal{A}, \le} $, we have $ v \; R_{\mathcal{A}, \le} \; V $.
        \item We have $ s \; R_{\mathcal{A}, \le} \; s_* $ by the definition of $ R_{\mathcal{A}, \le} $ because by definition $ s_* = [s]_{\equiv_{\mathcal{A}, \le}} $.
        \item Let $ u \in Q $ and $ U \in U_* $ such that $ u \; R_{\mathcal{A}, \le} \; U $. We must prove that $ u \in F $ if and only if $ U \in F_* $. By the definition of $ R_{\mathcal{A}, \le} $ we have $ u \in U $. This means that, if $ u \in F $, then $ U \in F_* $ by the definition of $ F_* $. Conversely, assume that $ U \in F_* $. By the definition of $ F_* $, there exists $ u_1 \in U $ such that $ u_1 \in F $. From $ u_1, u \in U $ we obtain $ u_1 \equiv_{\mathcal{A}, \le} u $. We know that $ \equiv_{\mathcal{A}, \le} $ is a bisimulation from $ \mathcal{A} $ to $ \mathcal{A} $, so from $ u_1 \equiv_{\mathcal{A}, \le} u $ and $ u_1 \in F $ we conclude $ u \in F $.
    \end{itemize}
    
    Let us prove Property~2. Assume that $ C \subseteq Q $ is $ \le $-convex. We must prove that $ R_{\mathcal{A}, \le}(C) $ is $ \le_* $-convex. Let $ U, V, W \in Q_* $ such that $ U \le_* V \le_* W $ and $ U, W \in R_{\mathcal{A}, \le}(C) $. We must prove that $ V \in R_{\mathcal{A}, \le}(C) $. If $ (V = U) \lor (V = W) $ the conclusion is immediate, so we can assume $ U \not = V \not = W $. From $ U, W \in R_{\mathcal{A}, \le}(C) $ we obtain that there exist $ u, w \in C $ such that $ u \; R_{\mathcal{A}, \le} \; U $ and $ w \; R_{\mathcal{A}, \le} \; W $. From the definition of $ R_{\mathcal{A}, \le} $ we obtain $ u \in U $ and $ w \in W $. Pick $ v \in V $. Notice that we cannot have $ v < u $; otherwise, we would obtain $ V \le_* U $ by the definition of $ \le_* $, so from $ U \le_* V $ and $ V \le_* U $ we would conclude $ U = V $ because $ \le_* $ is a total order, which contradicts $ U \not = V $. Since $ \le $ is a total order, we must have $ u \le v $. Analogously, one shows that $ v \le w $. We know that $ C $ is $ \le $-convex, so from $ u \le v \le w $ and $ u, w \in C $ we obtain $ v \in C $. By the definition of $ R_{\mathcal{A}, \le} $ we have $ v \; R_{\mathcal{A}, \le} \; V $, so we conclude that $ V \in R_{\mathcal{A}, \le}(C) $.

    Let us prove Property~3. Assume that $ C_* \subseteq Q_* $ is $ \le_* $-convex. We must prove that $ R_{\mathcal{A}, \le}^{-1}(C_*) $ is $ \le $-convex. Let $ u, v, w \in Q $ such that $ u \le v \le w $ and $ u, w \in R_{\mathcal{A}, \le}^{-1}(C_*) $. We must prove that $ v \in R_{\mathcal{A}, \le}^{-1}(C_*) $. From $ u, w \in R_{\mathcal{A}, \le}^{-1}(C_*) $ we obtain that there exist $ U, W \in C_* $ such that $ u \; R_{\mathcal{A}, \le} \; U $ and $ w \; R_{\mathcal{A}, \le} \; W $. From the definition of $ R_{\mathcal{A}, \le} $ we obtain $ u \in U $ and $ w \in W $. Let $ V = [v]_{\equiv_{\mathcal{A}, \le}} $. From $ u \le v \le w $ we obtain $ U \le_* V \le_* W $ by the definition of $ \le_* $. We know that $ C_* $ is $ \le_* $-convex, so from $ U \le_* V \le_* W $ and $ U, W \in C_* $ we obtain $ V \in C_* $. By the definition of $ R_{\mathcal{A}, \le} $ we have $ v \; R_{\mathcal{A}, \le} \; V $, so we conclude that $ v \in R_{\mathcal{A}, \le}^{-1}(C_*) $. \qedhere
\end{proof}

In the next lemma, we show a crucial result: the identity relation is the \emph{unique} Wheeler autobisimulation on a quotient automaton. Intuitively, this is true by the maximality of $ \equiv_{\mathcal{A}, \le} $ because the quotient automaton of $ (\mathcal{A}, \le) $ is defined starting from $ \equiv_{\mathcal{A}, \le} $ (see Lemma~\ref{lem:quotient}).

\begin{lemma}\label{lem:uniquewheelerautobisimulation}
    Let $ (\mathcal{A}, \le) $ be a Wheeler NFA. Consider the Wheeler NFA $ (\mathcal{A}_*, \le_*) $ of Lemma~\ref{lem:quotient}, where $ \mathcal{A}_* = (Q_*, E_*, s_*, F_*) $. Then, $ id_{Q_*} $ is the unique Wheeler autobisimulation on $ (\mathcal{A}_*, \le_*) $. In particular, $ \equiv_{\mathcal{A}_*, \le_*} $ is equal to $ id_{Q_*} $.
\end{lemma}

\begin{proof}
    By Remark~\ref{rem:exautobisimulation}, $ id_{Q_*} $ is a Wheeler autobisimulation on $ (\mathcal{A}_*, \le_*) $. Now, let $ R \subseteq Q_* \times Q_* $ be any Wheeler autobisimulation on $ (\mathcal{A}_*, \le_*) $. We must prove that $ R = id_{Q_*} $. Since $ R $ is an autobisimulation, $ R $ is reflexive, so we have $ id_{Q_*} \subseteq R $. We are only left with proving that $ R \subseteq id_{Q_*} $. Assume that $ U, V \in Q_* $ satisfy $ U \; R \; V $. We must prove that $ U = V $. Pick $ u \in U $ and $ v \in V $, and consider the relation $ R' = R_{\mathcal{A}, \le}^{-1} \circ R \circ R_{\mathcal{A}, \le} \subseteq Q \times Q $, where $ R_{\mathcal{A}, \le} \subseteq Q \times Q_* $ is defined in Lemma~\ref{lem:bisimulationtoquotient}. By the definition of $ R_{\mathcal{A}, \le} $, we have $ u \; R_{\mathcal{A}, \le} \; U $ and $ V \; R_{\mathcal{A}, \le}^{-1} \;v $, so $ u \; R' v $. By Lemma~\ref{lem:bisimulationtoquotient} and Lemma~\ref{lem:bisimequiv}, we know that $ R' $ is a bisimulation from $ (\mathcal{A}, \le) $ to $ (\mathcal{A}, \le) $. Moreover, for every $ w \in Q $ we have $ w \; R' \; w $ because $ w \; R_{\mathcal{A}, \le} \; W $, $ W \; R \; W $ and $ W \; R_{\mathcal{A}, \le}^{-1} w $, where $ W = [w]_{{\equiv_{\mathcal{A}, \le}}} $. This means that $ id_Q \subseteq R' $, so $ R' $ is an autobisimulation on $ (\mathcal{A}, \le) $. Since $ u \; R' \; v $, by the maximality of $ \equiv_{\mathcal{A}, \le} $ we obtain $ u \; \equiv_{\mathcal{A}, \le} \; v $, so from $ u \in U $ and $ v \in V $ we conclude $ U = V $. \qedhere
\end{proof}

\subsection{Combining the Previous Results}

We now have all the main ingredients required to prove Theorem~\ref{theorem:minimalbisimilar} and Corollary~\ref{cor:quotientisstateminimal}.

The next lemma shows that quotienting a quotient Wheeler NFA does not return a smaller Wheeler NFA.

\begin{lemma}\label{lem:stateminimalisomorphism}
    Consider a $ \sim_{\mathfrak{C}} $-equivalence class $ C $, and let $ (\mathcal{A}, \le) $ be a minimal element of $ C $. Consider the Wheeler NFA $ (\mathcal{A}_*, \le_*) $ of Lemma~\ref{lem:quotient}. Then, there exists an isomorphism from $ \mathcal{A} $ to $ \mathcal{A}_* $.
\end{lemma}

\begin{proof}
    Let $ \mathcal{A} = (Q, E, s, F) $ and $ \mathcal{A}_* = (Q_*, E_*, s_*, F_*) $. By Lemma~\ref{lem:bisimulationtoquotient}, $ (\mathcal{A}_*, \le_*) $ is an element of $ C $. By the definition of $ \mathcal{A}_* $, we have $ Q_* = \mathcal{P}_{\equiv_{\mathcal{A}, \le}} $. Since $ Q_* $ is a partition of $ Q $, we have $ |Q_*| \le |Q| $. We know that $ (\mathcal{A}, \le) $ is a minimal element of $ C $ and $ (\mathcal{A}_*, \le_*) $ is an element of $ C $, so we have $ |Q_*| = |Q| $, or equivalently $ |\mathcal{P}_{\equiv_{\mathcal{A}, \le}}| = |Q| $. This means that $ \equiv_{\mathcal{A}, \le} $ is equal to $ id_Q $. By the definition of $ \mathcal{A}_* $, if $ \phi $ is the function from $ Q $ to $ Q_* $ such that $ \phi(u) = [u]_{\equiv_{\mathcal{A}, \le}} $ for every $ u \in Q $, then $ \phi $ is an isomorphism from $ \mathcal{A} $ to $ \mathcal{A}_* $. \qedhere
\end{proof}

The next lemma shows that Wheeler NFAs in the same $ \sim_{\mathfrak{C}} $-equivalence class must have the same quotient Wheeler NFA (up to isomorphism).

\medskip

\noindent{\textbf{Statement of Lemma~\ref{lem:quotientautomotaisomorphic}.}
    Consider a $ \sim_{\mathfrak{C}} $-equivalence class $ C $, let $ (\mathcal{A}, \le) $ and $ (\mathcal{A}', \le') $ be two elements of $ C $, and let $ (\mathcal{A}_*, \le_*) $ and $ (\mathcal{A}'_*, \le'_*) $ be the corresponding Wheeler NFAs built in Lemma~\ref{lem:quotient}, with $ \mathcal{A}_* = (Q_*, E_*, s_*, F_*) $ and $ \mathcal{A}'_* = (Q'_*, E'_*, s'_*, F'_*) $. Let $ Q_* = \{U_1, U_2, \dots, U_n \} $ and $ Q'_* = \{U'_1, U'_2, \dots, U'_{n'} \} $, where $ U_1 <_* U_2 <_* \dots <_* U_n $ and $ U'_1 <'_* U'_2 <'_* \dots <'_* U'_{n'} $. Then, $ n = n' $. Moreover, let $ \phi $ be the function from $ Q_* $ to $ Q'_* $ such that $ \phi(U_i) = U'_i $ for every $ 1 \le i \le n $. Then, (i)~$ \phi $ is a Wheeler bisimulation from $ (\mathcal{A}_*, \le_*) $ to $ (\mathcal{A}'_*, \le'_*) $, (ii)~$ \phi $ is the \emph{unique} Wheeler bisimulation from $ (\mathcal{A}_*, \le_*) $ to $ (\mathcal{A}'_*, \le'_*) $, and (iii)~$ \phi $ is an isomorphism from $ \mathcal{A}_* $ to $ \mathcal{A}'_* $.
\begin{proof}
    By Lemma~\ref{lem:bisimulationtoquotient}, $ (\mathcal{A}_*, \le_*) $ and $ (\mathcal{A}'_*, \le'_*) $ are elements of $ C $ because $ (\mathcal{A}, \le) $ and $ (\mathcal{A}', \le') $ are elements of $ C $. This means that there exists a Wheeler bisimulation $ R \subseteq Q_* \times Q'_* $ from $ (\mathcal{A}_*, \le_*) $ to $ (\mathcal{A}'_*, \le'_*) $. The conclusion will follow if we prove that $ R $ is a bijective function from $ Q_* $ to $ Q_*' $ (which implies that $ R $ is an isomorphism from $ \mathcal{A}_* $ to $ \mathcal{A}'_* $ by Lemma~\ref{lem:bisimulationisomorphism}, and in particular $ n = n' $) and that we must have $ R(U_i) = U'_i $ for every $ 1 \le i \le n $.
    
    By Lemma~\ref{lem:compositionisautobisimulation}, $ R^{-1} \circ R $ is a Wheeler autobisimulation on $ (\mathcal{A}_*, \le_*) $ and $ R \circ R^{-1} $ is a Wheeler autobisimulation on $ (\mathcal{A}'_*, \le'_*) $. By Lemma~\ref{lem:uniquewheelerautobisimulation}, we have $ R^{-1} \circ R = id_{Q_*} $ and $ R \circ R^{-1} = id_{Q_*'} $.

    Let us prove that $ R $ is a function. We must prove that $ |R(U)| = 1 $ for every $ U \in Q_* $. Fix $ U \in Q_* $. By Lemma~\ref{lem:R(u)nonempty}, we have $ |R(U)| \ge 1 $, so we only need to prove that $ |R(U)| \le 1 $. Assume that $ U'_1, U'_2 \in Q'_* $ satisfy $ U \; R \; U'_1 $ and $ U \; R \; U'_2 $. We must prove that $ U'_1 = U'_2 $. We have $ U'_1 \; R^{-1} \; U $ and $ U'_1 \; (R \circ R^{-1}) \; U'_2 $, so $ U'_1 = U'_2 $ because $ R \circ R^{-1} = id_{Q_*'} $.

    Let us prove that $ R $ is bijective. We must prove that $ |R^{-1}(U')| = 1 $ for every $ U' \in Q'_* $. Fix $ U' \in Q'_* $. By Lemma~\ref{lem:R(u)nonempty}, we have $ |R^{-1}(U')| \ge 1 $, so we only need to prove that $ |R^{-1}(U')| \le 1 $. Assume that $ U_1, U_2 \in Q_* $ satisfy $ U_1 \; R \; U' $ and $ U_2 \; R \; U' $. We must prove that $ U_1 = U_2 $. We have $ U' \; R^{-1} \; U_2 $ and $ U_1 \; (R^{-1} \circ R) \; U_2 $, so $ U_1 = U_2 $ because $ R^{-1} \circ R = id_{Q_*} $.

    We are only left with showing that $ R(U_i) = U'_i $ for every $ 1 \le i \le n $. We proceed by induction on $ i $. For $ i = 1 $, we have $ U_1 = s_* $ and $ U'_1 = s'_* $ by Axiom~1, and we know that $  R(s_*) = s'_* $ because $ R $ is an isomorphism from $ \mathcal{A}_* $ to $ \mathcal{A}'_* $, so $ R(U_i) = R(s_*) = s'_* = U'_1 $. Now assume that $ 2 \le i \le n $. By the inductive hypothesis, we know that $ R(U_j) = U'_j $ for every $ 1 \le j \le i - 1 $. We must show that $ R(U_i) = U'_i $. Let $ 1 \le k \le n $ be such that $ R(U_i) = U'_k $. We must prove that $ k = i $. We only need to prove that we cannot have $ k < i $ and we cannot have $ i < k $. We cannot have $ k < i $ because by the inductive hypothesis we would have $ R(U_k) = U'_k $, so we would have $ \{U_k, U_i \} \subseteq R^{-1}(U'_k) $ and $ R $ would not be injective. Moreover, we cannot have $ i < k $ because by the inductive hypothesis we would have $ R(\{U_1, U_2, \dots, U_{i - 2}, U_{i - 1}, U_i \}) = \{U'_1, U'_2, \dots, U'_{i - 2}, U'_{i - 1}, U'_k \} $, which would be a contradiction because $ R $ is a Wheeler bisimulation from $ (\mathcal{A}_*, \le_*) $ to $ (\mathcal{A}'_*, \le'_*) $, $ \{U_1, U_2, \dots, U_{i - 2}, U_{i - 1}, U_i \} $ is $ \le_* $-convex, but $ \{U'_1, U'_2, \dots, U'_{i - 2}, U'_{i - 1}, U'_k \} $ would not be $ \le'_* $-convex. \qedhere
\end{proof}

We can now prove that in every $ \sim_{\mathfrak{C}} $-equivalence class there is only one minimal Wheeler NFA (up to isomorphism).

\medskip

\noindent{\textbf{Statement of Theorem~\ref{theorem:minimalbisimilar}.}
    Consider a $ \sim_{\mathfrak{C}} $-equivalence class $ C $, and let $ (\mathcal{A}, \le) $ and $ (\mathcal{A}', \le') $ be two minimal elements of $ C $. Then, there exists an isomorphism from $ \mathcal{A} $ to $ \mathcal{A}' $.

\begin{proof}
    Let $ (\mathcal{A}_*, \le_*) $ and $ (\mathcal{A}'_*, \le'_*) $ be the Wheeler NFAs built in Lemma~\ref{lem:quotient} starting from $ (\mathcal{A}, \le) $ and $ (\mathcal{A}', \le') $, respectively. By Lemma~\ref{lem:quotientautomotaisomorphic}, there exists an isomorphism $ \phi $ from $ \mathcal{A}_* $ to $ \mathcal{A}'_* $. By Lemma~\ref{lem:stateminimalisomorphism}, there exists an isomorphism $ \phi_1 $ from $ \mathcal{A} $ to $ \mathcal{A}_* $ and there exists an isomorphism $ \phi_2 $ from $ \mathcal{A}' $ to $ \mathcal{A}'_* $. We conclude that $ \phi_2^{-1} \circ \phi \circ \phi_1 $ is an isomorphism from $ \mathcal{A} $ to $ \mathcal{A}' $. \qedhere
\end{proof}

Lastly, we show that the minimal Wheeler NFA of a $ \sim_{\mathfrak{C}} $-equivalence class can be obtained by quotienting \emph{any} Wheeler NFA in the class.

\medskip

\noindent{\textbf{Statement of Corollary~\ref{cor:quotientisstateminimal}.}
   Consider a $ \sim_{\mathfrak{C}} $-equivalence class $ C $, let $ (\mathcal{A}, \le) $ be an element of $ C $, and let $ (\mathcal{A}_*, \le_*) $ be the Wheeler NFA of Lemma~\ref{lem:quotient}. Then, $ (\mathcal{A}_*, \le_*) $ is a minimal element of $ C $. 

\begin{proof}
    Let $ (\mathcal{A}', \le') $ be a minimal element of $ C $. Notice that $ (\mathcal{A}_*, \le_*) $ is an element of $ C $ by Lemma~\ref{lem:bisimulationtoquotient}, so the conclusion will follow if we show that there exists an isomorphism from $ \mathcal{A}_* $ to $ \mathcal{A}' $.
    
    Let $ (\mathcal{A}'_*, \le'_*) $ the Wheeler NFA of Lemma~\ref{lem:quotient} obtained from $ (\mathcal{A}', \le') $. By Lemma~\ref{lem:quotientautomotaisomorphic}, there exists an isomorphism $ \phi $ from $ \mathcal{A}_* $ to $ \mathcal{A}_*' $. By Lemma~\ref{lem:stateminimalisomorphism}, there exists an isomorphism $ \phi_1 $ from $ \mathcal{A}' $ to $ \mathcal{A}'_* $. We conclude that $ \phi_1^{-1} \circ \phi $ is an isomorphism from $ \mathcal{A}_* $ to $ \mathcal{A}' $. \qedhere    
\end{proof}

\subsection{Determinism}

\noindent{\textbf{Statement of Lemma~\ref{lem:relatingdeterminismandbisimulation}.}
    Let $ \mathcal{L} \subseteq \Sigma^* $ be a Wheeler language, and let $ (\mathcal{A}, \le) $ be a Wheeler NFA. Then, $ (\mathcal{A}, \le) $ is a minimal element of $ D_\mathcal{L} $ if and only if it is a minimal element of $ T_\mathcal{L} $. In particular, every minimal element of $ T_\mathcal{L} $ is a Wheeler DFA.

\begin{proof}
    We know that $ D_\mathcal{L} \subseteq T_\mathcal{L} $, so we only need to prove that every minimal element of $ T_\mathcal{L} $ is in $ D_\mathcal{L} $. Assume that $ (\mathcal{A}, \le) $ is a minimal element of $ T_\mathcal{L} $. We need to prove that $ (\mathcal{A}, \le) $ is a Wheeler DFA such that $ \mathcal{L}(\mathcal{A}) = \mathcal{L} $.

    By Definition, $ T_\mathcal{L} $ contains a Wheeler DFA $ (\mathcal{A}', \le') $ such that $ \mathcal{L}(\mathcal{A}') = \mathcal{L} $, so by Lemma~\ref{lem:bisimulationandequivalence} we have $ \mathcal{L}(\mathcal{A}) = \mathcal{L}(\mathcal{A}') = \mathcal{L} $ because $ \mathcal{A} $ and $ \mathcal{A}' $ belong to $ T_\mathcal{L} $. We are only left with showing that $ \mathcal{A} $ is a DFA.

    Let $ (\mathcal{A}'_*, \le'_*) $ be the Wheeler NFA of Lemma~\ref{lem:quotient} built starting from $ (\mathcal{A}', \le') $. Since $ \mathcal{A}' $ is a DFA, by Lemma~\ref{lem:quotient} $ \mathcal{A}'_* $ is a DFA. By Corollary~\ref{cor:quotientisstateminimal}, $ (\mathcal{A}'_*, \le'_*) $ is a minimal element of $ T_\mathcal{L} $. By Theorem~\ref{theorem:minimalbisimilar}, there exists an isomorphism from $ \mathcal{A}'_* $ to $ \mathcal{A} $. Since $ \mathcal{A}'_* $ is a DFA, we conclude that $ \mathcal{A} $ is a DFA. \qedhere
\end{proof}

\section{Omitted Proofs and Results from Section~\ref{sec:algorithms}}

The main goal of this section is to prove that the quotient Wheeler NFA can be built in linear time (Theorem~\ref{theor:lineartimeminimization}). As discussed in Section~\ref{sec:introduction} and Section~\ref{sec:algorithms}, we assume that the edge labels can be sorted in linear time.

Fix a Wheeler NFA $ (\mathcal{A}, \le) $, with $ \mathcal{A} = (Q, E, s, F) $, and write $ Q = \{u_1, u_2, \dots, u_{|Q|} \} $, where $ u_1 < u_2 < \dots < u_{|Q|} $. In Section~\ref{sec:algorithms}, we defined the bit array $ B_{\mathcal{A}, \le} [2, |Q|] $. The idea is that $ B_{\mathcal{A}, \le}[2, |Q|] $ stores all the information required to compute $ \equiv_{\mathcal{A}, \le} $ because $ \equiv_{\mathcal{A}, \le} $ is $ \le $-convex by Theorem~\ref{theor:maximumbisimulationproperties} (see Lemma~\ref{lem:B_Aisallweneed}).

Our linear-time algorithm for building $ B_{\mathcal{A}, \le} [2, |Q|] $ (Algorithm~\ref{alg:computingmaximum}) is based on the following lemma.

\medskip

\noindent{\textbf{Statement of Lemma~\ref{lem:ideabehindalgorithm}.}
    Let $ (\mathcal{A}, \le) $ be a Wheeler NFA, with $ \mathcal{A} = (Q, E, s, F) $. Fix $ 2 \le i \le |Q| $ (in particular, $ \ell_i^{\min} \not = \bot $).
    \begin{enumerate}
        \item If $ Out(i - 1)  \not = Out(i) $, then $ B_{\mathcal{A}, \le}[i] = 1 $.
        \item If $ (u_{i - 1} \in F \land u_i \not \in F) \lor (u_{i - 1} \not \in F \land u_i  \in F) $, then $ B_{\mathcal{A}, \le}[i] = 1 $.
        \item Assume that $ B_{\mathcal{A}, \le}[i] = 1 $. If $ \ell_i^{\min} \ge 2 $, then $ B_{\mathcal{A}, \le}[\ell_i^{\min}] = 1 $, and if $ (\ell_{i - 1}^{\max} \not = \bot) \land (\ell_{i - 1}^{\max} \le |Q| - 1) $, then $ B_{\mathcal{A}, \le}[\ell_{i - 1}^{\max} + 1] = 1 $.
    \end{enumerate}

\begin{proof}
    \begin{enumerate}
        \item Assume that $ Out(i - 1)  \not = Out(i) \not = \emptyset $. Then, there exists $ a \in Out(i - 1) \setminus Out(i) $ or $ a \in Out(i) \setminus Out(i - 1) $. Since $ \equiv_{\mathcal{A}, \le} $ is a bisimulation from $ \mathcal{A} $ to $ \mathcal{A} $, we obtain $ u_{i - 1} \not \equiv_{\mathcal{A}, \le} u_i $, so $ B_{\mathcal{A}, \le}[i] = 1 $.
        \item Assume that $ (u_{i - 1} \in F \land u_i \not \in F) \lor (u_{i - 1} \not \in F \land u_i  \in F) $. Since $ \equiv_{\mathcal{A}, \le} $ is a bisimulation from $ \mathcal{A} $ to $ \mathcal{A} $, we obtain $ u_{i - 1} \not \equiv_{\mathcal{A}, \le} u_i $, so $ B_{\mathcal{A}, \le}[i] = 1 $.
        \item Assume that $ B_{\mathcal{A}, \le}[i] = 1 $. We prove the two properties separately.
        \begin{itemize}
            \item Assume that $ B_{\mathcal{A}, \le}[i] = 1 $ and $ \ell_i^{\min} \ge 2 $. We know that $ \ell_i^{\min} \not = \bot $ (which implies $ a_i^{\min} \not = \bot $), and $ B_{\mathcal{A}, \le}[\ell_i^{\min}]  $ is well defined because $ \ell_i^{\min} \ge 2 $. Suppose for the sake of a contradiction that $ B_{\mathcal{A}, \le}[\ell_i^{\min}] \not = 1 $. This means that $ u_{\ell_i^{\min} - 1} \equiv_{\mathcal{A}, \le} u_{\ell_i^{\min}} $. Since $ \equiv_{\mathcal{A}, \le} $ is a bisimulation from $ \mathcal{A} $ to $ \mathcal{A} $ and $ (u_{\ell_i^{\min}}, u_i, a_i^{\min}) \in E $, there exists $ 1 \le i' \le |Q| $ such that $ u_{i'} \equiv_{\mathcal{A}, \le} u_i $ and $ (u_{\ell_i^{\min} - 1}, u_{i'}, a_i^{\min}) \in E $. We distinguish three cases, and we show that of all them lead to a contradiction (see Figure~\ref{fig:notcrossingalgorithms}).
    \begin{itemize}
        \item Case 1: $ i < i' $. From $ (u_{\ell_i^{\min}}, u_i, a_i^{\min}), (u_{\ell_i^{\min} - 1}, u_{i'}, a_i^{\min}) \in E $ and $ i < i' $ we conclude $ \ell_i^{\min} \le \ell_i^{\min} - 1 $ by Axiom~3, which is a contradiction.
        \item Case 2: $ i' = i $. We have $ (u_{\ell_i^{\min} - 1}, u_i, a_i^{\min}) \in E $, so by the minimality of $ \ell_i^{\min} $ we conclude $ \ell_i^{\min} \le \ell_i^{\min} - 1 $, which is a contradiction.
        \item Case 3: $ i' < i $. We know that $ \equiv_{\mathcal{A}, \le} $ is $ \le $-convex (by Theorem~\ref{theor:maximumbisimulationproperties}), $ i' \le i - 1 < i $ and $ u_{i'} \equiv_{\mathcal{A}, \le} u_i $, so we obtain that $ u_{i - 1} \equiv_{\mathcal{A}, \le} u_i $ and we conclude that $ B_{\mathcal{A}, \le}[i] \not = 1 $, which is a contradiction.
    \end{itemize}    

            \item Assume that $ (\ell_{i - 1}^{\max} \not = \bot) \land (\ell_{i - 1}^{\max} \le |Q| - 1) $. We know that $ \ell_{i - 1}^{\max} \not = \bot $ (which implies $ a_{i - 1}^{\max} \not = \bot $), and $ B_{\mathcal{A}, \le}[\ell_{i - 1}^{\max} + 1]  $ is well defined because $ \ell_{i - 1}^{\max} \le |Q| - 1 $. From $ \ell_{i - 1}^{\max} \not = \bot $ we also obtain $ a_{i - 1}^{\max} \not = \bot $. Suppose for the sake of a contradiction that $ B_{\mathcal{A}, \le}[\ell_{i - 1}^{\max} + 1] \not = 1 $. This means that $ u_{\ell_{i - 1}^{\max}} \equiv_{\mathcal{A}, \le} u_{\ell_{i - 1}^{\max} + 1} $. Since $ \equiv_{\mathcal{A}, \le} $ is a bisimulation from $ \mathcal{A} $ to $ \mathcal{A} $ and $ (u_{\ell_{i - 1}^{\max}}, u_{i - 1}, a_{i - 1}^{\max}) \in E $, there exists $ 1 \le i' \le |Q| $ such that $ u_{i - 1} \equiv_{\mathcal{A}, \le} u_{i'} $ and $ (u_{\ell_{i - 1}^{\max} + 1}, u_{i'}, a_{i - 1}^{\max}) \in E $. We distinguish three cases, and we show that of all them lead to a contradiction.
            \begin{itemize}
                \item Case 1: $ i' < i - 1 $. From $ (u_{\ell_{i - 1}^{\max} + 1}, u_{i'}, a_{i - 1}^{\max}), (u_{\ell_{i - 1}^{\max}}, u_{i - 1}, a_{i - 1}^{\max}) \in E $ and $ i' < i - 1 $ we conclude $ \ell_{i - 1}^{\max} + 1 \le \ell_{i - 1}^{\max} $ by Axiom~3, which is a contradiction.
                \item Case 2: $ i' = i - 1 $. We have $ (u_{\ell_{i - 1}^{\max} + 1}, u_{i - 1}, a_{i - 1}^{\max}) \in E $, so by the maximality of $ \ell_{i - 1}^{\max} $ we conclude $ \ell_{i - 1}^{\max} + 1 \le \ell_{i - 1}^{\max} $, which is a contradiction.
                \item Case 3: $ i - 1 < i' $. We know that $ \equiv_{\mathcal{A}, \le} $ is $ \le $-convex (by Theorem~\ref{theor:maximumbisimulationproperties}), $ i - 1 < i \le i' $ and $ u_{i - 1} \equiv_{\mathcal{A}, \le} u_{i'} $, so we obtain that $ u_{i - 1} \equiv_{\mathcal{A}, \le} u_i $ and we conclude that $ B_{\mathcal{A}, \le}[i] \not = 1 $, which is a contradiction. \qedhere
            \end{itemize}
        \end{itemize}
    \end{enumerate}
\end{proof}

We now introduce several results that will be expedient to prove the correctness of Algorithm~\ref{alg:computingmaximum} (see Lemma~\ref{lem:algcorrectness}). Our first result is a general lemma.

\begin{lemma}\label{lem:convexequivalencerelationimpliesproperty2}
    Let $ Q $ be a set, let $ \le $ be a total order on $ Q $, and let $ \equiv $ be a $ \le $-convex equivalence relation on $ Q $. If $ C \subseteq Q $ is $ \le $-convex, then $ \bigcup_{u \in C}[u]_{\equiv} $ is $ \le $-convex.
\end{lemma}

\begin{proof}
    Define $ C' = \bigcup_{u \in C}[u]_{\equiv} $, and note that $ C \subseteq C' $. Let $ v', w', x' \in Q $ such that $ v' \le w' \le x' $ and $ v', x' \in C' $. We must prove that $ w' \in C' $. Since $ v', x' \in C' $, there exists $ v, x \in C $ such that $ v' \in [v]_\equiv $ and $ x' \in [x]_\equiv $. At least one of the following three cases must occur.
    \begin{itemize}
        \item Case 1: $ w' \le v $, Then, $ v' \le w' \le v $, so from $ v, v' \in [v]_\equiv $ we conclude $ w' \in [v]_\equiv \subseteq C $ because $ [v]_\equiv $ is $ \le $-convex.
        \item Case 2: $ x \le w' $. Then, $ x \le w' \le x' $, so from $ x, x' \in [x]_\equiv $ we conclude $ w' \in [v]_\equiv \subseteq C $ because $ [v]_\equiv $ is $ \le $-convex.
        \item Case 3: $ v < w' < x $. Then, from $ v, x \in C $ we conclude $ w' \in C \subseteq C' $ because $ C $ is $ \le $-convex. \qedhere
    \end{itemize}
\end{proof}

Let us focus on Algorithm~\ref{alg:computingmaximum}. We start with a remark.

\begin{remark}\label{rem:basicalgorithm}
    Note that at the end of Algorithm~\ref{alg:computingmaximum} we have $ B[i] \in \{0, 1 \} $ for every $ 2 \le i \le |Q| $. Moreover, if at some time the algorithm sets $ B[i] \gets 1 $ for some $ 2 \le i \le |Q| $ (in Line~\ref{line:firstB1} or Line~\ref{line:firstB2} or Line~\ref{line:firstB2bis}), then we still have $ B[i] = 1 $ at the end of the algorithm.
\end{remark}

We now define a relation $ \sim $ that depends on the execution of Algorithm~\ref{alg:computingmaximum}.

\begin{definition}\label{def:usedincorrectness}
    Let $ (\mathcal{A}, \le) $ be a Wheeler NFA, with $ \mathcal{A} = (Q, E, s, F) $. Let $ \sim $ be the relation from $ Q $ to $ Q $ such that for every $ 1 \le i, i' \le |Q| $ we have $ u_i \sim u_{i'} $ if and only if for every $ \min \{i, i' \} + 1 \le k \le \max \{i, i' \} $ we have $ B[k] = 0 $ at the end of the execution of Algorithm~\ref{alg:computingmaximum} on input $ (\mathcal{A}, \le) $.
\end{definition}

To prove Lemma~\ref{lem:algcorrectness}, we will prove that the relation $ \sim $ of Definition~\ref{def:usedincorrectness} is a Wheeler autobisimulation on $ (\mathcal{A}, \le) $ (see Lemma~\ref{lem:correctnesswheelerautobisimulation}). We first need some auxiliary results.

\begin{lemma}\label{lem:correctnessequivalencerelation}
    Let $ (\mathcal{A}, \le) $ be a Wheeler NFA, with $ \mathcal{A} = (Q, E, s, F) $. Then, the relation $ \sim $ of Definition~\ref{def:usedincorrectness} is a $ \le $-convex equivalence relation on $ Q $.
\end{lemma}

\begin{proof}
    Let us prove that $ \sim $ is an equivalence relation on $ Q $. From the definition of $ \sim $, we immediately obtain that, for every $ 1 \le i \le |Q|$, $ u_i \sim u_i $  (that is, $ \sim $ is reflexive) and, for every $ 1 \le i, i' \le |Q| $, $ u_i \sim u_{i'} $ implies $ u_{i'} \sim u_i $ (that is, $ \sim $ is symmetrical). Let us prove that $ \sim $ is transitive. Assume that $ 1 \le i, i', i'' \le |Q| $ satisfy $ u_i \sim u_{i'} $ and $ u_{i'} \sim u_{i''} $. We must prove that $ u_i \sim u_{i''} $. We distinguish two cases.
        \begin{itemize}
            \item Assume that $ i \le i'' $. By the definition of $ \sim $, we need to prove that, for every $ i + 1 \le k \le i'' $, we have $ B[k] = 0 $ at the end of the algorithm. Fix $ i + 1 \le k \le i'' $. Note that $ \min \{i, i' \} \le \min \{i', i'' \} \le i' \le \max \{i, i' \} \le \max \{i', i'' \} $ and $ \min \{i, i' \} + 1 \le i + 1 \le k \le i'' \le \max \{i, i'' \} $. If $ k \le \max \{i, i' \} $, we have $ \min \{i, i' \} + 1 \le k \le \max \{i, i' \} $, and from $ u_i \sim u_{i'} $ we conclude that $ B[k] = 0 $ at the end of the algorithm. Otherwise, we have $ \max \{i, i' \} + 1 \le k $, so
            $ \min \{i', i'' \} + 1 \le \max \{i, i' \} + 1 \le k $ and $ \min \{i', i'' \} + 1 \le k \le \max \{i', i'' \} $, and from $ u_{i'} \sim u_{i''} $ we conclude that $ B[k] = 0 $ at the end of the algorithm.
            \item Assume that $ i'' < i $. By the symmetry of $ \sim $, we have $ u_{i''} \sim u_{i'} $ and $ u_{i'} \sim u_i $, so by the previous case we obtain $ u_{i''} \sim u_i $, and again by the symmetry of $ \sim $ we conclude $ u_i \sim u_{i''} $.
        \end{itemize}
        Let us prove that the equivalence relation $ \sim $ is $ \le $-convex. Fix $ 1 \le i \le |Q| $. We must prove that $ [u_i]_\sim $ is $ \le $-convex. Let $ 1 \le h_1 \le h_2 \le h_3 \le |Q| $ such that $ u_{h_1}, u_{h_3} \in [u_i]_\sim $. We must prove that $ u_{h_2} \in [u_i]_\sim $, that is, we need to prove that $ u_i \sim u_{h_2} $. Since $ u_{h_1}, u_{h_3} \in [u_i]_\sim $, we have $ u_i \sim u_{h_1} $ and $ u_i \sim u_{h_3} $. We know that $ \sim $ is an equivalence relation, so $ u_{h_1} \sim u_{h_3} $, and we only need to prove that $ u_{h_1} \sim u_{h_2} $. By the definition of $ \sim $, we have to prove that, for every $ h_1 + 1 \le k \le h_2 $, we have $ B[k] = 0 $ at the end of the algorithm. Fix $ h_1 + 1 \le k \le h_2 $. Since $ h_2 \le h_3 $, we have $ h_1 + 1 \le k \le h_3 $, so we have $ B[k] = 0 $ at the end of the algorithm because $ u_{h_1} \sim u_{h_3} $. \qedhere
\end{proof}

\begin{lemma}\label{lem:correctnessbisimulation}
    Let $ (\mathcal{A}, \le) $ be a Wheeler NFA, with $ \mathcal{A} = (Q, E, s, F) $. Then, the relation $ \sim $ of Definition~\ref{def:usedincorrectness} is a bisimulation from $ \mathcal{A} $ to $ \mathcal{A} $.
\end{lemma}

\begin{figure}
     \centering
           \begin{subfigure}[b]{0.49\textwidth}
        \centering
        \scalebox{.7}{
        \begin{tikzpicture}[->,>=stealth', semithick, auto, scale=1]
\tikzset{
  every node/.style={inner sep=0pt, outer sep=0pt}
}
\node[draw=none, inner sep=0pt, anchor=center] (1)    at (1,0)	{$ u_{\min \{j, j' \}} $};
\node[draw=none, inner sep=0pt, anchor=center] (2)    at (4, 0)	{$ u_k $};
\node[draw=none, inner sep=0pt, anchor=center] (3)    at (6, 0)	{$ u_{\max \{j, j' \}} $};
\node[draw=none, inner sep=0pt, anchor=center] (1')    at (-2, -3)	{$ u_{\min \{i, i' \}} $};
\node[draw=none, inner sep=0pt, anchor=center] (2')    at (0,-3)	{$ u_{\ell_k}^{\min} $};
\node[draw=none, inner sep=0pt, anchor=center] (3')    at (3.5,-3)	{$ u_{\max \{i, i' \}} $};

\draw (1') edge [] node [] {$ a $} (1);
\draw (2') edge [] node [] {$ a_k^{\min} $} (2);
\draw (3') edge [] node [] {$ a $} (3);

\end{tikzpicture}
}
\end{subfigure}
     \begin{subfigure}[b]{0.49\textwidth}
        \centering
        \scalebox{.7}{
        \begin{tikzpicture}[->,>=stealth', semithick, auto, scale=1]
\tikzset{
  every node/.style={inner sep=0pt, outer sep=0pt}
}
\node[draw=none, inner sep=0pt, anchor=center] (1)    at (1,0)	{$ u_{\min \{j, j' \}} $};
\node[draw=none, inner sep=0pt, anchor=center] (3)    at (4, 0)	{$ u_k $};
\node[draw=none, inner sep=0pt, anchor=center] (1')    at (-2, -3)	{$ u_{\min \{i, i' \}} $};
\node[draw=none, inner sep=0pt, anchor=center] (2')    at (0,-3)	{$ u_{\ell_k}^{\min} $};
\node[draw=none, inner sep=0pt, anchor=center] (3')    at (3.5,-3)	{$ u_{\max \{i, i' \}} $};

\draw (1') edge [] node [] {$ a $} (1);
\draw (2') edge [] node [] {$ a_k^{\min} $} (3);
\draw (3') edge [] node [] {$ a $} (3);

\end{tikzpicture}
}
\end{subfigure}  
 	\caption{Showing that $ a_k^{\min} = a $ in the proof of Lemma~\ref{lem:correctnessbisimulation}. The case $ k < \max \{j, j' \} $ is on the left and the case $ k = \max \{j, j' \} $ is on the right. In both cases, $ a \preceq a_k^{\min} $ follows from Axiom~2. The inequality $ a_k^{\min} \preceq a $ follows from Axiom~2 if $ k < \max \{j, j' \} $ and from the minimality of $ a_k^{\min} $ if $ k = \max \{j, j' \} $.}\label{fig:correctness1}
\end{figure}

\begin{proof}
    We need to show that the four properties listed in Definition~\ref{def:bisimulation} are true.

    \begin{itemize}
                \item Let $ 1 \le i, i', j \le |Q| $ and $ a \in \Sigma $ be such that $ u_i \sim u_{i'} $ and $ (u_i, u_j, a) \in E $. We must prove that there exists $ 1 \le j' \le |Q| $ such that $ u_j \sim u_{j'} $ and $ (u_{i'}, u_{j'}, a) \in E $. If $ i' = i $ we can choose $ j' = j $, so in the rest of the proof we can assume $ i' \not = i $.

                Let us prove that $ Out(i) = Out (i') $. Suppose for the sake of a contradiction that $ Out(i) \not = Out(i') $. This means that there exists $ \min \{i, i' \} + 1 \le h \le \max \{i, i' \} $ such that $ Out(h - 1) \not = Out(h) $ (otherwise, by repeated transitivity we would have $ Out(i) = Out(i') $). Then, $ h \ge 2 $, so at some time the algorithm sets $ B[h] \gets 1 $ in Line~\ref{line:firstB1} and at the end of the algorithm we have $ B[h] = 1 $. This is a contradiction because $ u_i \sim u_{i'} $, so by the definition of $ \sim $ we must have $ B[h] = 0 $ at the end of the algorithm.
                
                Since $ Out(i) = Out (i') $ and $ (u_i, u_j, a) \in E $, there exists $ 1 \le j' \le |Q| $ such that $ (u_{i'}, u_{j'}, a) \in E $. We only need to prove that $ u_j \sim u_{j'} $.
                
                Suppose for the sake of a contradiction that $ u_j \sim u_{j'} $ is not true. This means that there exists $ \min \{j, j' \} + 1 \le k \le \max \{j, j' \} $ such that $ B[k] \not = 0 $ at the end of the algorithm. In particular, $ k \ge 2 $, so $ a_k^{\min} \not = \bot $, $ \ell_k^{\min} \not = \bot $ and $ (u_{\ell_k^{\min}}, u_k, a_k^{\min}) \in E $ is well defined.
                
                Let us prove that $ (u_{\min\{i, i'\}}, u_{\min\{j, j'\}}, a), (u_{\max\{i, i'\}}, u_{\max\{j, j'\}}, a) \in E $. To this end, we only need to prove that one of these two edges is equal to $ (u_i, u_j, a) $ and the other is equal to $ (u_{i'}, u_{j'}, a) $. We distinguish three cases.
                \begin{itemize}
                    \item If $ j = j' $, the conclusion is immediate.
                    \item If $ j < j' $, then $ i \le i' $ by Axiom~3, so $ (u_{\min\{i, i'\}}, u_{\min\{j, j'\}}, a) = (u_i, u_j, a) $ and $ (u_{\max\{i, i'\}}, u_{\max\{j, j'\}}, a) = (u_{i'}, u_{j'}, a) $.
                    \item If $ j' < j $, then $ i' \le i $ by Axiom~3, so $ (u_{\min\{i, i'\}}, u_{\min\{j, j'\}}, a) = (u_{i'}, u_{j'}, a) $ and $ (u_{\max\{i, i'\}}, u_{\max\{j, j'\}}, a) = (u_i, u_j, a) $.
                \end{itemize}

\begin{figure}
     \centering
           \begin{subfigure}[b]{0.49\textwidth}
        \centering
        \scalebox{.7}{
        \begin{tikzpicture}[->,>=stealth', semithick, auto, scale=1]
\tikzset{
  every node/.style={inner sep=0pt, outer sep=0pt}
}
\node[draw=none, inner sep=0pt, anchor=center] (1)    at (1,0)	{$ u_{\min \{j, j' \}} $};
\node[draw=none, inner sep=0pt, anchor=center] (2)    at (4, 0)	{$ u_k $};
\node[draw=none, inner sep=0pt, anchor=center] (3)    at (6, 0)	{$ u_{\max \{j, j' \}} $};
\node[draw=none, inner sep=0pt, anchor=center] (1')    at (-2, -3)	{$ u_{\min \{i, i' \}} $};
\node[draw=none, inner sep=0pt, anchor=center] (2')    at (0,-3)	{$ u_{\ell_k}^{\min} $};
\node[draw=none, inner sep=0pt, anchor=center] (3')    at (3.5,-3)	{$ u_{\max \{i, i' \}} $};

\draw (1') edge [] node [] {$ a_k^{\min} $} (1);
\draw (2') edge [] node [] {$ a_k^{\min} $} (2);
\draw (3') edge [] node [] {$ a_k^{\min} $} (3);

\end{tikzpicture}
}
\end{subfigure}
     \begin{subfigure}[b]{0.49\textwidth}
        \centering
        \scalebox{.7}{
        \begin{tikzpicture}[->,>=stealth', semithick, auto, scale=1]
\tikzset{
  every node/.style={inner sep=0pt, outer sep=0pt}
}
\node[draw=none, inner sep=0pt, anchor=center] (1)    at (1,0)	{$ u_{\min \{j, j' \}} $};
\node[draw=none, inner sep=0pt, anchor=center] (3)    at (4, 0)	{$ u_k $};
\node[draw=none, inner sep=0pt, anchor=center] (1')    at (-2, -3)	{$ u_{\min \{i, i' \}} $};
\node[draw=none, inner sep=0pt, anchor=center] (2')    at (0,-3)	{$ u_{\ell_k}^{\min} $};
\node[draw=none, inner sep=0pt, anchor=center] (3')    at (3.5,-3)	{$ u_{\max \{i, i' \}} $};

\draw (1') edge [] node [] {$ a_k^{\min} $} (1);
\draw (2') edge [] node [] {$ a_k^{\min} $} (3);
\draw (3') edge [] node [] {$ a_k^{\min} $} (3);

\end{tikzpicture}
}
\end{subfigure}  
 	\caption{Showing that $ \min \{i, i' \} \le \ell_k^{\min} \le \max \{i, i' \} $ in the proof of Lemma~\ref{lem:correctnessbisimulation}. The case $ k < \max \{j, j' \} $ is on the left and the case $ k = \max \{j, j' \} $ is on the right. In both cases, $ \min \{i, i' \} \le \ell_k^{\min} $ follows from Axiom~3. The inequality $ \ell_k^{\min} \le \max \{i, i' \} $ follows from Axiom~3 if $ k < \max \{j, j' \} $ and from the minimality of $ \ell_k^{\min} $ if $ k = \max \{j, j' \} $.}\label{fig:correctness2}
\end{figure}

                Let us prove that $ a_k^{\min} = a $ (see Figure~\ref{fig:correctness1}). We prove separately that $ a \preceq a_k^{\min} $ and $ a_k^{\min} \preceq a $. From $ (u_{\min\{i, i'\}}, u_{\min\{j, j'\}}, a), (u_{\ell_k^{\min}}, u_k, a_k^{\min}) \in E $, $ \min\{j, j'\} < k $ and Axiom~2 we obtain $ a \preceq a_k^{\min} $. Now, let us prove that $ a_k^{\min} \preceq a $. We know that $ k \le \max \{j, j' \} $. If $ k < \max \{j, j' \} $, then from $ (u_{\ell_k^{\min}}, u_k, a_k^{\min}) \in E $, $ (u_{\max\{i, i'\}}, u_{\max\{j, j'\}}, a) \in E $ and Axiom~2 we obtain $ a_k^{\min} \preceq a $. If $ k = \max \{j, j' \} $, then from $ (u_{\ell_k^{\min}}, u_k, a_k^{\min}) \in E $, $ (u_{\max\{i, i'\}}, u_{\max\{j, j'\}}, a) \in E $ and the minimality of $ a_k^{\min} $ we obtain $ a_k^{\min} \preceq a $.

                Let us prove that $ \min \{i, i' \} \le \ell_k^{\min} \le \max \{i, i' \} $ (see Figure~\ref{fig:correctness2}). We prove separately that $ \min \{i, i' \} \le \ell_k^{\min} $ and $ \ell_k^{\min} \le \max \{i, i' \} $. From $ (u_{\min\{i, i'\}}, u_{\min\{j, j'\}}, a) \in E $, $ (u_{\ell_k^{\min}}, u_k, a_k^{\min}) \in E $, $ a_k^{\min} = a $, $ \min\{j, j'\} < k $ and Axiom~3 we obtain $ \min \{i, i' \} \le \ell_k^{\min} $. Now, let us prove that $ \ell_k^{\min} \le \max \{i, i' \} $. We know that $ k \le \max \{j, j' \} $. If $ k < \max \{j, j' \} $, then from $ (u_{\ell_k^{\min}}, u_k, a_k^{\min}) \in E $, $ (u_{\max\{i, i'\}}, u_{\max\{j, j'\}}, a) \in E $, $ a_k^{\min} = a $ and Axiom~3 we obtain $ \ell_k^{\min} \le \max \{i, i' \} $. If $ k = \max \{j, j' \} $, then from $ (u_{\ell_k^{\min}}, u_k, a_k^{\min}) \in E $, $ (u_{\max\{i, i'\}}, u_{\max\{j, j'\}}, a) \in E $, $ a_k^{\min} = a $ and the minimality of $ \ell_k^{\min} $ we obtain $ \ell_k^{\min} \le \max \{i, i' \} $.

                We have proved that $ \min \{i, i' \} \le \ell_k^{\min} $. We distinguish the cases $ \min \{i, i' \} < \ell_k^{\min} $ and $ \min \{i, i' \} = \ell_k^{\min} $.
                \begin{itemize}
                    \item Assume that $ \min \{i, i' \} < \ell_k^{\min} $. In particular $ \ell_k^{\min} \ge 2 $, so $ B[\ell_k^{\min}] $ is defined. We will show that at the end of the algorithm $ B[\ell_k^{\min}] = 1 $. This is a contradiction because $ \min \{i, i' \} + 1 \le \ell_k^{\min} \le \max \{i, i' \} $, so from $ u_i \sim u_{i'} $ and the definition of $ \sim $ we must have $ B[\ell_k^{\min}] = 0 $ at the end of the algorithm.

                    Let us show that at the end of the algorithm we have $ B[\ell_k^{\min}] = 1 $. To this end, by Remark~\ref{rem:basicalgorithm} we only need to show that at some time the algorithm sets $ B[\ell_k^{\min}] \gets 1 $. Since at the end of the algorithm $ B[k] \not = 0 $, at some time the algorithm sets $ B[k] \gets 1 $ (in Line~\ref{line:firstB1} or Line~\ref{line:firstB2} or Line~\ref{line:firstB2bis}). Then, in the previous line, the algorithm adds $ k $ to the queue $ K $ (Line~\ref{line:next1} or Line~\ref{line:next2} or Line~\ref{line:next2bis}). At some later time, the algorithm removes $ k $ from the queue $ K $ (Line~\ref{line:dequeue}), and the condition in Line~\ref{line:recursivecheck} is checked. We know that $ \ell_k^{\min} \ge 2 $. If $ B[\ell_k^{\min}] \not = 0 $, then the algorithm has already set $ B[\ell_k^{\min}] \gets 1 $ at some earlier time, and we are done. If $ B[\ell_k^{\min}] = 0 $, then Line~\ref{line:firstB2} is executed and the algorithms sets $ B[\ell_k^{\min}] \gets 1 $.
                    \item Assume that $ \min \{i, i' \} = \ell_k^{\min} $. We know that  $ \min \{j, j' \} + 1 \le k \le \max \{j, j' \} $, so $ \min \{j, j' \}  \le k - 1 < \max \{j, j' \} $. Then, $ u_{k - 1} $ has at least one incoming edge because (i)~if $ \min \{j, j' \}  < k - 1 $, then $ k - 1 \ge 2 $, and (ii)~if $ \min \{j, j' \} = k - 1 $, then we know that $ (u_{\min\{i, i'\}}, u_{\min\{j, j'\}}, a) \in E $. This means that $ a_{k - 1}^{\max} \not = \bot $, so $ \ell_{k - 1}^{\max} \not = \bot $ and the edge $ (u_{\ell_{k - 1}^{\max}}, u_{k - 1}, a_{k - 1}^{\max}) \in E $ is well defined.

                    Let us prove that $ a_{k - 1}^{\max} = a $ (intuitively, our argument is similar to the one used to prove that $ a_k^{\min} = a $). We prove separately that $ a_{k - 1}^{\max} \preceq a $ and $ a \preceq a_{k - 1}^{\max} $. From $ (u_{\ell_{k - 1}^{\max}}, u_{k - 1}, a_{k - 1}^{\max}) \in E $, $ (u_{\max\{i, i'\}}, u_{\max\{j, j'\}}, a) \in E $, $ k - 1 < \max \{j, j' \} $ and Axiom~2 we obtain $ a_{k - 1}^{\max} \preceq a $. Now, let us prove that $ a \preceq a_{k - 1}^{\max} $. We know that $ \min \{j, j' \}  \le k - 1 $. If $ \min \{j, j' \} < k - 1 $, then from $ (u_{\min\{i, i'\}}, u_{\min\{j, j'\}}, a) \in E $, $ (u_{\ell_{k - 1}^{\max}}, u_{k - 1}, a_{k - 1}^{\max}) \in E $ and Axiom~2 we obtain $ a \preceq a_{k - 1}^{\max} $. If $ \min \{j, j' \} = k - 1 $, then from $ (u_{\min\{i, i'\}}, u_{\min\{j, j'\}}, a) \in E $, $ (u_{\ell_{k - 1}^{\max}}, u_{k - 1}, a_{k - 1}^{\max}) \in E $ and the maximality of $ a_{k - 1}^{\max} $ we obtain $ a \preceq a_{k - 1}^{\max} $.

                    Let us prove that $ \ell_{k - 1}^{\max} = \min\{i, i'\} $ (intuitively, our argument is similar to the one used to prove that $ \min \{i, i' \} \le \ell_k^{\min} \le \max \{i, i' \} $). We prove separately that $ \ell_{k - 1}^{\max} \le \min\{i, i'\} $ and $ \min\{i, i'\} \le \ell_{k - 1}^{\max} $. From $ (u_{\ell_{k - 1}^{\max}}, u_{k - 1}, a_{k - 1}^{\max}) \in E $, $ (u_{\ell_k^{\min}}, u_k, a_k^{\min}) \in E $, $ a_{k - 1}^{\max} = a = a_k^{\min} $, $ k - 1 < k $ and Axiom~3 we obtain $ \ell_{k - 1}^{\max} \le \ell_k^{\min} $, so from $ \min \{i, i' \} = \ell_k^{\min} $ we obtain $ \ell_{k - 1}^{\max} \le \min \{i, i' \} $. Now, let us prove that $ \min\{i, i'\} \le \ell_{k - 1}^{\max} $. We know that $ \min\{j, j'\} \le k - 1 $. If $ \min\{j, j'\} < k - 1 $, then from $ (u_{\min\{i, i'\}}, u_{\min\{j, j'\}}, a) \in E $, $ (u_{\ell_{k - 1}^{\max}}, u_{k - 1}, a_{k - 1}^{\max}) \in E $, $ a_{k - 1}^{\max} = a $ and Axiom~3 we obtain $ \min\{i, i'\} \le \ell_{k - 1}^{\max} $. If $ \min\{j, j'\} = k - 1 $, then from $ (u_{\min\{i, i'\}}, u_{\min\{j, j'\}}, a) \in E $, $ (u_{\ell_{k - 1}^{\max}}, u_{k - 1}, a_{k - 1}^{\max}) \in E $, $ a_{k - 1}^{\max} = a $ and the maximality of $ \ell_{k - 1}^{\max} $ we obtain $ \min\{i, i'\} \le \ell_{k - 1}^{\max} $.

                    Since $ i' \not = i $, we have $ \ell_{k - 1}^{\max} = \min\{i, i'\} < \max\{i, i'\} $, so $ \ell_{k - 1}^{\max} \le |Q| - 1 $ and $ B[\ell_{k - 1}^{\max} + 1] $ is defined. We will show that at the end of the algorithm $ B[\ell_{k - 1}^{\max} + 1] = 1 $. This is a contradiction because $ \min \{i, i' \} + 1 = \ell_{k - 1}^{\max} + 1 \le \max \{i, i' \} $, so from $ u_i \sim u_{i'} $ and the definition of $ \sim $ we must have $ B[\ell_{k - 1}^{\max} + 1] = 0 $ at the end of the algorithm.

                    Let us show that at the end of the algorithm we have $ B[\ell_{k - 1}^{\max} + 1] = 1 $. To this end, by Remark~\ref{rem:basicalgorithm} we only need to show that at some time the algorithm sets $ B[\ell_{k - 1}^{\max} + 1] \gets 1 $. Since at the end of the algorithm $ B[k] \not = 0 $, at some time the algorithm sets $ B[k] \gets 1 $ (in Line~\ref{line:firstB1} or Line~\ref{line:firstB2} or Line~\ref{line:firstB2bis}). Then, in the previous line, the algorithm adds $ k $ to the queue $ K $ (Line~\ref{line:next1} or Line~\ref{line:next2} or Line~\ref{line:next2bis}). At some later time, the algorithm removes $ k $ from the queue $ K $ (Line~\ref{line:dequeue}), and the condition in Line~\ref{line:recursivecheckbis} is checked. We know that $ \ell_{k - 1}^{\max} \not = \bot $ and $ \ell_{k - 1}^{\max} \le |Q| - 1 $. If $ B[\ell_{k - 1}^{\max} + 1] \not = 0 $, then the algorithm has already set $ B[\ell_{k - 1}^{\max} + 1] \gets 1 $ at some earlier time, and we are done. If $ B[\ell_{k - 1}^{\max} + 1] = 0 $, then Line~\ref{line:firstB2bis} is executed and the algorithms sets $ B[\ell_{k - 1}^{\max} + 1] \gets 1 $.
                \end{itemize}
            \item Let $ 1 \le i, i', j' \le |Q| $ and $ a \in \Sigma $ be such that $ u_i \sim u_{i'} $ and $ (u_{i'}, u_{j'}, a) \in E $. We must prove that there exists $ 1 \le j \le |Q| $ such that $ u_j \sim u_{j'} $ and $ (u_i, u_j, a) \in E $. Since $ \sim $ is symmetrical (because $ \sim $ is an equivalence relation by Lemma~\ref{lem:correctnessequivalencerelation}), we have $ u_{i'}  \sim u_i $, and by the previous point there exists $ 1 \le j \le |Q| $ such that $ u_{j'} \sim u_j $ and $ (u_i, u_j, a) \in E $. We conclude that $ u_j \sim u_{j'} $ again by the symmetry of $ \sim $.
            \item We have $ s \sim s $ because $ \sim $ is reflexive (because $ \sim $ is an equivalence relation by Lemma~\ref{lem:correctnessequivalencerelation}).
            \item Let $ 1 \le i, i' \le |Q| $ be such that $ u_i \sim u_{i'} $. We must prove that $ u_i \in F $ if and only if $ u_{i'} \in F $. Suppose for the sake of a contradiction that $ (u_i \in F \land u_{i'} \not \in F) \lor (u_i \not \in F \land u_{i'} \in F) $. This means that there exists $ \min \{i, i' \} + 1 \le h \le \max \{i, i' \} $ such that $ (u_{h - 1} \in F \land u_h \not \in F) \lor (u_{h - 1} \not \in F \land u_h \in F) $ (otherwise, by repeated transitivity we would have $ (u_i \in F \land u_{i'} \in F) \lor (u_i \not \in F \land u_{i'} \not \in F) $). Then, $ h \ge 2 $, so at some time the algorithm sets $ B[h] \gets 1 $ in Line~\ref{line:firstB1} and at the end of the algorithm we have $ B[h] = 1 $ by Remark~\ref{rem:basicalgorithm}. This is a contradiction because $ u_i \sim u_{i'} $, so by the definition of $ \sim $ we must have $ B[h] = 0 $ at the end of the algorithm. \qedhere
        \end{itemize}
\end{proof}

\begin{lemma}\label{lem:correctnesswheelerautobisimulation}
    Let $ (\mathcal{A}, \le) $ be a Wheeler NFA, with $ \mathcal{A} = (Q, E, s, F) $. Then, the relation $ \sim $ of Definition~\ref{def:usedincorrectness} is a Wheeler autobisimulation on $ (\mathcal{A}, \le) $.
\end{lemma}

\begin{proof}
    We know that $ \sim $ is reflexive (because $ \sim $ is an equivalence relation by Lemma~\ref{lem:correctnessequivalencerelation}), so we only need to prove that $ \sim $ is a Wheeler bisimulation from $ (\mathcal{A}, \le) $ to $ (\mathcal{A}, \le) $.
    
    Property~1 is true because $ \sim $ is a bisimulation from $ \mathcal{A} $ to $ \mathcal{A} $ by Lemma~\ref{lem:correctnessbisimulation}.
            
    Let us prove Property~2. Let $ C \subseteq Q $ be a $ \le $-convex set. We must prove that $ \{u' \in Q \;|\; (\exists u \in C)(u \sim u') \} $ is $ \le $-convex. Since $ \sim $ is an equivalence relation by Lemma~\ref{lem:correctnessequivalencerelation}, we have $ \{u' \in Q \;|\; (\exists u \in C)(u \sim u') \} = \bigcup_{u \in C}[u]_{\equiv} $, so the conclusion follows from Lemma~\ref{lem:convexequivalencerelationimpliesproperty2} because $ \sim $ is $ \le $-convex by Lemma~\ref{lem:correctnessequivalencerelation}.

    Let us prove Property~3. Let $ C' \subseteq Q $ be a $ \le $-convex set. We must prove that $ \{u \in Q \;|\; (\exists u' \in C')(u \sim u') \} $ is $ \le $-convex. The conclusion follows by arguing as we did in the previous point because $ \{u \in Q \;|\; (\exists u' \in C')(u \sim u') \} = \bigcup_{u' \in C'}[u']_{\equiv} $.  \qedhere   
\end{proof}

We are finally ready to show the correctness of Algorithm~\ref{alg:computingmaximum}.

\medskip

\noindent{\textbf{Statement of Lemma~\ref{lem:algcorrectness}.}
    Let $ (\mathcal{A}, \le) $ be a Wheeler NFA, with $ \mathcal{A} = (Q, E, s, F) $. On input $ (\mathcal{A}, \le) $, Algorithm~\ref{alg:computingmaximum} returns the array $ B_{\mathcal{A}, \le}[2, |Q|] $.

\begin{proof}
    At the end of the algorithm we have $ B[i] \in \{0, 1 \} $ for every $ 2 \le i \le |Q| $ (see Remark~\ref{rem:basicalgorithm}). Consequently, to prove that at the end of the algorithm we have $ B[2, |Q|] = B_{\mathcal{A}, \le}[2, |Q|] $, we need to prove that (i)~for every $ 2 \le i \le |Q| $, if at the end of the algorithm we have $ B[i] = 1 $, then $ B_{\mathcal{A}, \le}[i] = 1 $, and (ii)~for every $ 2 \le i \le |Q| $, if at the end of the algorithm we have $ B[i] = 0 $, then $ B_{\mathcal{A}, \le}[i] = 0 $. We prove these two claims separately.

    \begin{itemize}
        \item Let us prove that, for every $ 2 \le i \le |Q| $, if at the end of the algorithm we have $ B[i] = 1 $, then $ B_{\mathcal{A}, \le}[i] = 1 $. By Remark~\ref{rem:basicalgorithm}, we only need to prove that for every $ 2 \le i \le |Q| $, if at some time the algorithm sets $ B[i] \gets 1 $, then we have $ B_{\mathcal{A}, \le}[i] = 1 $.

        First, consider the case where the algorithm sets $ B[i] \gets 1 $ in Line~\ref{line:firstB1}. This means that $ (u_{i - 1} \in F \land u_i \not \in F) \lor (u_{i - 1} \not \in F \land u_i  \in F)  \lor (Out(i - 1) \not = Out(i)) $. In all cases, we have $ B_{\mathcal{A}, \le}[i] = 1 $ by Lemma~\ref{lem:ideabehindalgorithm}.

        Now, consider the case where the algorithm sets $ B[i] \gets 1 $ in Line~\ref{line:firstB2} or Line~\ref{line:firstB2bis}. We consider the executions of Line~\ref{line:firstB2} and Line~\ref{line:firstB2bis} in the order in which they occur, and we proceed inductively. Assume that at some time the algorithm sets $ B[\ell_i^{\min}] \gets 1 $ in Line~\ref{line:firstB2} (note that $ 2 \le \ell_i^{\min} \le |Q| $ by the condition in Line~\ref{line:recursivecheck}) or it sets $ B[\ell^{\max}_{i - 1} + 1] \gets 1 $ in Line~\ref{line:firstB2bis} (note that $ 2 \le \ell^{\max}_{i - 1} + 1 \le |Q| $ by the condition in Line~\ref{line:recursivecheckbis}). In the former case, our goal is to prove that $ B_{\mathcal{A}, \le}[\ell_i^{\min}] = 1 $, and in the latter case, our goal is to prove that $ B_{\mathcal{A}, \le}[\ell^{\max}_{i - 1} + 1] = 1 $. By Lemma~\ref{lem:ideabehindalgorithm}, we only need to prove that $ B_{\mathcal{A}, \le}[i] = 1 $. We know that $ i $ was removed from the queue $ K $ in Line~\ref{line:dequeue}. Since $ K $ is initially empty, $ i $ was previously added to $ K $. This must have happened in Line~\ref{line:next1} or Line~\ref{line:next2} or Line~\ref{line:next2bis}. If it happened in Line~\ref{line:next1}, then the algorithm set $ B[i] \gets 1 $ immediately after (in Line~\ref{line:firstB1}), and we have already shown that this implies $ B_{\mathcal{A}, \le}[i] = 1 $. If it happened in Line~\ref{line:next2} or Line~\ref{line:next2bis}, then the algorithm set $ B[i] \gets 1 $ immediately after (in Line~\ref{line:firstB2} or Line~\ref{line:firstB2bis}), so by the inductive hypothesis we obtain $ B_{\mathcal{A}, \le}[i] = 1 $.
        \item Let us prove that, for every $ 2 \le i \le |Q| $, if at the end of the algorithm we have $ B[i] = 0 $, then $ B_{\mathcal{A}, \le}[i] = 0 $. Let $ \sim $ be the relation from $ Q $ to $ Q $ introduced in Definition~\ref{def:usedincorrectness}. In particular, by the definition of $ \sim $, we have $ u_{i - 1} \sim u_i $ for every $ 2 \le i \le |Q| $ such that $ B[i] = 0 $ at the end of the algorithm. By Lemma~\ref{lem:correctnesswheelerautobisimulation}, $ \sim $ is a Wheeler autobisimulation on $ (\mathcal{A}, \le) $, so by the maximality of $ \equiv_{\mathcal{A}, \le} $ (see Theorem~\ref{theor:maximumbisimulationproperties}) we obtain $ u_{i - 1} \equiv_{\mathcal{A}, \le} u_i $ for every $ 2 \le i \le |Q| $ such that $ B[i] = 0 $ at the end of the algorithm, and so $ B_{\mathcal{A}, \le}[i] = 0 $ for every $ 2 \le i \le |Q| $ such that $ B[i] = 0 $ at the end of the algorithm. \qedhere
    \end{itemize}
          
\end{proof}

Let us discuss the running time of Algorithm~\ref{alg:computingmaximum}.

\medskip

\noindent{\textbf{Statement of Lemma~\ref{lem:algrunning}.}
    Let $ (\mathcal{A}, \le) $ be a Wheeler NFA, with $ \mathcal{A} = (Q, E, s, F) $. Algorithm~\ref{alg:computingmaximum} can be implemented to run in $ O(|E|) $ time on input $ (\mathcal{A}, \le) $.

\begin{proof}
     First, note that every $ 2 \le i \le |Q| $ is added to the queue $ K $ at most once (in particular, the queue $ K $ never contains more than $ |Q| - 1 $ elements). Indeed, for every $ 1 \le i \le |Q| $, (i)~an element can only be added to the queue in Line~\ref{line:next1} or Line~\ref{line:next2} or Line~\ref{line:next2bis}; (ii)~immediately after $ i $ is added to $ K $ a first time, the algorithm sets $ B[i] \gets 1 $ (in Line~\ref{line:firstB1} or Line~\ref{line:firstB2} or Line~\ref{line:firstB2bis}), and $ B[i] = 1 $ holds until the end of the algorithm (see Remark \ref{rem:basicalgorithm}); (iii)~the element $ i $ can be added to $ K $ during at most one execution of Line~\ref{line:next1}; (iv)~the element $ i $ can be added to $ K $ during an execution of Line~\ref{line:next2} or in Line~\ref{line:next2bis} only if it has never been added to $ K $ before, because the element $ i $ can be added to $ K $ during an execution of Line~\ref{line:next2} or in Line~\ref{line:next2bis} only if $ B[i] = 0 $ immediately before (Line~\ref{line:recursivecheck} or Line~\ref{line:recursivecheckbis}).
    
    We first show that the initialization of our algorithm (Lines~\ref{line:initializationstart}-\ref{line:initializationend}) only requires $ O(|E|) $ time.
    
    \begin{itemize}
        \item Let us show that we can compute $ \ell_i^{\min} $ for every $ 1 \le i \le |Q| $ in $ O(|E|) $ time (one can analogously compute $ \ell_i^{\max} $ for every $ 1 \le i \le |Q| $ in $ O(|E|) $ time). To this end, we will also compute $ a_i^{\min} $ for every $ 1 \le i \le |Q| $. For every $ 1 \le i \le n $, let $ a^*_i $ and $ \ell^*_i $ be two variables. We process all edges in $ E $ (in any order), maintaining the following invariant: immediately after processing a set $ E' \subseteq E $ of edges, for every $ 1 \le i \le |Q| $ the variables $ a^*_i $ and $ \ell^*_i $ are equal to the values $ a_i^{\min} $ and $ \ell^*_i $ computed by only considering the edges in $ E' $. Formally, for every $ 1 \le i \le |Q| $, if there exists no edge in $ E' $ reaching $ u_i $, we have $ a^*_i = \bot $ and $ \ell^*_i = \bot $, and if there exists at least one edge in $ E' $ reaching $ u_i $, then $ a_i^* $ is the smallest character labeling some edge in $ E' $ reaching $ u_i $, and $ \ell^*_i $ is the smallest $ 1 \le j \le |Q| $ such that $ (u_j, u_i, a_i^*) \in E' $.

        Before processing the edges in $ E $, we set $ a^*_i \gets \bot $ and $ a^*_i \gets \bot $ for every $ 1 \le i \le |Q| $. This ensures that the invariant is true at the beginning. Now assume inductively that the invariant is true immediately after processing a set $ E' \subseteq E $ of edges, and let $ (u_i, u_j, a) $ be the next edge in $ E $ to be processed. If $ a^*_j = \bot $, we set $ a^*_j \gets a $ and $ \ell^*_j \gets i $. If $ a^*_j \not = \bot $, then by the invariant we have $ a^*_j \in \Sigma $, $ l^*_j \not = \bot $ and $ 1 \le \ell^*_j \le |Q| $. If $ a^*_j \prec a $, we do not do anything. If $ a \prec a^*_j $, we set $ a^*_j \gets a $ and $ \ell^*_j \gets i $. Now assume that $ a = a^*_j $. If $ \ell^*_j \le i $, we do not do anything. If $ i < \ell^*_j $, we set $ \ell^*_j \gets i $. Then, the invariant is also true immediately after processing $ E' \cup \{(u_i, u_j, a) \} $.

        In particular, after processing all edges in $ E $, we have $ E' = E $ and so by the invariant we conclude $ a^*_i = a_i^{\min} $ and $ \ell^*_i = \ell_i^{\min} $ for every $ 1 \le i \le |Q| $. The total required time is $ O(|E| + |Q|) \subseteq O(|E|) $ (see Remark~\ref{rem:complexitystates}) because we only need $ O(1) $ time per edge and per state.

        \item Let us show that we can compute $ Z[2, |Q|] $ in $ O(|E|) $ time. We sort all edges $ (u_i, u_j, a) $ by the key $ (i, a, j) $. This can be done in $ O(|E| + |Q|) \subseteq O(|E|) $ time via radix sort~\cite{hopcroft1983data} because the edge labels can be sorted in linear time. By scanning the list of all sorted edges and advancing two pointers, in $ O(|E|) $ time we can decide whether $ Out(i - 1) \not = Out(i) $ (and so we can compute $ Z[i] $) for every $ 2 \le i \le |Q| $. 

        \item The initialization of $ B[2, |Q|] $ requires $ O(|Q|) \subseteq O(|E|) $ time, and we can initialize the queue $ K $ in $ O(|Q|) \subseteq O(|E|) $ time because the queue never needs to contain more than $ |Q| - 1 $ elements.
    \end{itemize}
    
We are only left with showing that the main part of our algorithm (Lines~\ref{line:mainstart}-\ref{line:mainend}) only requires $ O(|E|) $ time. Lines~\ref{line:mainstart}-\ref{line:mainintermidate1} require $ O(|Q|) \subseteq O(|E|) $ time. Lines~\ref{line:mainintermediate2}-\ref{line:mainendbefore} also require $ O(|Q|) \subseteq O(|E|) $ time because we have already shown that each $ 2 \le i \le |Q| $ is added to the queue at most once. \qedhere

\end{proof}

We can finally prove that the quotient Wheeler NFA can be built in linear time.

\medskip

\noindent{\textbf{Statement of Theorem~\ref{theor:lineartimeminimization}.}
    Let $ (\mathcal{A}, \le) $ be a Wheeler NFA, with $ \mathcal{A} = (Q, E, s, F) $. We can build the Wheeler NFA $ (\mathcal{A}_*, \le_*) $ of Lemma~\ref{lem:quotient} in $ O(|E|) $ time.

\begin{proof}
    By Lemma~\ref{lem:algcorrectness} and Lemma~\ref{lem:algrunning}, we can compute $ B_{\mathcal{A}, \le}[2, |Q|] $ in $ O(|E|) $ time. Let $ \mathcal{A}_* = (Q_*, E_*, s_*, F_*) $. By definition, $ Q_* $ is the set of all $ \equiv_{\mathcal{A}, \le} $-equivalence classes. By Lemma~\ref{lem:B_Aisallweneed}, in $ O(|Q|) $ time, we can associate each $ u \in Q $ with its $ \equiv_{\mathcal{A}, \le} $-equivalence class by scanning all states in $ Q $ following the Wheeler order $ \le $ and checking the values of $ B_{\mathcal{A}, \le}[2, |Q|] $. In particular, in this way we determine $ |Q_*| $. By definition, $ s_* $ is the $ \equiv_{\mathcal{A}, \le} $-equivalence class of $ s $. By the definition of $ F_* $, we can determine $ F_* $ in $ O(|Q|) $ time by scanning all $ u \in Q $ and adding $ [u]_{\equiv_{\mathcal{A}, \le}} $ to $ F_* $ if $ u \in F $ (note that in this way some states may be added to $ F_* $ multiple times, but we can remove duplicates by sorting all states in $ F_* $ in $ O(|F_*|) \subseteq O(|Q|) $ time). We are only left with showing how to determine $ E_* $. By the definition of $ E_* $, we can determine $ E_* $ in $ O(|E|) $ time by scanning all edges $ (u, v, a) \in E $ and adding $ ([u]_{\equiv_{\mathcal{A}, \le}}, [v]_{\equiv_{\mathcal{A}, \le}}, a) $ to $ E_* $ (note that in this way some edges may be added to $ E_* $ multiple times, but we can remove duplicates by sorting all edges in $ E_* $ in $ O(|E_*|) \subseteq O(|E|) $ time via radix sort, as we did in the proof of Lemma~\ref{lem:algrunning}). The total running time is $ O(|E| + |Q|) \subseteq O(|E|) $ (see Remark~\ref{rem:complexitystates}). \qedhere
\end{proof}

We are now ready to prove Corollary~\ref{cor:checkingifbisimlar}.

\medskip

\noindent{\textbf{Statement of Corollary~\ref{cor:checkingifbisimlar}.}
     Let $ (\mathcal{A}, \le) $ and $ (\mathcal{A}', \le') $ be  Wheeler NFAs, with $ \mathcal{A} = (Q, E, s, F) $ and $ \mathcal{A}' = (Q', E', s', F') $. We can determine whether there exists a Wheeler bisimulation from $ (\mathcal{A}, \le) $ to $ (\mathcal{A}', \le') $ in $ O(|E| + |E|') $ time.

\begin{proof}
    Let $ (\mathcal{A}_*, \le_*) $ and $ (\mathcal{A}'_*, \le'_*) $ be the Wheeler NFAs built in Lemma~\ref{lem:quotient} starting from $ (\mathcal{A}, \le) $ and $ (\mathcal{A}', \le') $, respectively. Let $ \mathcal{A}_* = (Q_*, E_*, s_*, F_*) $ and $ \mathcal{A}'_* = (Q'_*, E'_*, s'_*, F'_*) $, with $ Q_* = \{U_1, U_2, \dots, U_n \} $ and $ Q'_* = \{U'_1, U'_2, \dots, U'_{n'} \} $, where $ U_1 <_* U_2 <_* \dots <_* U_n $ and $ U'_1 <'_* U'_2 <'_* \dots <'_* U'_{n'} $. By Theorem~\ref{theor:lineartimeminimization}, we can build $ (\mathcal{A}_*, \le_*) $ and $ (\mathcal{A}'_*, \le'_*) $ in $ O(|E|) $ time and $ O(|E'|) $ time, respectively.

    If $ n \not = n' $, then by Lemma~\ref{lem:quotientautomotaisomorphic} there exists no Wheeler bisimulation from $ (\mathcal{A}, \le) $ to $ (\mathcal{A}', \le') $, so in the rest of the proof we can assume $ n = n' $. Let $ R $ be the function from $ Q_* $ to $ Q'_* $ such that $ R(U_i) = U'_i $ for every $ 1 \le i \le n $. If $ R $ is not an isomorphism from $ \mathcal{A}_* $ to $ \mathcal{A}'_* $, then by Lemma~\ref{lem:quotientautomotaisomorphic} there exists no Wheeler bisimulation from $ (\mathcal{A}, \le) $ to $ (\mathcal{A}', \le') $. If $ R $ is an isomorphism from $ \mathcal{A}_* $ to $ \mathcal{A}'_* $, then it is immediate to check that $ R $ is a Wheeler bisimulation from $ (\mathcal{A}_*, \le_*) $ to $ (\mathcal{A}'_*, \le'_*) $, so $ R_{\mathcal{A}', \le'}^{-1} \circ R \circ R_{\mathcal{A}, \le} $  is a Wheeler bisimulation from $ (\mathcal{A}, \le) $ to $ (\mathcal{A}', \le') $ by Lemma~\ref{lem:bisimulationtoquotient} and Lemma~\ref{lem:bisimequiv}. Consequently, to finish the proof, we only need to show that in $ O(|E| + |E|') $ time we can check whether $ R $ is an isomorphism from $ \mathcal{A}_* $ to $ \mathcal{A}'_* $.

    For every $ U \in F_* $, we check whether $ R(U) \in F'_* $, and for every $ (U, V, a) \in E_* $, we check whether $ (R(U), R(V), a) \in E'_* $. For every $ U' \in F'_* $, we check whether $ R^{-1}(U') \in F_* $, and for every $ (U', V', a) \in E'_* $, we check whether $ (R^{-1}(U), R^{-1}(V), a) \in E_* $. Then, $ R $ is an isomorphism from $ \mathcal{A}_* $ to $ \mathcal{A}'_* $ if and only all checks are successful. The time required for all the checks is $ O(|F_*| + |E_*| + |F'_*| + |E'_*|) \subseteq O(|Q| + |E| + |Q'| + |E'|) \subseteq O(|E| + |E|') $ (see Remark~\ref{rem:complexitystates}). \qedhere
\end{proof}

We are only left with proving our lower bound in the comparison model (Theorem~\ref{theor:lowerboundquotienting}).

\medskip

\noindent{\textbf{Statement of Theorem~\ref{theor:lowerboundquotienting}.}
    Let $ (\mathcal{A}, \le) $ be a Wheeler NFA, with $ \mathcal{A} = (Q, E, s, F) $. In the comparison model, the time complexity of building the Wheeler NFA $ (\mathcal{A}_*, \le_*) $ of Lemma~\ref{lem:quotient} is $ \Theta(|E| \log |E|) $. This is true even if we know that $ \mathcal{A} $ is a DFA.

\begin{proof}
    We obtain an $ O(|E| \log |E|) $ algorithm as follows. First, we sort all edge labels in $ O(|E| \log |E|) $ time via any comparison-based sorting algorithm (e,g, merge sort~\cite{cormen2022introduction}). Let $ \mathcal{A}' $ be the NFA obtained from $ \mathcal{A} $ by replacing every edge label with its position in the list of all sorted edge labels (more precisely, if $ \mathcal{A} $ contains $ d $ distinct edge labels, then every edge label of $ \mathcal{A}' $ is an integer between $ 0 $ and $ d - 1 $). Then, $ (\mathcal{A}', \le) $ is a Wheeler NFA. Let $ (\mathcal{A}'_*, \le'_*) $ be the Wheeler NFA built in Lemma~\ref{lem:quotient} starting from $ (\mathcal{A}', \le) $. Then, (i)~$ \mathcal{A}'_* $ is obtained from $ \mathcal{A}' $ by replacing every edge label with its position in the list of all sorted edge labels, and (ii)~$ \le'_* $ and $ \le_* $ are the same total order. We first build $ (\mathcal{A}', \le) $ from $ (\mathcal{A}_*, \le_*) $ in $ O(|E|) $ time. Then by Theorem~\ref{theor:lineartimeminimization} we can build $ (\mathcal{A}'_*, \le'_*) $ from $ (\mathcal{A}', \le) $ in $ O(|E|) $ time (we can apply Theorem~\ref{theor:lineartimeminimization} because every edge label of $ \mathcal{A}' $ is an integer smaller than $ |E| $, so the edge labels can be sorted in $ O(|E|) $ time via counting sort~\cite{cormen2022introduction}). Lastly, we build $ (\mathcal{A}_*, \le_*) $ from $ (\mathcal{A}'_*, \le'_*) $ in $ O(|E|) $ time.
    
    We are left with showing that the time complexity of building $ (\mathcal{A}_*, \le_*) $ is $ \Omega(|E| \log |E|) $, even if we know that $ \mathcal{A} $ is a DFA. Let $ f $ be a function for which there exists an algorithm building $ (\mathcal{A}_*, \le_*) $ in at most $ f(|E|) $ time for every Wheeler DFA $ (\mathcal{A}, \le) $. We need to prove that $ f(|E|) \in \Omega(|E| \log |E|) $. To this end, consider the \emph{set equivalence} problem: given $ n \ge 1 $ and given $ b_1, b_2, \dots b_n $ and $ c_1, c_2, \dots, c_n $ such that the $ b_i $'s are pairwise distinct and the $ c_i $'s are pairwise distinct, decide whether $ \{b_1, b_2, \dots, b_n \} = \{c_1, c_2, \dots, c_n \} $. In the comparison model, the complexity of the set equivalence problem is $ \Omega (n \log n) $~\cite[Theorem 2]{reingold1972optimality}. Consequently, we only need to show that we can solve the set equivalence problem in $ f(2n + 2) + O(n) $ time, because then the $ \Omega (n \log n) $ lower bound for set equivalence problem implies $ f(n) \in \Omega (n \log n) $.

    Let $ b_1, b_2, \dots b_n, c_1, c_2, \dots, c_n $ be an instance of the set equivalence problem. We need to show that, in $ f(2n + 2) + O(n) $ time, we can decide whether $ \{b_1, b_2, \dots, b_n \} = \{c_1, c_2, \dots, c_n \} $. Let $ \#_1 $ and $ \#_2 $ be new characters that are not in $ \Sigma $, and assume that $ \#_1 \prec \#_2 \prec a $ for every $ a \in \Sigma $. In particular, since in the comparison model we can compare two characters in $ \Sigma $ in $ O(1) $ time, we can also compare two characters in $ \Sigma \cup \{\#_1, \#_2 \} $ in $ O(1) $ time. Let $ \mathcal{A} = (Q, E, s, F) $ be the DFA defined as follows: $ Q = \{z_1, z_2, s, t \} $ consists of four pairwise distinct states, $ E = \{(s, z_1, \#_1), (s, z_2, \#_2) \} \cup \{(z_1, t, b_i) \;|\; 1 \le i \le n \} \cup \{(z_2, t, c_i) \;|\; 1 \le i \le n \} $, and $ F = \{t \} $ (see Figure~\ref{fig:comparison} for an example). Note that $ |E| = 2n + 2 $ and $ \mathcal{A} $ is a DFA (because the $ b_i $'s are pairwise distinct and the $ c_i $'s are pairwise distinct). Let $ \le $ be the total order on $ Q $ such that $ s < z_1 < z_2 < t $. It is immediate to check that $ (\mathcal{A}, \le) $ is a Wheeler DFA. We can build $ (\mathcal{A}, \le) $ in $ O(n) $ time, and then we can build the Wheeler DFA $ (\mathcal{A}_*, \le_*) $ of Lemma~\ref{lem:quotient} in $ f(2n + 2) $ time.
    
    Notice that for every distinct $ u, u' \in Q $, if $ \{u, u' \} \not = \{z_1, z_2 \} $, then $ u \not \equiv_{\mathcal{A}, \le} u' $ (because $ \equiv_{\mathcal{A}, \le} $ is a bisimulation from $ \mathcal{A} $ to $ \mathcal{A} $ and the set of all characters labeling some edge leaving $ u $ is distinct from the set of all characters labeling some edge leaving $ u' $). Moreover, since $ z_1 $ and $ z_2 $ are consecutive in the Wheeler order $ \le $, we have $ z_1 \equiv_{\mathcal{A}, \le} z_2 $ if and only if $ \{b_1, b_2, \dots, b_n \} = \{c_1, c_2, \dots, c_n \} $. Consequently, we have $ \{b_1, b_2, \dots, b_n \} = \{c_1, c_2, \dots, c_n \} $ if and only if the number of states of $ \mathcal{A}_* $ is three. The total time required to check whether $ \{b_1, b_2, \dots, b_n \} = \{c_1, c_2, \dots, c_n \} $ is $ f(2n + 2) + O(n) $. \qedhere
\end{proof}

\begin{figure}
        \centering
        \scalebox{.8}{
        \begin{tikzpicture}[->,>=stealth', semithick, auto, scale=1]
\node[state] (1)    at (0,0)	{$ s $};
\node[state] (2)    at (-3, 2)	{$ z_1 $};
\node[state] (3)    at (3, 2)	{$ z_2 $};
\node[state, accepting] (4)    at (0, 4)	{$ t $};
\draw (1) edge [] node [] {$ \#_1 $} (2);
\draw (1) edge [] node [] {$ \#_2 $} (3);
\draw (2) edge [] node [] {$ a, d, b, c $} (4);
\draw (3) edge [] node [] {$ c, b, a, d $} (4);
\end{tikzpicture}
}
\caption{The Wheeler DFA built in the proof of Theorem~\ref{theor:lowerboundquotienting} starting from the instance $ b_1 = a $, $ b_2 = d $, $ b_3 = b $, $ b_4 = c $, $ c_1 = c $, $ c_2 = b $, $ c_3 = a $, $ c_4 = d $.}
\label{fig:comparison}
\end{figure} 

\end{document}